\documentclass[aos,preprint]{imsart}
\RequirePackage[OT1]{fontenc}
\RequirePackage{amsthm,amsmath}
\RequirePackage{natbib}
\RequirePackage[colorlinks,citecolor=blue,urlcolor=blue]{hyperref}

\arxiv{arXiv:0000.0000}

\startlocaldefs
\numberwithin{equation}{section}
\theoremstyle{plain}

\endlocaldefs

\usepackage[final]{graphics}
\usepackage{graphicx}
\usepackage{amssymb}
\usepackage{amsmath,amsthm}
\usepackage{lipsum}
\usepackage{wasysym}

\usepackage{amsmath}
\usepackage{amsbsy}
\usepackage{enumerate}
\usepackage{bm}
\usepackage{color}

\usepackage{multirow}

\newtheorem{theorem}{Theorem}
\newtheorem{lemma}{Lemma}

\newtheorem{definition}{Definition}
\newtheorem{corollary}{Corollary}

\theoremstyle{remark}
\newtheorem{remark}{Remark}

\usepackage{array}
\usepackage{arydshln}

\newcommand{\pr}{\text{pr}}
\newcommand{\var}{\textnormal{var}}
\newcommand{\cov}{\textnormal{cov}}

\newcommand{\obs}{{\mbox{\scriptsize obs}}}

\newcommand{\T}{\textsc{t}}
\newcommand{\LME}{\textsc{lme}}
\newcommand{\LM}{\textsc{lm}}

\newcommand{\SP}{\textsc{s-p}}
\newcommand{\CR}{\textsc{c-r}}

\newcommand{\E}{E}

\newcommand{\vones}{\mathbf{1}}
\newcommand{\vzeros}{\mathbf{0}}

\newcommand{\vg}{\bm{g}}
\newcommand{\vY}{\bm Y}

\newcommand{\mY}{\mathbf{Y}}
\newcommand{\mD}{\mathbf{D}}

\newcommand{\mI}{\mathbf{I}}
\newcommand{\mIw}{\mathbf{I}_{W}}
\newcommand{\mIs}{\mathbf{I}_{M}}

\newcommand{\mJ}{\mathbf{J}}

\newcommand{\mJs}{\mJ_{M}}

\newcommand{\mJn}{\mJ_{N}}

\newcommand{\mP}{\mathbf{P}}
\newcommand{\mPs}{\mP_{M}}
\newcommand{\mPw}{\mP_{W}}

\newcommand{\mC}{\mathbf{C}}
\newcommand{\mCf}{\mC}
\newcommand{\mCi}{\mC_{\textnormal{in}}}
\newcommand{\mCb}{\mC_{\textnormal{btw}}}

\newcommand{\mHb}{\mathbf{H}_{\textnormal{btw}}}

\newcommand{\mSi}{\mathbf{S}^2_{\textnormal{in}}}
\newcommand{\mSb}{\mathbf{S}^2_{\textnormal{btw}}}
\newcommand{\mS}{\mathbf{S}^2}
\newcommand{\Si}{S^2_{\textnormal{in}}}
\newcommand{\Sb}{S^2_{\textnormal{btw}}}

\newcommand{\mPb}{\mP_{\textnormal{btw}}}
\newcommand{\mPi}{\mP_{\textnormal{in}}}

\newcommand*{\mybox}[1]{\framebox{#1}}

\newcommand{\Yt}{\widetilde{\mY}}

\newcommand{\vZ}{{\bm Z}}
\begin{document}
%%%%%%%%
%\noindent\textbf{To-Do list}
%\begin{enumerate}
%\item check the note in mac for word choice
%\item this version from 0110, with the main text micro-adjusted, updating the proofs only.
%\item check the compatibility of the proof ... redefine $A_w$, $B_{(wm)}$ when necessary!!!!!!
%\end{enumerate}
%
%%%%%%%%%%%%%%%%%%%%%
%%
%\noindent\textbf{Notation Protocol}
%%
%%%%%%%%%%%%%%%%%%%%%
%\begin{enumerate}
%\item from between to within as always as possible! --- reads betetr ... between- and within-block
%\item running index: no parenthesis; double index: parenthesis in the subscript
%\item Treatment index: in parenthesis like `(k)' as much as possible, e.g., $g_A(k)$, with the exception of $N_{k(w)}$, to be consistent with $N_k$ ...
%\end{enumerate}

\newpage
%%%%%%%%%%%%%%%%%%%%%
%%%%%%%%%%%%%%%%%%%%%%
\begin{frontmatter}
\title{Randomization-Based Causal Inference\\ from Unbalanced $2^2$ Split-Plot Designs\thanksref{T1}}
\runtitle{Neymanian Causal Inference for  $2^2$ Split-Plot Designs}
\thankstext{T1}{Special thanks go to Professor Richard Tuck and Professor Joseph Blitzstein, for all the seemingly irrelevant, yet profoundly affecting, inspirations that had transformed this manuscript; and to Steven Finch (Harvard), for being our first reader, as not only a sharp fellow statistician with many insightful comments, but also a meticulous English teacher.}

\begin{aug}
\author{\fnms{Anqi} \snm{Zhao}\thanksref{m1}\ead[label=e1]{anqizhao@fas.harvard.edu}},
\author{\fnms{Peng} \snm{Ding}\thanksref{m2}\ead[label=e2]{pengdingpku@gmail.com}}
\and
\author{\fnms{Tirthankar} \snm{Dasgupta}\thanksref{m1}
\ead[label=e3]{dasgupta@stat.harvard.edu}
\ead[label=u1,url]{http://http://statistics.fas.harvard.edu/people/tirthankar-dasgupta}}

%\thankstext{m1}{Some comment}
\runauthor{A. Zhao et al.}

\affiliation{Harvard University\thanksmark{m1} and University of California at Berkeley\thanksmark{m2}}

\address{Department of Statistics\\
Harvard University\\
Science Center, 1 Oxford Street\\
Cambridge, MA\\
\printead{e1}\\
\phantom{E-mail:\ }\printead*{e3}}

\address{Department of Statistics\\
University of California, Berkeley\\
Evans Hall\\
Berkeley, CA\\
\printead{e2}}
\end{aug}

\begin{abstract}
Given two 2-level factors of interest,
a \emph{$2^2$ split-plot design}
(a) takes each of the $2^2=4$ possible factorial combinations as a treatment,
(b) identifies one factor as `whole-plot,'
(c) divides the experimental units into blocks, and
(d) assigns the treatments in such a way that all units within the same block receive the same level of the whole-plot factor.
%%%%%%%%%%%%%
Assuming the potential outcomes framework, we propose in this paper a randomization-based estimation procedure for  causal inference from $2^2$ designs that are not necessarily balanced. Sampling variances of the point estimates are derived in closed form as linear combinations of the between- and within-block covariances of the potential outcomes. Results are compared to those under complete randomization as measures of design efficiency. Interval estimates are constructed based on conservative estimates of the sampling variances, and the frequency coverage properties evaluated via simulation. Asymptotic connections of the proposed approach to the model-based super-population inference are also established. Superiority over existing model-based alternatives is reported under a variety of settings for both binary and continuous outcomes.
\end{abstract}

\begin{keyword}[class=MSC]
\kwd[Primary ]{62K15}
\kwd{62K10}
\kwd[; secondary ]{62K05}
\end{keyword}

\begin{keyword}
\kwd{Between-block additivity}
\kwd{Model-based inference}
\kwd{Neymanian inference}
\kwd{Potential outcomes framework}
\kwd{Projection matrix}
\kwd{Within-block additivity}
%\kwd{Randomization-based inference}
\end{keyword}

\end{frontmatter}
%%%%%%%%%%%%%%%%%%%%%%%%%%%%%
% ========================================
\section{Introduction}
% ========================================
%%%%%%%%%%%%%%%%%%%%%%%%%%%%%%

%%%%%%%%%%%%%%%%%%%%%%%%%%%%%%%%%%
% ===================================
% 第一段 FACTORIAL EXPERIMENTS AND SPLIT-PLOT DESIGNS
\subsection{Split-plot designs for factorial experiments}
%%%%%%%%%%%%%%%%%%%%%%%%%%%
Factorial experiments, originally developed in the context of agricultural experiments \citep{Fisher1925,Fisher1935,Yates1937}
and later extensively used in industrial and engineering applications,
are nowadays undergoing a third popularity surge among social, behavioral, and biomedical sciences,
as a result of the massive trend in these areas to generalize the previous treatment-control experiments to
include multiple factors.
%The need to randomize multiple factors at the same time leads to the introduction of a variety of factorial designs, among which
%split-plot design has always remained a popular choice
%multi-factor settings.
%Split-plot design, as a multi-factor assignment mechanism that admits more flexibility than the traditional completely randomized design,
%especially when practical difficulties like hard-to-change factor or economic constraints renders the latter practically infeasible {\citep{Jones2009}}.
Among the plethora of possible multi-factor randomization schemes available,
split-plot design, thanks to its flexibility and ease of application,
%have always remained at the top of users' preference lists,
%has always been a popular choice,
has always remained a popular choice,
%has always remained at the top of users' preferences,
especially when practical difficulties
%such as
like economic constraints or hard-to-change factor
preclude %the feasibility of
%render
% the use of
the use of simple, unrestricted randomizations % alternatives% infeasible
%completely randomized option unavailable
{\citep{Jones2009}}.
%Among the plethora of multi-factor assignment mechanisms available,
%split-plot design, thanks to its flexibility and ease of application,
%%have always remained at the top of users' preference lists,
%has always remained users' top choice,
%especially when practical difficulties like hard-to-change factor or economic constraints preclude the use of complete randomization {\citep{Jones2009}}.
%%%%%%%%%%%%%%%%%%%%%%%%%%%%%%%
%%%%%%%%%%%%%%%%%%%%%%%%%%%%%%%
As a motivating example, consider a simplified version of the education experiment described in \cite{Tir2015}.
The goal is to evaluate the efficacies of two interventions --- $A$: a mid-year quality review by a team of experts, and $B$: a bonus scheme to teachers --- on 224 schools in the state of New York.
Assume two possible actions for each intervention --- application or non-application,
%the design question is to decide the subsets of schools to which each intervention should be applied.
%Take the   as four possible treatments,
a complete randomization of the four combinations likely scatters the schools to be reviewed throughout the state.
Given the travel and time cost this may incur,
%a completely randomized assignment tends to scatter the schools to be reviewed throughout the state, thereby providing the experts' with a rather challenging travel itinerary.
%, thereby giving the experts' a very long trip.
%%%%%%%%%%%%%%%%%%%%%%%%%%%%%%%%%%%
%%%%%%%%%%%%%%%%%%%%%%%%%%%%%%%%%%%
a more practical alternative would be to divide the 224 schools by geographic proximity into sixteen `blocks,' choose eight at random, and conduct expert quality review for all schools therein. The teacher bonus scheme can then be applied to half of the schools within each block.
%%%%%%%%%%%
%a more practical alternative would be to
%(1)\,divide the 224 schools into sixteen `blocks' according to geographic proximity,
%(2)\,choose at random eight blocks and conduct the quality review for all schools therein,
%(3)\,choose at random half of the schools from each block to implement the bonus scheme.
This exemplifies split-plot design.
See \cite{BHH2005}, \cite{CC1957},  and \cite{WandH2009} for formal definitions.

%%%%%%%%%%%%%%%%%%%%%%%%
% ===================================
% 第二段 RANDOMIZATION BASED APPROACH TO ANALYZING SPLIT-PLOT DESIGNS.
\subsection{Randomization-based approach to analyzing split-plot designs}
%%%%%%%%%%%%%%%%%%%%%%%%%%%%%%%%%%%
Most factorial experiments, like any experiment, receive regression-based methods as their default `treatment.'
For those under split-plot designs,
this default is either the analysis of variance (\textsc{anova}) or the linear mixed effects model \citep{WandH2009}.
Despite the good intention of both methods to adjust for the block structure that defines split-plot designs,
the actual variance estimation often turns out inconsistent \citep{Gelman2005,Hinkelmann2008},
likely due to the required model assumptions not being satisfied. %met.
A detailed examination of this argument can be found in \cite{Freedman2006,Freedman2008a}, which recommended randomized-based inference as the proper solution.
%This is not too much of a surprise, given most of the model assumptions are not justified by the design.

%;
%most discussion, although made in the context of treatment-control experiments, applies to factorial experiments as well.

Despite its long tradition  in the context of treatment-control experiments \citep{PT2015},
randomization-based inference remains an almost uncharted field when it comes to factorial experiments.
The recent works of \cite{Tir2015} and \cite{Espinosa2015} are, to the best of our knowledge, the only literature along this line,
each documenting improvements of randomization-based analysis over existing model-based methods
in the context of multi-factor completely randomized designs.
Generalizing their methods to split-plot designs could be a promising next step.
%breaking the current absolute, yet questionable,
%predominance of regression models in analyzing experimental data.
%As \cite{Freedman2006} put it,
%`Generalizing from the experimental subjects to a broader population --- \emph{external validity} --- is a major concern, but beyond the scope of this article.'

%%%%%%%%%%%%%%%%%%%%%%%%
% ===================================
% 第三段 CONTRIBUTION
\subsection{Contributions}
%%%%%%%%%%%%%%%%%%%%%%%%
The contribution of this paper is three-fold.
First,
we develop the first randomization-based estimation procedure for causal inference under $2^2$ split-plot designs,
and demonstrate its superior frequency coverage properties over existing alternatives.
%%%%%%%%%%%%%%%%%%%%%
%%%%%%%%%%%%%%%%%
%%%%%%%%%%%%%
Second,
motivated by split-plot designs' signature block structure,
we propose a decomposition of the potential outcomes that
links the relative efficiency between a split-plot design and a complete randomization of the same size to the level of heterogeneity among blocks.
%%The
%propose a decomposition of the potential outcomes that relates the relative efficiency of split-plot design --- when compared to the complete randomization of the same size --- to the presence of random block effects from a super-population perspective,
%%rendering knowledge about the latter, if available, useful guidance for the choice of design.
%%{\color{red}we
%%introduce a decomposition of the variations among potential outcomes that relates the relative estimation precision of split-plot design
%%how
%%the presence of block effects affects
%%the relative efficiency between a split-plot design and
%%its completely randomized counterpart
%%can be explained in terms of the presence of block effects,
%%an intuitive explanation
%%that helps relate the relative efficiency of split-plot designs
%% --- when compared to complete randomizations of the same size ---
%%the relation between the block effect and split-plot designs' relative efficiency.
%%when
%%compared to
%%how the relative efficiency of a split-plot design, when compared to the complete randomization of the same size, is affected by the presence of block effects,
%%in terms of
%%to the presence of random block effects from a super-population perspective,
This allows any empirical knowledge about the latter, when available, to be admitted as possible aid for deciding between designs.
%%%%%%%%%%%%%%%%%%%%
%%%%%%%%%%%%%%%%%
%%%%%%%%%%%%%

Third, in an attempt to reconcile the finite-population randomization-based perspective and a hypothetical super-population model-based perspective, we offer a heuristic argument that connects the two. This connection is established by using the asymptotics of the finite-population randomization-based residual covariances to justify the block-diagonal structure assumed by the linear mixed effects model for the covariances of its super-population sampling errors.
This, to the best of our knowledge, is the very first attempt that aims at reconciling the difference between finite and super-population inferences.
 %%%%%%%%%%%%%%%%%%%%%%%
% ===================================
% 第四段 ORGANIZATION
\subsection{Organization of the article}
%%%%%%%%%%%%%%%%%%%%%%%
The article is organized as follows.
We review in Section \ref{section::POreview} the potential outcomes framework, and discuss possible extensions when the experimental units exhibit certain block structure.
We define in Section \ref{section::causalEffect} the causal questions in $2^2$ factorial experiments, and introduce in Section \ref{section::AM} the split-plot design as one possible randomization scheme. %assignment mechanism.
% available for such experiments.
%Sampling variances of the estimates are derived in Section \ref{section::var}, and estimation of these variances is addressed in Section \ref{section::varHat}.
%Sampling variances of the estimates are derived in Section \ref{section::var}, and estimation of these variances is addressed in Section \ref{section::varHat}.
Sampling variances of the estimates are derived in Section \ref{section::var}, and their estimation addressed in Section \ref{section::varHat}.
%Section \ref{section::model} then discusses the connection and distinction between the model-based and randomization-based inferences, and demonstrate via simulation the latter's superior frequency coverage properties in Section \ref{section::simulation}.
%The theoretical connection and distinction between the randomization-based and model-based inferences are then discussed in Section \ref{section::model},  and the superiority of the proposed randomization-based interval estimates over existing model-based alternatives demonstrated in Section \ref{section::simulation}.
We discuss the
%theoretical
connection and distinction between the model-based and randomization-based inferences in Section \ref{section::model}, and demonstrate the latter's superior frequency coverage properties in Section \ref{section::simulation}.
% the superiority of the proposed randomization-based interval estimates over existing model-based alternatives .
We conclude in Section \ref{section::discussion}.
All proofs are deferred to the online supplementary material.

%%%%%%%%%%%%%%%%%%%%%%%%%%%%%%%%%%%%%%%%%%%
% ============================================================
\section{Potential outcomes and additivity assumptions\label{section::POreview}}
% =============================================================
%%%%%%%%%%%%%%%%%%%%%%%%%%%%%%%%%%%%%%%%%%%
We review in this section the major concepts within the \emph{potential outcomes framework} \citep{Neyman1923,Rubin1974,Rubin1978,Rubin2005}, and discuss some possible extensions when the experimental units are nested under blocks. %exhibit certain block structure.
%%%%%%%%%%%%%%%%%%%%%%%%%%
%%%%%%%%%%%%%%%%%%%%%%%
\subsection{Potential outcomes framework for causal inference}
%%%%%%%%%%%%%%%%%
Consider an experiment in which $K$ different \emph{treatment}s are to be tested on $N$ experimental \emph{unit}s.
The Stable Unit Treatment Value Assumption \citep{Rubin1980} allows us to write the \emph{potential outcome} of unit $i$ when exposed to treatment $k$ as $Y_i(k)$.
%and to define the causal effects as the difference in potential outcomes between different treatments.
Whereas causal effects are then defined as comparisons of such potential outcomes for a given set of units,
any experiment,  however well designed and implemented, allows us to observe at most one of $K$ potential outcomes per unit, according to the treatment it receives.
This poses the \emph{fundamental problem of causal inference}\,\citep{Holland1986}.
%Causal effects are then defined as comparisons of such potential outcomes for some given set of units.
%That any experiment, however well designed and implemented, allows us to observe but one of the $K$ potential outcomes per unit, according to the treatment it receives.
%with the \emph{fundamental problem of causal inference}\,\citep{Holland1986} %is thus posed by the fact
%being that
%As a result,
Various assumptions are introduced in this context
as attempts to infer the unobserved from the observed,
the most common being that of the \emph{strict additivity}.
%%The nature of this problem makes its reformulation as a missing data imputation problem quite straight-forward --- simply treating the unobservables as missing data will do the trick.
%%The \emph{strict additivity assumption}, among other possibilities,
%%among other restrictions on how $Y_i(k)$ differ from one another,
%%This \emph{fundamental problem of causal inference} \citep{Holland1986}
%and motivates the introduction of various assumptions as attempts to infer the unobserved from the observed.
%
%Various assumptions are introduced to address  the \emph{fundamental problem of causal inference}\,\citep{Holland1986} that any experiment, however well designed and implemented, allows us to observe but one of the $K$ potential outcomes per unit, according to the treatment it receives.

\begin{definition}
\label{def::strictAdd}
The potential outcomes $Y_i(k)$\!  of $N$\! units under $K$\! treatments are `{strictly additive}' if the differences between any two treatments are constant across all units, i.e., $Y_i(l) = Y_i(k) + C(k,l)$ for   some fixed real numbers $C(k,l)$.
\end{definition}
%%%%%%%%%%%%%%%%%%%%
%For any positive integer $p$, denote by
%${\vones}_p$ the $p$-dimensional vector of all ones,
%${\mathbf I}_p$ the $p$-dimensional identity matrix,
%${\mathbf J}_p = \vones_p\vones_p^{\T}$ the $p \times p$ matrix of all ones,
%and $\mP_p =  \mI_p - p^{-1}\mJ_p$ the $p$-dimensional projection matrix with column space orthogonal to $\vones_p$.
%%%%%%%%%%%%%%

\smallskip

\noindent
Let $\bar Y(k) = N^{-1} \sum_{i=1}^N  Y_i(k)$
be the population average under treatment $k$,
\begin{center}
$
S^2(k,l)
=
(N-1)^{-1}\sum_{i=1}^N \{ Y_i(k) - \bar{Y}(k) \} \{ Y_i(l) - \bar{Y}(l)\}
$
\end{center}
be the \emph{finite-population covariance} of $Y_i(k)$ and $Y_i(l)$,
and
\begin{center}
$\mS = \left(\hspace{-1mm}\left( S^2(k,l) \right)\hspace{-1mm}\right)_{K\times K}$
\end{center}
%$$\mS = \left(\hspace{-1mm}\left( S^2(k,l) \right)\hspace{-1mm}\right)_{K\times K}$$
be the \emph{finite-population variance-covariance matrix}.
Lemma \ref{def1} gives an alternative characterization of strict additivity in terms of $S^2(k,l)$.

\begin{lemma}
\label{def1}
The potential outcomes $Y_i(k)$\!  of $N$\! units under $K$\! treatments are {strictly additive} if and only if the finite-population covariances $S^2(k,l)$ are the same for all $k, l \in \{1, \ldots, K\}$, i.e., $\mS = S_0^2 {\mJ}_K$, where $\mJ_K$ is a $K \times K$ matrix of 1's and $S^2_0$ is a fixed non-negative number.
\end{lemma}

%%%%%%%%%%%%%%%%%%%%%%
%%%%%%%%%%%%%%%%%%%%%%
For simplicity,
we will omit the `finite-population' before `covariance' in the following text when no confusion would arise.
All averages and covariances over \emph{a finite set of fixed numbers} will be finite-population in nature, and defined the same way as $\bar Y(k)$ and $S^2(k,l)$ are defined for $Y_i(k)$.

%%%%%%%%%%%%%%%%%%%%%%%%%%%%%%%%%%%%%
\subsection{Experimental units with block structure\label{sec::scienceWithBlock}}
%%%%%%%%%%%%%%%%%%%%%%%%%%%%%%%%%%%%%
Whereas all definitions and discussion above apply universally to any $K$-treatment experiment with $N$ experimental units,
possible extensions arise when the experimental units in question exhibit certain block structure --- due to either intrinsic characteristics like geographic proximity, or extrinsic arrangements as induced by the design.

Without essential loss of generality, assume the $N$ experimental units are nested under $W$ blocks, each of size $M = N/W$. Generalization to unequal block sizes is straightforward.
%, and assume the mapping from $\{(w,m): w = 1, \ldots, W; m=1, \ldots, M\}$ to $\{1,\ldots, N\}$ in lexicographic order such that $i = M(w-1)+m$.
Index the blocks by $w$, running from 1 to $W$,
and the units within block $w$ by $(wm)$, running from $(w1)$ to $(wM)$.
Define the \emph{block average potential outcome}s as
\begin{center}
$Y_{(w)}(k) \,=\, M^{-1}\sum_{m=1}^M Y_{(wm)}(k) \quad (k=1, \ldots,K)$.
\end{center}
%These aggregated potential outcomes inspire the following definitions of \emph{between-} and \emph{within-block additivities} as two possible relaxations of the strict additivity in Definition \ref{def::strictAdd}.

These aggregated potential outcomes enable %allow for
the definitions of some weaker forms of additivity as compared to that in Definition \ref{def::strictAdd}.
%\\
%These aggregated potential outcomes allow us to relax the strict additivity into the following between- and within-block forms.} %us to relax the definition of strict additivity.
%the definition of some less strict forms of additivity % additivities
%as compared to
%that in Definition \ref{def::strictAdd}.
%%%%%%%%%%%%%
%In particular, given an $N\times K$ \textsc{pom} $\mY$ with block structure, let
%${\mathbf Y}_{\text{block}} = \left(\hspace{-1mm}\left({Y}_{(w)}(k)\right)\hspace{-1mm}\right)$ be the ${W\times K}$ \emph{block-level \textsc{pom}}, and let
%$\mathbf{Y}_{(w)} = \left(\hspace{-1mm}\left({Y}_{(wm)}(k)\right)\hspace{-1mm}\right)$ be the ${M\times K}$  \emph{within-block \textsc{pom}} of block $w$ $(w = 1, \ldots, W)$.
\begin{definition}
\label{def::2}
The potential outcomes $Y_{(wm)}(k)$ of $N$\! units in $W$\! blocks under $K$\! treatments are
\begin{itemize}
\item `{between-block additive}' if the corresponding block average potential outcomes $Y_{(w)}(k)$ are strictly additive across all $w$, i.e., $Y_{(w)}(l) = Y_{(w)}(k) + C(k,l)$ for some fixed real numbers $C(k,l)$;
\smallskip
\item `{within-block additive}' if for each $w$, the potential outcomes $Y_{(wm)}(k)$ of the $M$ units within block $w$ are strictly additive, i.e., $Y_{(wm)}(l) = Y_{(wm)}(k) + C_w(k,l)$ for some fixed real numbers $C_w(k,l)$.
\end{itemize}
\end{definition}

%\begin{remark}
%If we are allowed to equate the strict additivity, informally, with the conventional assumption of constant regression coefficients across all units in the context of linear regression model,
%then the within-block additivity can, likewise, be equated with a somewhat unconventional variant of constant regression coefficients across all units within the same block.
%\end{remark}

Strictly additive  potential outcomes, if nested under blocks, must be strictly additive within each block and have strictly additive block averages.
Lemma \ref{l:SPstrictadd} asserts that the converse is also true.

\begin{lemma}\label{l:SPstrictadd}
The potential outcomes of $N$ units in $W$ blocks are strictly additive if and only if they are both between- and within-block additive.
\end{lemma}

%%%%%%%%%%%%%%%%%%
%%%%%%%%%%%%%%%%%%%%%%%%%%%%%%%%
Define
the  \emph{between-block covariance} of $Y_{(wm)}(k)$ and $Y_{(wm)}(l)$ by the covariance of $Y_{(w)}(k)$ and $Y_{(w)}(l)$:
\begin{equation}
\label{eq::SbtwY_element}
S^2_{\textnormal{btw}}(k,l)
\,=\, (W-1)^{-1} \sum_{w=1}^W \{Y_{(w)}(k) - \bar{Y}(k)\} \{Y_{(w)}(l) - \bar{Y}(l)\}\,,
\end{equation}
and their \emph{within-block covariance} by
\begin{equation}
\label{eq::SinY_element}
S^2_{\textnormal{in}}(k,l)
\,=\, W^{-1}\sum_{w=1}^W S^2_{(w)}(k,l)\,,
\end{equation}
where
\begin{center}
$
S^2_{(w)}(k,l) = (M-1)^{-1} \sum_{m=1}^M \{Y_{(wm)}(k) - Y_{(w)}(k) \}\{ Y_{(wm)}(l) - Y_{(w)}(l)\}$
\end{center}
is the covariance of $Y_{(wm)}(k)$ and $Y_{(wm)}(l)$ within block $w$.
%Let
%${\bm Y}_{\text{block}}(k) = \left({Y}_{(1)}(k), \, \ldots, \, {Y}_{(W)}(k)\right)^\T$.
%We measure
%%%%%%%%%%%%%%%%%
Let
\begin{center}
$\mSb = (\hspace{-1mm}( S^2_{\textnormal{btw}}(k,l) )\hspace{-1mm})_{K\times K}$\,,\quad $\mSi =(\hspace{-1mm}( S^2_{\textnormal{in}}(k,l))\hspace{-1mm})_{K\times K}$.
\end{center}
Applying Lemma \ref{def1} to Definition \ref{def::2} allows us to characterize the between- and within-block additivities via their respective covariances as follows.

\begin{lemma}
\label{def2}
The potential outcomes $Y_{(wm)}(k)$ of $N$ units in $W$ blocks under $K$ treatments are
\begin{itemize}
\item between-block additive if and only if $\mSb= S^2_{\textrm{btw}} {\mJ}_K$
for some non-negative number $S^2_{\textnormal{btw}}$;
\smallskip
\item within-block additive if and only if $\mSi= S^2_{\textrm{in}} {\mJ}_K$
for some non-negative number $S^2_{\textnormal{in}}$.
\end{itemize}
\end{lemma}

\begin{remark}
For potential outcomes that are {between-block additive}, the common value $S^2_{\textrm{btw}}$ provides a measure of the block variability.
\end{remark}

%%%%%%%%%%%%%%%%%%%%%%%%%%%%%%%%%%%%
\subsection{Decomposition of covariances}
%%%%%%%%%%%%%%%%%%%%%%%%%%%%%%%%%%%%
For any positive integer $p$,
let ${\vones}_p$ be the $p$-dimensional vector of 1's,
${\mathbf J}_p = {\vones}_p{\vones}_p^\T$ be the $p \times p$ matrix of 1's,
${\mathbf I}_p$ be the $p \times p$ identity matrix,
and $\mP_p =  \mI_p - p^{-1}\mJ_p$ be the $p \times p$ projection matrix with column space orthogonal to $\vones_p$.
Let $\otimes$ denote the Kronecker product.

Let $\bm Y(k) = (Y_1(k), \ldots, Y_N(k))^\T =  (Y_{(11)}(k), \ldots, Y_{(WM)}(k))^\T$ be the same potential outcomes vector indexed in two different ways --- the running index $i$ and the block-unit double-index $(wm)$.
Straightforward algebra determines
\begin{center}
$
\mPi
%&=\, \mI_N - \left\{\mI_W \otimes (M^{-1}\mJ_M)\right\}
\,=\, {\mI}_W \otimes {\mathbf{P}}_M\,,\quad
\mPb
%&=\,
%\mI_W \otimes (M^{-1}\mJ_M) - N^{-1}\mJ_N
\,=\,
%\mI_W \otimes (M^{-1}\mJ_M) -  (W^{-1}\mJ_W) \otimes (M^{-1}\mJ_M) \nonumber\\
%&=\,
{\mathbf{P}}_W \otimes (M^{-1} {\mJ}_M)
$
\end{center}
as two mutually orthogonal projection matrices, with
\begin{equation}
\label{eq:Projections}
\begin{array}{rll}
\mPi \bm Y(k) \hspace{-4mm}
& \hspace{-2mm}=&\hspace{-2mm}\left(\hspace{-1mm}\left( Y_{(wm)}(k) - Y_{(w)}(k)\right)\hspace{-1mm}\right)_{N\times 1}
,\vspace{1mm}\\
%\\
%%%%%
\mPb \bm Y(k)  \hspace{-4mm}
& \hspace{-2mm}=&\hspace{-2mm}
\left(\hspace{-1mm}\left(Y_{(w)}(k) - \bar{Y}(k)\right)\hspace{-1mm}\right)_{N\times 1}.
\end{array}
\end{equation}
%Also, for the set of potential outcomes, $Y_i(k)$, let $\bm Y_i = (Y_i(1), \ldots, Y_i(K))^\T$ be the $K$ potential outcomes of unit $i$,
%Let $\bm Y(k) = (Y_1(k), \ldots, Y_N(k))^\T$  be the $N$ potential outcomes under treatment $k$, and
Let $\mY= (\hspace{-1mm}( Y_i(k) )\hspace{-1mm})_{N\times K}$ be the $N\times K$ \emph{potential outcome matrix} (\textsc{pom}) with %$\vY_i^\T$ as its $i$th row and
$\bm Y(k)$ as its $k$th column.
It follows from \eqref{eq:Projections} that
%\begin{center}
%$
%Y_{(wm)}(k) -  \bar{Y}(k) =\, \{Y_{(wm)}(k) - Y_{(w)}(k)\}
%+ \{Y_{(w)}(k) - \bar{Y}(k)\}
%$
%\end{center}
%can now be written in matrix form
%\begin{align*}
%%\label{eq::ANOVA_science_element}
%Y_{(wm)}(k)\,=\, \bar{Y}(k) + \{Y_{(w)}(k) - \bar{Y}(k)\}
%+ \{Y_{(wm)}(k) - Y_{(w)}(k)\}&,\\
%Y_{(wm)}(k) -  \bar{Y}(k) =\, \{Y_{(w)}(k) - \bar{Y}(k)\}
%+ \{Y_{(wm)}(k) - Y_{(w)}(k)\}
%&\end{align*} that
\begin{align}
\label{eq::gougu}
\mathbf{P}_N  \bm Y(k) &=\, \left( \mI_N - N^{-1}\vones_N\vones_N^\T \right) \bm Y(k) \nonumber\\
&=\, \bm Y(k)  - \vones_N \bar{Y}(k) \,=\, \left(\hspace{-1mm}\left( Y_{(wm)}(k) - \bar{Y}(k)\right)\hspace{-1mm}\right) \nonumber\\
&=\, \left(\hspace{-1mm}\left( Y_{(wm)}(k) - Y_{(w)}(k)\right)\hspace{-1mm}\right)
 +
\left(\hspace{-1mm}\left(Y_{(w)}(k) - \bar{Y}(k)\right)\hspace{-1mm}\right)\nonumber\\
&=\,
 \mPi \bm Y(k)+ \mPb \bm Y(k)\,,\nonumber\vspace{1mm}\\
%%%%%%%%%%%%%%%%%%%%%%%%%
\mathbf{P}_N  \mY &=\,  \mPi \mY+ \mPb \mY\,,\nonumber\vspace{1mm}\\
 \mY^\T \mathbf{P}_N  \mY &=\,   \mY \mPi \mY+  \mY \mPb \mY
\end{align}
and that
\begin{align}
\label{mSbmSi}
&S^2_{\text{in}}(k,l) \,=\, \frac{\vY(k)^\T \mPi \vY(l)}{W(M-1)}\,,&&S^2_{\text{btw}}(k,l) \,=\, \frac{\vY(k)^\T \mPb \vY(l)}{(W-1)M}\,,\nonumber\\
&\mSi \,=\, \frac{\mathbf{Y}^\T \mPi \mY}{W(M-1)}\,,&&\mSb \,=\, \frac{\mathbf{Y}^\T \mPb \mY}{(W-1)M}\,.
\end{align}
%\begin{equation}
%\label{mSbmSi}
%\mSb \,=\, \frac{1}{(W-1)M}\{(W-1)M\}^{-1}\mathbf{Y}^\T \mPb \mY
%\mSi \,=\, \frac{1}{W(M-1)}\mathbf{Y}^\T \mPi \mY\,,
%\end{equation}
%\begin{align}
%\label{mSbmSi}
%&S^2_{\text{btw}}(k,l) \,=\, \{(W-1)M\}^{-1}\vY(k)^\T \mPb \vY(l)\,,\\
%&
%S^2_{\text{in}}(k,l)\,=\, \{W(M-1)\}^{-1}\vY(k)^\T \mPi \vY(l)\,,\nonumber\\
%&\mSb \,=\, \{(W-1)M\}^{-1}\mathbf{Y}^\T \mPb \mY\,,\,\,
%\mSi \,=\, \{W(M-1)\}^{-1}\mathbf{Y}^\T \mPi \mY\,.\nonumber
%\end{align}
%
Combining \eqref{eq::gougu} with \eqref{mSbmSi} yields the first major result of this article.

\begin{theorem}
\label{thm::decompS}
The {variance-covariance matrix} $\mS$ is a linear combination of $\mSb$ and $\mSi$:
\begin{equation*}
\mS  \,=\, \frac{(W-1)M}{N-1} \mSb + \frac{W(M-1)}{N-1}\mSi\,. %\label{eq::decomp_M}
\end{equation*}
%Entrywise,
%\begin{equation*}
%S^2(k,l)  \,=\, \frac{(W-1)M}{N-1} S^2_{\textnormal{btw}}(k,l)  + \frac{W(M-1)}{N-1}S^2_{\textnormal{in}}(k,l) \quad (k,l = 1, \ldots K)\,. %\label{eq::decomp_M}
%\end{equation*}
\end{theorem}
%\begin{proof}
%\begin{equation*}
%(N-1)\mS  \,=\,\bm Y^{\T} \mP_N \bm Y \,=\,\bm Y^{\T} \mPb\bm Y + \bm Y^{\T} \mPi \bm Y\,=\,
%(W-1)M \mSb + W(M-1)\mSi\,.
%\end{equation*}
%\end{proof}

We have so far introduced,  in the context of general $K$-treatment experiments,
all concepts about the potential outcomes framework that we consider relevant to the current topic.
%all
%the concepts and results about the study population and their potential outcome
%We have so far introduced, in the context of general $K$-treatment experiments,
%all
%the concepts and results about the study population and their potential outcomes.
% that we consider relevant to the current topic.
Specific discussion of $2^2$ factorial experiments starts in the next section, in which
we formally introduce %the $2^2$ factorial experiment, %as a
this special type of $4$-treatment experiment,
%%all discussions
%%Discussions from now on will be conducted in the context of $2^2$ factorial experiments unless specified otherwise.
together with %define its chief causal interest under the Neymanian inferential framework.
its chief causal questions of interest.

%%%%%%%%%%%%%%%%%%%%%%%%%%%%%%%%
% ========================================
\section{Causal effects for $2^2$ factorial experiments\label{section::causalEffect}}
% ========================================
\emph{$2^2$ factorial experiments}, as the name suggests, involve different $K=4$ treatments as the $2^2$ possible combinations of two 2-level \emph{factor}s.
Code the two factors as $A$ and $B$.
Of chief causal interest are the \emph{main effect of factor $A$} (indexed by `$A$'), the \emph{main effect of factor $B$}  (indexed by `$B$'), and the \emph{effect of interaction between $A$ and $B$} (indexed by `$AB$', also refer to as factor $AB$).
We set out in this section their formal definitions at unit, block, and population levels.
 %%%%%%%%%%%%%%%%%%%%%%%%%%%%%%%%%
% ==============================
\subsection{Causal effects at unit and population levels}
%%%%%%%%%%%%%%%%%%%%%%
Code the two levels of factor $A$ as $\{-1_A, +1_A\}$ and those of factor $B$ as $\{-1_B, +1_B\}$.
We represent the four combinations as $(-1_A, -1_B)$, $(-1_A, +1_B)$, $(+1_A, -1_B)$, $(+1_A, +1_B)$, and name them in lexicographic order as treatments 1 to 4 (Table \ref{tb::4trt}).

\begin{table}[ht]
\caption{\label{tb::4trt}The four treatments in a $2^2$ factorial experiment.}
\begin{tabular}{cccc}
\hline
Treatment & Factor $A$ & Factor $B$ & Interaction ($AB$) \\\hline
1 & $-1_A$ & $-1_B$ & $+1_{AB}$\\
2 & $-1_A$ & $+1_B$ & $-1_{AB}$\\
3 & $+1_A$ & $-1_B$ & $-1_{AB}$\\
4 & $+1_A$ & $+1_B$ & $+1_{AB}$\\\hline
\end{tabular}
\end{table}

Given a study population of $N$ units, denote by $\vY_i = (Y_i(1), Y_i(2), Y_i(3), Y_i(4))^\T$ the potential outcomes vector of unit $i$.
Let ${\bm g}_A = (-1,-1,+1,+1)^{\T}$ summarize the levels of factor $A$  in treatments 1 to 4 --- i.e., the `Factor $A$' column in Table \ref{tb::4trt}
--- and ${\bm g}_B = (-1,+1,-1,+1)^{\T}$, ${\bm g}_{AB} =   (+1,-1,-1,+1)^{\T}$ likewise.
The \emph{factorial effect of factor $F \in \mathcal{F} = \{A, B, AB\}$} on unit $i$ is defined as
\begin{center}
$
\tau_{i\text{-}F} \,=\, 2^{-1} \vg_F^{\T} {\bm Y}_i\,,
$
\end{center}
with population average
%= N^{-1}\vones_N^{\T} \mY$,
\begin{equation}
\label{eq::estimands}
\tau_F
\,=\, N^{-1}\sum_{i=1}^N \tau_{i\text{-}F}\,=\,
%\frac{1}{N} 2^{-1} \vg_F^{\T}\sum_{i=1}^N  {\bm Y}_i \,=\, 2^{-1} \vg_F^{\T} \left(  \frac{1}{N} \sum_{i=1}^N  {\bm Y}_i\right) \nonumber \\
%&=\,
2^{-1} \vg_F^{\T}(  \bar Y(1), \bar Y(2), \bar Y(3), \bar Y(4) )^\T.
\end{equation}
Let
$
S^2_F
= (N-1)^{-1}\sum_{i=1}^N( \tau_{i\text{-}F} - \tau_F)^2
$.
Lemma \ref{def3} restates strict additivity in terms of the factorial effects and their variances.

\begin{lemma}\label{def3}
The $4\times N$ potential outcomes of $N$ units in a $2^2$ factorial experiment, $Y_i(k)$ ($i = 1, \ldots, N$; $k = 1,2,3,4$), are strictly additive if and only if all three unit-level factorial effects are constant across all units, i.e., $\tau_{i\text{-}F} = \tau_F$ for all $i \in \{1, \ldots, N\}$ and $F \in \mathcal{F}$; or equivalently,  $S^2_F=0$ for each $F \in \mathcal{F}$.
\end{lemma}

%%%%%%%%%%%%%%%%%%%%%%%%%%%%%%%%
\subsection{Causal effects at block level}
%%%%%%%%%%%%%%%%%%%%%%%%%%%%%%
When the study population is nested under blocks,
%factorial effects, like the previously discussed potential outcomes, can be defined on a wider basis.
%%%%%%%%%%%%%
%Specifically, double-index the units by $(wm)$, we
further define
\begin{align}
\label{eq::blockFE}
\tau_{(w)\text{-}F}
&=\,  M^{-1}\sum_{m=1}^M \tau_{(wm)\text{-}F}\\
&=\, 2^{-1} \vg_F^{\T} \left(
Y_{(w)}(1),
Y_{(w)}(2),
Y_{(w)}(3),
Y_{(w)}(4) \right)^\T\nonumber
\end{align}
as the \emph{block average factorial effects}.
% of factor $F \in \mathcal{F}$ on block $w$,
%where
The $\tau_{(wm)\text{-}F}$ in \eqref{eq::blockFE} are the same unit-level factorial effects as $\tau_{i\text{-}F}$, only now under block-unit double-index $(wm)$.
%
%The definition of $\tau_{(w)\text{-}F}$ in \eqref{eq::blockFE} can be interpreted in two different ways: either as the average of unit-level factorial effects (the first equality), or as the contrast of block average potential outcomes (the second equality).
With all blocks being of equal size, the three levels of factorial effects satisfy
\begin{center}
$
%\tau_{(w)\text{-}F} \,=\,  M^{-1}\sum_{m=1}^M \tau_{(wm)\text{-}F} \,, \,\,\,
\tau_{F} \,=\, W^{-1}\sum_{w=1}^W\tau_{(w)\text{-}F} \,=\, N^{-1} \sum_{i=1}^N \tau_{i\text{-}F}\,.
$
\end{center}
%%%%%%%%%%%%%%%%%%
%Let $\bm \tau_F = (\tau_{1\text{-}F}, \ldots, \tau_{N\text{-}F})^\T$.
Define the \emph{between-} and \emph{within-block variances} of $\tau_{i\text{-}F}$ the same way \eqref{eq::SbtwY_element}--\eqref{eq::SinY_element} defined $S^2_\text{btw}(k,k)$ and $S^2_\text{in}(k,k)$:
\begin{align}
\label{eq::SinF_element}
S^2_{F\text{-btw}} &=\, (W-1)^{-1} \sum_{w=1}^W({\tau}_{(w)\text{-}F} - \tau_F)^2\,,\\
%\,=\, \frac{\bm \tau_F^\T \mPb \bm\tau_F}{(W-1)M}\,,\\
S^2_{F\text{-in}}&=\, W^{-1} \sum_{w=1}^W\left\{(M-1)^{-1} \sum_{m=1}^M (\tau_{(wm)\text{-}F} - \tau_{(w)\text{-}F})^2 \right\}.\nonumber
%\,=\, \frac{\bm \tau_F^\T \mPi \bm \tau_F}{W(M-1)}\nonumber
\end{align}
%where the matrix forms follow from \eqref{mSbmSi}.
%%%%%%%%%%%%%%%%%%
These variances give an alternative characterization of the {between- and within-block additivities} as detailed in Lemma \ref{def4}.

\begin{lemma}\label{def4}
Given $N$ experimental units in a $2^2$ factorial experiment that are nested under $W$ blocks and indexed by $(wm)$,
the corresponding $4\times N$ potential outcomes are
\begin{itemize}
\item {between-block additive} if and only if all three block average factorial effects $\tau_{(w)\text{-}F}$ are constant across all blocks, i.e., $\tau_{(w)\text{-}F} = \tau_F$ for all $w \in \{1, \ldots, W\}$ and $F \in \mathcal{F}$,  or equivalently, $S^2_{F\text{-in}} = 0$ for each $F \in \mathcal{F}$;
\item {within-block additive} if and only if all three unit-level factorial effects $\tau_{(wm)\text{-}F}$ are constant within each block, i.e.,
$\tau_{(wm)\text{-}F} = \tau_{(w)\text{-}F}$ for all $w \in \{1, \ldots, W\}$ and $F \in \mathcal{F}$, or equivalently, $S^2_{F\text{-btw}}= 0$ for each $F \in \mathcal{F}$.
\end{itemize}
\end{lemma}

%%%%%%%%%%%%%%%%%%%%%%%%%%%%%%%%%%%
%%%%%%%%%%%%%%%%%%%%%%%%%%%%%%%%%%%%%%%
% ---------------------------------------------------------------------------------------
\section{$2^2$ split-plot design\label{section::AM}}
% --------------------------------------------------------------------------------------
%%%%%%%%%%%%%%%%%%%%%%%%%%%%%%%%%%%%%%%
We introduce in this section the $2^2$ split-plot design as a treatment assignment mechanism distinct from complete randomization.

%%%%%%%%%%%%%%%%%%
\subsection{Notation and definitions}
%%%%%%%%%%%%%%%%%
%In a $K$-treatment experiment with $N$ units and
% ($k=1, \ldots, K$, $\sum_{k=1}^K = N$),
Assume fixed \emph{treatment arm size}s $N_k$ ($k=1, \ldots, K$, $\sum_{k=1}^K N_k = N$).
Let $T_i$ be the assignment variable, taking the value $k$ if unit $i$ is assigned to treatment $k$.
Let ${\bm Z}(k) = ( {I}_{\{T_1 =k \}}, \ldots, {I}_{\{ T_N = k \}})^{\T}$, with $\sum_{i=1}^N {I}_{\{T_i =k \}} = N_k$.
Let $\bm Z^* = (N_1^{-1}\bm Z(1)^\T, \ldots, N_K^{-1}\bm Z(K)^\T)^\T$,
in which we normalize each $\bm Z(k)$ by the sum of its entries $N_k$.
Refer to $\bm Z^*$ as the \emph{assignment vector}.
It %provides an alternative representation of conveys the same information as that contained in
gives a full representation of the randomization result, in a form that promises easier algebra than $\{T_i\}_{i=1}^N$.
%such that each segment $N_k^{-1}\bm Z(k)$ having entry-wise sum 1.

% such that $\sum_{k=1}^K N_k = N$
%The set of all possible values $(\bm Z(1), \ldots, \bm Z(K))$ can take
%--- the vectorization of $\{\bm Z(k)\}_{k=1}^K$ with each $\bm Z(k)$ normalized by its sum of elements ---
%: $\vones_{KN}^\T \bm Z^* = K$.

%%%%%%%%%%%%%%%%%%%%%%%%%%%%%%%%%%%%
\subsection{A classical agricultural experiment}
%%%%%%%%%%%%%%%%%%%%%%%%%%%%%%%%%%%%
Consider a classical agricultural experiment in which two levels of irrigation (factor $A$) and two levels of fertilizers (factor $B$) are to be tested on $N=8$ plots of land (experimental units) nested within four whole-plots (blocks) (Figure \ref{t:plots}).
%%%%%%%%%%%%%%%%%%%
\begin{figure}[ht]
\centering
\caption{Eight plots of land nested within four whole-plots, labeled by both running index $i$ $(i = 1, \ldots, 8)$ and block-unit double-index $(wm)$ $(w = 1, \ldots, 4; m = 1,2)$.\label{t:plots}}
\begin{tabular}{lll}
\begin{tabular}{|c|}
\multicolumn{1}{c}{Whole-plot 1}\\\hline
Plot 1\,\, (11)
\\\hline
%%%%%%%%%%
Plot 2\,\, (12)
\\\hline
\end{tabular}
&\quad&
%%%%%%%%%%%%%
%%%%%%%%%%%%%
\begin{tabular}{|c|}
\multicolumn{1}{c}{Whole-plot 2}\\\hline
Plot 3\,\, (21)
\\\hline
%%%%%%%%%%
Plot 4\,\, (22)
\\\hline
\end{tabular}
\vspace{1mm}\\
%%%%%%%%%%%%%
%%%%%%%%%%%%%
\begin{tabular}{|c|}
\multicolumn{1}{c}{Whole-plot 3}\\\hline
Plot 5\,\,  (31)
\\\hline
%%%%%%%%%%
Plot 6\,\,  (32)
\\\hline
\end{tabular}
&\quad&
%%%%%%%%%%%%%
%%%%%%%%%%%%%
\begin{tabular}{|c|}
\multicolumn{1}{c}{Whole-plot 4}\\\hline
Plot 7\,\,  (41)
\\\hline
%%%%%%%%%%
Plot 8\,\,  (42)
\\\hline
\end{tabular}
\end{tabular}
\end{figure}
%%%%%%%%%%%%%%%%%

Assume each combination is to be replicated on $N/K = 8/4=2$ plots.
The assignment process can be visualized as
a distribution of the eight \emph{tags}:
\begin{center}
$
\begin{array}{llll}
\mybox{\small$(-1_A, -1_B)$}&\mybox{\small$(-1_A, +1_B)$} &\mybox{\small$(+1_A, -1_B)$}&\mybox{\small$(+1_A, +1_B)$}
\end{array}$
\\
$
\begin{array}{llll}
\mybox{\small$(-1_A, -1_B)$}&\mybox{\small$(-1_A, +1_B)$} &\mybox{\small$(+1_A, -1_B)$}&\mybox{\small$(+1_A, +1_B)$}
\end{array}$
\end{center}
to the eight plots.
%A `design' specifies how the distribution is to be done.
%%%%%%%%%%%%
A \emph{completely randomized} design, as the name suggests, distributes the tags at complete random.
Any arrangement of the eight tags is equally likely, with Figures \ref{fig::SP} and \ref{fig::CR} being two examples.
%with $\bm Z(1) + \bm Z(2) + \bm Z(3) + \bm Z(4) = \vones_8$.
%%%%%%%%%%%%%%%%%%%

A \emph{split-plot} design, on the other hand, requires the two plots within each whole-plot to receive the same level of irrigation.
This could be due to resource constraint, say, the technical difficulty in applying irrigation to areas smaller than a whole-plot, or on purpose, to minimize the bias from block heterogeneity when comparing the fertilizers.
Whereas Figure \ref{fig::SP} satisfies the requirement and remains a possible arrangement,
the different irrigation levels in plots 1 and 2 disqualify Figure \ref{fig::CR} from the candidate pool.

In general, with the experimental units in hand nested under several blocks,
a split-plot design identifies one factor as the \emph{whole-plot factor}, and restricts its level to be the same within each block. The possible assignments under a split-plot design thus constitute a proper subset of the possible assignments under a completely randomized one.
This brings out the first and most salient distinction between the two designs.

%%%%%%%%%%%%%%%%%%%
Formal definitions of these two designs are given in the next section, along with the sampling moments of their respective assignment vectors $\bm Z^*$.
Not only do these sampling moments enable a first quantitative comparison between the two designs, but also provide the fundamental building blocks for computing the sampling variances of our major estimates to be introduced in Section \ref{section::var}.

%%%%%%%%%%%%%%%
\begin{figure}[ht]
\caption{An assignment possible under both the completely randomized and split-plot designs. \label{fig::SP}}
%\begin{tabular}{|c|}
%\multicolumn{1}{c}{Whole-plot 1}\\\hline
%Plot 1:
%\\\hline
%%%%%%%%%%%
%Plot 2:  $(-1_A, -1_B)$
%\\\hline
%\end{tabular}
%\quad
%%%%%%%%%%%%%%
%%%%%%%%%%%%%%
%\begin{tabular}{|c|}
%\multicolumn{1}{c}{Whole-plot 2}\\\hline
%Plot 3: $(+1_A, +1_B)$
%\\\hline
%%%%%%%%%%%
%Plot 4: $(+1_A, -1_B)$
%\\\hline
%\end{tabular}
%\vspace{1mm}\\
%%%%%%%%%%%%%%
%%%%%%%%%%%%%%
%\begin{tabular}{|c|}
%\multicolumn{1}{c}{Whole-plot 3}\\\hline
%Plot 5: $(+1_A, -1_B)$
%\\\hline
%%%%%%%%%%%
%Plot 6:  $(+1_A, +1_B)$
%\\\hline
%\end{tabular}
%\quad
%%%%%%%%%%%%%%
%%%%%%%%%%%%%%
%\begin{tabular}{|c|}
%\multicolumn{1}{c}{Whole-plot 4}\\\hline
%Plot 7:  $(-1_A, +1_B)$
%\\\hline
%%%%%%%%%%%
%Plot 8:  $(-1_A, -1_B)$
%\\\hline
%\end{tabular}
\begin{tabular}{cccc}
\hline
Treatment $k$ & Combination & Recipients $i$ & Indicators of recipients $\bm Z(k)$\\\hline
1 & $(-1_A, -1_B)$ & 2, 8 &$(0, 1, 0, 0, 0, 0, 0, 1)^\T$\\
2 & $(-1_A, +1_B)$ & 1, 7 &$(1, 0, 0, 0, 0, 0, 1, 0)^\T$\\
3 & $(+1_A, -1_B)$ & 4, 5 &$(0,0, 0, 1,1, 0, 0, 0)^\T$\\
4 & $(+1_A, +1_B)$ & 3, 6 &$(0, 0, 1, 0, 0, 1, 0, 0)^\T$\\
\hline
\end{tabular}

\begin{tabular}{lll}\\
\begin{tabular}{l|c|}
\multicolumn{1}{c}{}&\multicolumn{1}{c}{Whole-plot 1}\\\cline{2-2}
Plot 1 &  $(-1_A, +1_B)$
\\\cline{2-2}
%%%%%%%%%%
Plot 2 &$(-1_A, -1_B)$
\\\cline{2-2}
\end{tabular}
&\quad&
\begin{tabular}{|c|l}
\multicolumn{1}{c}{Whole-plot 2}&\multicolumn{1}{c}{}\\\cline{1-1}
$(+1_A, +1_B)$ &Plot 3
\\\cline{1-1}
%%%%%%%%%%
$(+1_A, -1_B)$&Plot 4
\\\cline{1-1}
\end{tabular}
\vspace{1mm}\\
%%%%%%%%%%%%%
%%%%%%%%%%%%%
\begin{tabular}{l|c|}
\multicolumn{1}{c}{}&\multicolumn{1}{c}{Whole-plot 3}\\
\cline{2-2}
Plot 5& $(+1_A, -1_B)$\\\cline{2-2}
%%%%%%%%%%
Plot 6& $(+1_A, +1_B)$
\\\cline{2-2}
\end{tabular}
&\quad&
\begin{tabular}{|c|l}
\multicolumn{1}{c}{Whole-plot 4}& \multicolumn{1}{c}{}\\\cline{1-1}
$(-1_A, +1_B)$ &Plot 7
\\\cline{1-1}
%%%%%%%%%%
$(-1_A, -1_B)$ &Plot 8
\\\cline{1-1}
\end{tabular}
\end{tabular}

\end{figure}

\begin{figure}[ht]
\caption{An assignment possible under the completely randomized design yet impossible under the split-plot design. \label{fig::CR}}
\begin{tabular}{cccc}
\hline
Treatment $k$ & Combination & Recipients $i$ & Indicators of recipients $\bm Z(k)$\\\hline
1 & $(-1_A, -1_B)$ & 1, 3 &$(1, 0, 1, 0, 0, 0, 0, 0)^\T$\\
2 & $(-1_A, +1_B)$ & 4, 8 &$(0, 0, 0, 1, 0, 0, 0, 1)^\T$\\
3 & $(+1_A, -1_B)$ & 2, 5 &$(0, 1,0, 0,1, 0, 0, 0)^\T$\\
4 & $(+1_A, +1_B)$ & 6, 7 &$(0, 0, 0, 0, 0, 1, 1, 0)^\T$\\
\hline
\end{tabular}
\begin{tabular}{lll}\\
%%%%%%%
\begin{tabular}{l|c|}
\multicolumn{1}{c}{}&\multicolumn{1}{c}{Whole-plot 1}\\\cline{2-2}
Plot 1 & $(-1_A, -1_B)$
\\\cline{2-2}
%%%%
Plot 2 &$(+1_A, -1_B)$
\\\cline{2-2}
\end{tabular}
%%%%%%%%
&\quad&
%%%%%%%%%%%%%
\begin{tabular}{|c|l}
\multicolumn{1}{c}{Whole-plot 2}&\multicolumn{1}{c}{}\\\cline{1-1}
$(-1_A, -1_B)$&Plot 3
\\\cline{1-1}
%%%%%%%%%%
$(-1_A, +1_B)$&Plot 4
\\\cline{1-1}
\end{tabular}
\vspace{1mm}\\
%%%%%%%%%%%%%
\begin{tabular}{l|c|}
\multicolumn{1}{c}{}&\multicolumn{1}{c}{Whole-plot 3}\\\cline{2-2}
Plot 5 & $(+1_A, -1_B)$
\\\cline{2-2}
Plot 6& $(+1_A, +1_B)$
\\\cline{2-2}
\end{tabular}
&\quad&
%%%%%%%%%%%%%
\begin{tabular}{|c|l}
\multicolumn{1}{c}{Whole-plot 4}& \multicolumn{1}{c}{}\\\cline{1-1}
$(+1_A, +1_B)$ &Plot 7
\\\cline{1-1}
%%%%%%%%%%
$(-1_A, +1_B)$ &Plot 8
\\\cline{1-1}
\end{tabular}
\end{tabular}
\end{figure}

%%%%%%%%%%%%%%%%%%%%%%%%%%%%%%%
\subsection{Designs and their respective assignment vectors}
%%%%%%%%%%%%%%%%%%%%%%%%%%%%%%%%
As far as $2^2$ factorial experiment is concerned,
%variations of design exist even after we decide `to assign in a completely randomized way.'
variations of complete randomization exist.
The two factors of interest can be assigned either one at a time, each by a two-treatment complete randomization, or, jointly,  via a single complete randomization of the treatment combinations 1 to 4.
%variations of design exist even after we decide `to assign in a completely randomized way.'
%We can, for example, assign the two factors one at a time, each by a two-treatment complete randomization of all $N$ units, such that a unit that gets $z_A$ level of factor $A$ and $z_B$ level of factor $B$ will receive $(z_A,z_B)$ as its final treatment.
%Or
%%%%%%%%%%%%
Being aware of such plurality,
we qualify by Definitions \ref{def::CR} and \ref{def::SP} the particular `$2^2$ completely randomized (\textsc{C-R}) design' and `$2^2$ split-plot ({\SP}) design' on which we will base most of the quantitative derivations in this article.

\begin{definition}%[$2^2$ completely randomized design]
\label{def::CR}
Given treatments 1 to 4 in a $2^2$ factorial experiment and $N$ experimental units,
a $2^2$ completely randomized design with planned treatment arm sizes $N_1$, $N_2$, $N_3$, and $N_4= N - \sum_{k=1}^3 N_k$ can be visualized as distributing a well shuffled deck of $N_1$ tags of treatment 1, $N_2$ tags of treatment 2, $N_3$ tags of treatment 3, and $N_4$ tags of treatment 4 to units $1$ to $N$, such that all partitions of the $N$ units into the four treatment arms are equally likely.
%\begin{enumerate}
%\item chooses $N_1$ of the $N$ units at complete random, and assigns them to treatment 1,
%\item chooses $N_2$ of the remaining $N-N_1$ units at complete random, and assigns them to treatment 2,
%\item chooses $N_3$ of the remaining $N-N_1-N_2$ units at complete random, and assigns them to treatment 3,
%\item assigns the remaining $N_4$ units to treatment 4.
%\end{enumerate}
\end{definition}

\begin{lemma}
\label{lem::covZ_CR}
Under the $2^2$ completely randomized design qualified by Definition \ref{def::CR},
the sampling expectation and variance-covariance matrix of the assignment vector $\bm Z^*$ are
\begin{eqnarray*}
\E_{\CR}\left( \bm Z^* \right) \,=\,  N^{-1} \vones_{4N}\,,\quad
\cov_{\CR} \left( \bm Z^* \right)  \,=\, \mCf \otimes \mP_N
\end{eqnarray*}
where
\begin{equation*}
 \mCf  \,=\,
\frac{1}{N(N-1)}\left( \textnormal{diag}\left\{ \frac{N}{N_1}, \frac{N}{N_2}, \frac{N}{N_3}, \frac{N}{N_4}\right\} - \mJ_4\right).
\end{equation*}
%When the design is balanced, the coefficient matrix simplifies to
%$\mCf =4\mP_4/\{N(N-1)\}$.
\end{lemma}

\smallskip

\begin{definition}%[$2^2$ split-plot design]
\label{def::SP}
Given two $2$-level factors of interest, whole-plot factor $A$ and sub-plot factor $B$, and $N$ experimental units nested within $W$ whole-plots (blocks), each of size $M=N/W$,
a $2^2$ split-plot design with planned size parameters $W_{+1}$ and $M_{+1}$ consists of two separate randomizations:
\begin{itemize}
\item Whole-plot randomization that assigns $W_{+1}$ of $W$ whole-plots chosen at complete random to $+1_A$ level of  whole-plot factor $A$, and the remaining $W_{-1} = W-W_{+1}$ ones to $-1_A$ level,
%%%%%%%%%%%%%%%%%%%%%%%%%%%%%%%%%%%%%%
\item Sub-plot randomization that assigns $M_{+1}$ of $M$ sub-plots chosen at complete random within each whole-plot to $+1_B$ level of sub-plot factor $B$, and the remaining $M_{-1} = M - M_{+1}$ ones to $-1_B$ level.
\end{itemize}
The final treatment for sub-plot $(wm)$ will be the combination of the level of factor $A$ whole-plot $w$ receives in the whole-plot randomization and the level of factor $B$ itself receives in the sub-plot randomization.
\end{definition}
%%%%%%%%%%%%%%%%%%%%%%%%%%%%%%
%%%%%%%%%%%%%%%%%%%%%%%%%%%%%%
We will use `whole-plot' and `block,' as well as `sub-plot' and `experimental unit,' interchangeably for the rest of the paper, so that the notations and definitions introduced in Section \ref{sec::scienceWithBlock} apply directly.
Let
\begin{center}
$
r_A \,=\, {W_{+1}}/{W_{-1}}\,, \,\,\,r_B ={M_{+1}}/{M_{-1}}
$
\end{center}
be the ratios of factor arm sizes for the whole-plot and sub-plot randomizations respectively.
\begin{theorem}
\label{lem::covZ_SP}
Under the $2^2$ split-plot design qualified by Definition \ref{def::SP}, the sampling expectation and variance-covariance matrix of the assignment vector $\bm Z^*$ are
\begin{eqnarray*}
\E_{\textsc{s-p}}\left( \bm Z^* \right) \,=\,  N^{-1} \vones_{4N}\,,\quad
\cov_{\textsc{s-p}} \left( \bm Z^* \right)  \,=\, \mCb \otimes \mPb + \mCi \otimes \mPi
\end{eqnarray*}
where
\begin{align*}
%\label{eq::mCb}
\mCb
&=\,
\frac{1}{N(W-1)}\begin{pmatrix}
r_A &r_A &-1&-1\\
r_A &r_A &-1&-1\\
-1& -1& r^{-1}_A& r^{-1}_A\\
-1& -1& r^{-1}_A& r^{-1}_A \end{pmatrix}
\,,\\
%%%%%%%%%%%%%%%%%
%\label{eq::mCi}
\mCi
&=\,
\frac{1}{NW(M-1)} \begin{pmatrix}
(1+r_A)r_B & - (1+r_A) & 0&0\\
- (1+r_A) &(1+r_A)r^{-1}_B &0&0\\
0& 0& (1+r^{-1}_A)r_B & - (1+r^{-1}_A)\\
0& 0& - (1+r^{-1}_A) &(1+r^{-1}_A)r^{-1}_B \end{pmatrix}.
\nonumber
\end{align*}
\end{theorem}

%{\color{red}Due to its length and complexity,
%we relegate the proof of Theorem \ref{lem::covZ_SP} to the online supplementary material. INVOLVES REPEATED USE OF THE LAW OF ITERATED EXPECTATIONS}

%%%%%%%%%%%%%%%%%%%%%

\smallskip

The whole-plot and sub-plot randomizations in Definition \ref{def::SP} are essentially two independent complete randomizations.
The resulting $2^2$ split-plot design can hence be thought of as a \emph{restricted completely randomized design} \citep{Bailey1983} in the sense that all possible assignments are equally likely.
%--- by subjecting a \emph{completely randomized design} with planned treatment arm sizes
% and \emph{the level of factor $A$ restricted to be the same within each whole-plot}.
%subject
%to the restriction that \emph{all sub-plots within the same whole-plot receive the same level of factor $A$}.
%%%%%%%%%%%%%%%%%%%%%%
Refer to the $2^2$ completely randomized design with the same planned treatment arm sizes \begin{equation}
\label{eq::sizes}
(N_1, N_2, N_3, N_4) \,=\, (W_{-1}M_{-1},W_{-1}M_{+1}, W_{+1}M_{-1},W_{+1}M_{+1})
\end{equation}
as its `(unrestricted) completely randomized counterpart.'
It follows from straightforward algebra that the respective coefficient matrices of the restricted and the unrestricted satisfy
\begin{equation*}
%\label{eq::vsC}
\mCf \,=\, \frac{W-1}{N-1} \mCb + \frac{W(M-1)}{N-1}\mCi\,.
\end{equation*}
This, together with Lemma \ref{lem::covZ_CR} and Theorem \ref{lem::covZ_SP}, allows us to write the effect of `restriction' on the variance-covariance matrix of $ Z^*$ as
\begin{align*}
%\label{expansionCovZ}
&\cov_{\textsc{s-p}} ( \bm Z^*)  - \cov_{\CR}( \bm Z^* ) \\
&\quad=\, \mCb \otimes \mPb + \mCi \otimes \mPi - \mCf \otimes \mP_N\nonumber\\
&\quad=\, \mCb \otimes \left( \mPb- \frac{W-1}{N-1} \mP_N\right) + \mCi \otimes \left(\mPi - \frac{W(M-1)}{N-1}\mP_N \right).
\end{align*}

%%%%%%%%%%%%%%%%%%%%
%%%%%%%%%%%%%%%%%%%%
% -------------------------------------------------------------------------------------------------
\section{Neymanian point estimates for $2^2$ factorial effects}\label{section::var}
% -----------------------------------------------
%%%%%%%%%%%%%%%%%%%%%%
Neymanian causal inference focuses on the population-level effects, and takes
 the three  population average factorial effects as its chief causal estimands of interest.
We define in this section the Neymanian point estimates of these three estimands, and derive their respective sampling variances under $2^2$ split-plot designs.

%%%%%%%%%%%%%%%%%%%%%%%%%%%%
% =======================================
\subsection{Point estimates and their sampling variances}
% =======================================
Recall that $T_i = k$ if unit $i$ is assigned to treatment $k$.
Let
\begin{center}
$
%\label{eq::yobsk}
\bar{Y}^\obs(k)
\,=\,
{N^{-1}_k} \sum_{i: T_i =k} Y_{i}^\obs
%\,=\,{N^{-1}_k} \sum_{i: \text{$Y_i(k)$ is observed}} Y_{i}(k)
%\,=\, \bm Y(k)^{\T}\left\{ {N^{-1}_k}{\bm Z}(k)\right\}
$
\end{center}
be the average observed outcome of treatment arm $k$.
Estimating the unobservable
$\bar Y(k)$ by $\bar{Y}^\obs(k)$  in the definition of $\tau_F$  in \eqref{eq::estimands} yields the \emph{Neymanian point estimate} of this population-level factorial effect:
%%%%%%%%%%%%%%%%
%%%%%%%%%%%%%%%%%%%
%Estimating the unobservable $\bar Y(k)$ in the definition of $\tau_F$  in \eqref{eq::estimands} by
%$$
%%\label{eq::yobsk}
%\bar{Y}^\obs(k)
%\,=\,
%{N^{-1}_k} \sum_{i: T_i =k} Y_{i}^\obs
%\,=\,{N^{-1}_k} \sum_{i: \text{$Y_i(k)$ is observed}} Y_{i}(k)
%%\,=\, \bm Y(k)^{\T}\left\{ {N^{-1}_k}{\bm Z}(k)\right\}
%$$
%yields the \emph{Neymanian point estimate} of this population-level factorial effect:
\begin{equation}
\label{eq::estimator}
\widehat{\tau}_F \,=\, 2^{-1}\vg_F^{\T}(\bar{Y}^\obs(1) , \bar{Y}^\obs(2) , \bar{Y}^\obs(3) , \bar{Y}^\obs(4) )^\T \quad (F \in \mathcal{F})\,.
\end{equation}
%%%%%%%%%%%
%%%%%%%%%%
%The stochastic element in \eqref{eq::estimator} is the observed averages $\bar{Y}^\obs(k)$, which varies from randomization to randomization unless $Y_i(1) = Y_i(2) = Y_i(3) = Y_i(4)$ for all $i$.
%%%%%%%%%%%%%%%%%%%%%%%%%%

Let $\widetilde{\mY}$ be the $4N\times 4$ block-diagonal matrix with diagonal vectors $\bm Y(k)$:
\begin{center}
$\widetilde{\mY}= \begin{pmatrix}
\bm Y(1) & &&\\
& \bm Y(2)&&\\
&& \bm Y(3)&\\
&&&\bm Y(4)
\end{pmatrix}$.
\end{center}
%Now that $Y_{i}^\obs = Y_{i}(T_i)$, we have
It follows from
$$
\bar{Y}^\obs(k)
\,=\,
{N^{-1}_k} \sum_{i: T_i =k} Y_{i}^\obs
\,=\,{N^{-1}_k} \sum_{i: T_i = k} Y_{i}(k)
\,=\, \bm Y(k)^{\T}\{ {N^{-1}_k}{\bm Z}(k)\}\,,
$$
that
\begin{align*}
%\bar{Y}^\obs(k)
%&=\,
%{N^{-1}_k} \sum_{i: T_i =k} Y_{i}^\obs
%\,=\,{N^{-1}_k} \sum_{i: T_i = k} Y_{i}(k)
%\,=\, \bm Y(k)^{\T}\{ {N^{-1}_k}{\bm Z}(k)\}\,,\\
\begin{pmatrix}
\bar{Y}^\obs(1) \\
\bar{Y}^\obs(2) \\
\bar{Y}^\obs(3) \\
\bar{Y}^\obs(4)
\end{pmatrix}
&=
\begin{pmatrix}
\bm Y(1)^{\T} & &&\\
& \bm Y(2)^{\T}&&\\
&& \bm Y(3)^{\T}&\\
&&&\bm Y(4)^{\T}
\end{pmatrix}
\begin{pmatrix}
{N^{-1}_1}{\bm Z}(1) \\{N^{-1}_2}{\bm Z}(2) \\{N^{-1}_3}{\bm Z}(3) \\{N^{-1}_4}{\bm Z}(4)
\end{pmatrix}
= \widetilde{\mY}^\T\bm{Z}^*.
\end{align*}
Substitute this into \eqref{eq::estimator} to see
\begin{equation}
\label{eq::zongzong}
\widehat{\tau}_F \,=\, 2^{-1}\vg_F^{\T} \widetilde{\mY}^\T\bm{Z}^* \quad (F \in \mathcal{F})\,,
\end{equation}
with assignment vector $\bm{Z}^*$ alone being stochastic on the right.
%%%%%%%%%%%%%%%%%%%%%%%%%%
The fact of \eqref{eq::zongzong} being true for any arbitrary $\bm{Z}^*$
allows us to take expectation and covariance of both sides with respect to  any arbitrary $2^2$ factorial assignment mechanism. This yields Lemma \ref{lem::zong}.
%Taking expectation and covariance of both sides of this identity \eqref{eq::zongzong}.
%%%%%%%%%%%%%%%%%%%%%%%%%%%%%%%%
%%%%%%%%%%%%%%%%%%%%%%%%%%%%%%%%
\begin{lemma}
\label{lem::zong}
The randomness in the Neymanian point estimate $\widehat{\tau}_F$,
under any arbitrary $2^2$ factorial assignment mechanism \textsc{(a-m)}, originates solely from the randomness in the assignment vector $\bm Z^*$, with
\begin{equation*}
\E_{\textsc{a-m}}(\widehat{\tau}_F) \,=\,  2^{-1}\vg_F^{\T} \Yt^{\T}\E_{\textsc{a-m}}( \bm Z^*)\,,
\,\,\,
\var_{\textsc{a-m}}(\widehat{\tau}_F) \,=\, 4^{-1} \vg_F^{\T} \Yt^{\T}\cov_{\textsc{a-m}}( \bm Z^*)\Yt \vg_F
\end{equation*}
where $\E_{\textsc{a-m}}$, $\var_{\textsc{a-m}}$, and $\cov_{\textsc{a-m}}$ are the expectation, variance, and covariance with respect to the sampling distribution under {\textsc{a-m}} over all possible assignments.
\end{lemma}

Explicit formulae under completely randomized designs follow immediately from combining Lemma \ref{lem::zong} with Lemma \ref{lem::covZ_CR}, and those under split-plot designs from combining Lemma \ref{lem::zong} with Theorem \ref{lem::covZ_SP}:
%Explicit formulae under completely randomized and split-plot designs follow immediately from substituting Lemma \ref{lem::covZ_CR}  and Theorem \ref{lem::covZ_SP}  into Lemma \ref{lem::zong}.
%applying Lemma \ref{lem::covZ_CR} and Theorem \ref{lem::covZ_SP} together with Lemma \ref{lem::zong}.

\begin{theorem}
\label{thm::factorial}
Under the $2^2$ completely randomized design qualified by Definition \ref{def::CR},
the Neymanian point estimate $\widehat{\tau}_F$ is unbiased for ${{\tau}}_F$ with sampling variance
\begin{equation}
\label{var_CR}
\var_\CR(\widehat{\tau}_F)
\,=\, 4^{-1}(N-1) \vg_F^{\T} \left(\mCf \circ \mS \right) \vg_F\quad(F \in \mathcal{F})\,.
\end{equation}
Here, `$\circ$' denotes the entrywise product, and $\mCf$ is the coefficient matrix defined in Lemma \ref{lem::covZ_CR}.
\end{theorem}

\begin{theorem}
\label{thm::theoretical-variances}
Under the $2^2$ split-plot design qualified by Definition \ref{def::SP},
the Neymanian point estimate $\widehat{\tau}_F$ is unbiased for ${{\tau}}_F$
with sampling variance
\begin{align}
\label{eq::theoreticalVar}
\var_{\textsc{s-p}}( \widehat{\tau}_F)
%&=\, 4^{-1} \vg_F^{\T} \left\{(W-1)M\, \mCb \circ \mSb + W(M-1)\mCi \circ \mSi \right\} \vg_F\\
&=\,  4^{-1}(W-1)M\,\vg_F^{\T}(\mCb \circ \mSb)\vg_F\\
&\quad + 4^{-1}W(M-1)\, \vg_F^{\T}(\mCi \circ \mSi)\vg_F\quad(F \in \mathcal{F})\,.\nonumber
\end{align}
\end{theorem}

%Combining Lemma \ref{lem::zong} with Theorem \ref{lem::covZ_SP}  yields the formulas under split-plot designs.
%%%%%%%%%%%%%%%%%%%%%

%%%%%%%%%%%%%%%%%%%%%%%%
\subsection{Comparison of precisions under strict additivity}
%%%%%%%%%%%%%%%%%%%%%%%%
Simplified forms of Theorems \ref{thm::factorial} and \ref{thm::theoretical-variances} are available when the potential outcomes are strictly additive, enabling intuitive comparisons of the estimation precision.

\begin{corollary}
\label{corollary::strictAdditive}
For strictly additive potential outcomes,
the sampling variances of $\widehat{\tau}_A$, $\widehat{\tau}_B$, and $\widehat{\tau}_{AB}$ under the $2^2$ split-plot design in Theorem \ref{thm::theoretical-variances} reduce to
\begin{equation}
\label{eq::strictAdd}
\begin{array}{l}
\var_{\textsc{s-p}}\left( \widehat{\tau}_A \right)
\,=\, W^{-1}\gamma_A S^2_{\textnormal{btw}}
+
(4N)^{-1}\gamma_A(\gamma_B-4)S^2_{\textnormal{in}}\,,\\
%%%%%%%%%%%%%%%%%%%%
\var_{\textsc{s-p}}\left( \widehat{\tau}_B \right)
\,=\,
\var_{\textsc{s-p}}\left( \widehat{\tau}_{AB} \right)
\,=\,
(4N)^{-1} \gamma_A\gamma_B S^2_{\textnormal{in}}\,,
\end{array}
\end{equation}
where $\gamma_A=r_A +r_A^{-1}+2$, and $\gamma_B = r_B+r_B^{-1}+2$.
\end{corollary}

\begin{remark}
With $x +x^{-1}+2 = \{\sqrt{x} -(\sqrt{x})^{-1}\}^2 +4$,
we have $\min_{r_A}\gamma_A= \gamma_A \big|_{r_A=1} = 4$ and $ \min_{r_B}\gamma_B= \gamma_B \big|_{r_B = 1} = 4$.
The increasing monotonicity of \eqref{eq::strictAdd} in $\gamma_A$ and $\gamma_B$ suggests
the three sampling variances be simultaneously minimized when $\gamma_A$ and $\gamma_B$ are at their respective minimums:
 \begin{align*}
\min_{\gamma_A,\gamma_B}\var_{\textsc{s-p}}( \widehat{\tau}_A )
&= \var_{\textsc{s-p}}( \widehat{\tau}_A ) \big|_{\gamma_A=4, \gamma_B=4}
=4S^2_{\textnormal{btw}}/W,\,\,\\
%%%%%%%%%%%%%%%%%%%%
\min_{\gamma_A,\gamma_B}\var_{\textsc{s-p}}( \widehat{\tau}_B)
&\left(=
\min_{\gamma_A,\gamma_B}\var_{\textsc{s-p}}( \widehat{\tau}_{AB})  \right)
= \var_{\textsc{s-p}}( \widehat{\tau}_B ) \big|_{\gamma_A=4, \gamma_B=4}
=
4 S^2_{\textnormal{in}}/N
\end{align*}
where $\gamma_A =4, \gamma_B = 4$ imply $r_A = r_B = 1$ --- i.e. the design being balanced.
This  establishes the optimality of balanced designs regarding strictly additive potential outcomes.%when the potential outcomes are strictly additive.
\end{remark}

\begin{remark}
\label{remark::BlockEffect}
The sampling variances of $\widehat{\tau}_A$ and $\widehat{\tau}_B$ in \eqref{eq::strictAdd} satisfy
\begin{equation}
\label{diffInPrecision}
\var_{\textsc{s-p}}( \widehat{\tau}_A)
-
\var_{\textsc{s-p}}( \widehat{\tau}_B) \,=\, W^{-1} \gamma_A ( S^2_{\textnormal{btw}} - S^2_{\textnormal{in}}/M)\,.
\end{equation}
This suggests more precise Neymanian estimation of the sub-plot factor $B$ than that of the whole-plot factor $A$ if $S^2_{\textnormal{btw}}-S^2_{\textnormal{in}} /M> 0$, and
vice versa if $S^2_{\textnormal{btw}}-S^2_{\textnormal{in}}/M < 0$.

An intuitive link between the discriminant $S^2_{\textnormal{btw}}-S^2_{\textnormal{in}} /M$ and the block heterogeneity can be established from a super-population perspective for potential outcomes generated from linear mixed effects models. Specifically,  assume the study population in question to be a random sample from some super-population such that
\begin{equation}
\label{eq::intuition_lme}
Y_{(wm)}(k)
\,=\, \mu(k) + \eta_w + \xi_{(wm)}\quad (w = 1, \ldots, W; m = 1,\ldots, M)
\end{equation}
follow the linear mixed effects model with  fixed treatment effects  $\mu(k)$, random block effects $\eta_w\overset{\text{iid}}{\sim} \mathcal{N}(0, \sigma_\eta^2)$, and individual sampling errors $\xi_{(wm)}\overset{\text{iid}}{\sim} \mathcal{N}(0, \sigma_\xi^2)$ jointly independent of $\eta_w$.
%{\color{red}Specifically,  assume the potential outcomes of the study population constitute a random sample from super-population generative model
%\begin{equation}
%\label{eq::intuition_lme}
%Y_{(wm)}(k)
%\,=\, \mu(k) + \eta_w + \xi_{(wm)}\quad (w = 1, \ldots, W; m = 1,\ldots, M)\,,
%\end{equation}
%by thinking of  with $\eta_w\overset{\text{iid}}{\sim} \mathcal{N}(0, \sigma_\eta^2)$ and $\xi_{(wm)}\overset{\text{iid}}{\sim} \mathcal{N}(0, \sigma_\xi^2)$  being jointly independent sampling errors at block- and unit-levels,}
%{\color{red}Without loss of generality, assume super-population generative model}
%\begin{equation}
%\label{eq::intuition_lme}
%Y_{(wm)}(1)
%\,=\, \eta_w + \xi_{(wm)}\quad (w = 1, \ldots, W; m = 1,\ldots, M)\,,
%\end{equation}
%where $\eta_w\overset{\text{iid}}{\sim} \mathcal{N}(0, \sigma_\eta^2)$ and $\xi_{(wm)}\overset{\text{iid}}{\sim} \mathcal{N}(0, \sigma_\xi^2)$  are jointly independent.

Assume, without loss of generality, $\mu(1) = 0$.
The $W$ block average potential outcomes under treatment 1 constitute $W$ iid normals with mean 0 and variance $\sigma_\eta^2 +\sigma_\xi^2/{M}$:
$$
Y_{(w)}(1)  \,=\, M^{-1} \sum_{m=1}^M Y_{(wm)}(1) \,=\,  \eta_w + M^{-1} \sum_{m=1}^M \xi_{(wm)}\overset{\text{iid}}{\sim} \mathcal{N}(0, \sigma_\eta^2 +\sigma_\xi^2/{M})\,.
$$
$S^2_{\textnormal{btw}}(1,1)$,  as the finite-population variance of $Y_{(w)}(1)$, is thus unbiased for the super-population variance parameter $\sigma_\eta^2 +\sigma_\xi^2/{M}$:
\begin{equation}
\label{eq::eureka_btw}
E^*\{ S^2_{\textnormal{btw}}(1,1)\} \,=\, \var^*\{Y_{(w)}(1) \}\,=\, \sigma_\eta^2 +\sigma_\xi^2/{M}
\end{equation}
where $E^*$ and $\var^*$ are the expectation and variance with respect to the sampling distribution represented via model \eqref{eq::intuition_lme}.
%%%%%%%%%
%%%%%%%%%%
Likewise, with
\begin{align*}
S^2_{(w)}(1,1) &=\,  (M-1)^{-1} \sum_{m=1}^{M}\left\{Y_{(wm)}(1) - Y_{(w)}(1) \right\}^2 \\
&=\, (M-1)^{-1}\sum_{m=1}^{M} \left( \xi_{(wm)} -  M^{-1} \sum_{m=1}^M \xi_{(wm)}\right)^2
\end{align*}
simplifying to the finite-population variance of iid normals $\{\xi_{(wm)}\}_{m=1}^M$, we have
$E^*\{S^2_{(w)}(1,1)\} = E^*(\xi_{(wm)})=\sigma_{\xi}^2$,
%E^*\left[(M-1)^{-1}\sum_{m=1}^{M} \left\{Y_{(wm)}(1) - Y_{(w)}(1) \right\}^2\right]
%\,=\, E^*\left\{(M-1)^{-1}\sum_{m=1}^{M} \left( \xi_{(wm)} - M^{-1} \sum_{m=1}^M \xi_{(wm)}\right)^2\right\}
%\,=\, \sigma_{\xi}^2\,.
%$
and thus
\begin{equation}
%E^*\{S^2_{(w)}(1,1)\} &=\, E^*(\xi_{(wm)}) \,=\,\sigma_{\xi}^2\,,  \nonumber\\
E^*\{S^2_{\textnormal{in}}(1,1)\}\,=\,E^*\left\{{W}^{-1}\sum_{w=1}^W S^2_{(w)}(1,1)\right\}
%\,=\, {W}^{-1}\sum_{w=1}^W E^*\{S^2_{(w)}(1,1)\}
\,=\, \sigma_{\xi}^2\,.\label{eq::eureka_in}
\end{equation}
%%%%%%%%%%%%%%%%%%%%%%%%%%%
%%%%%%%%%%%%%%%%%%%%%%%%%%%
Under strict additivity --- as it is guaranteed by model \eqref{eq::intuition_lme},  abbreviate $S^2_{\textnormal{btw}}(1,1)$ as $S^2_{\textnormal{btw}}$ and $S^2_{\textnormal{in}}(1,1)$ as $S^2_{\textnormal{in}}$ (by summoning Lemmas \ref{l:SPstrictadd} and \ref{def2}).
Formulae \eqref{eq::eureka_btw} and \eqref{eq::eureka_in} together yield
\begin{align}
\label{eq::RMK}
E^*(S^2_{\textnormal{btw}} - S^2_{\textnormal{in}}/{M})
\,=\,
E(S^2_{\textnormal{btw}}) -  E(S^2_{\textnormal{in}})/M  \,=\,\sigma_\eta^2\,\,\geq\, 0\,,
\end{align}
equating the  super-population expectation of the discriminant to the super-population variance of the random block effects $\eta_w$.
This, coupled with formula \eqref{diffInPrecision}, suggests the average sampling variance of the sub-plot estimate $\widehat{\tau}_B$ be strictly smaller than that of the whole-plot estimate $\widehat{\tau}_A$ --- unless $\sigma_\eta^2=0$ and \eqref{eq::intuition_lme} degenerates to a simple linear model that admits no random block effects.
\end{remark}

Recall from Theorem \ref{thm::theoretical-variances}  and Corollary \ref{corollary::strictAdditive} the decomposition of overall sampling variances under split-plot designs into the between- and within-whole-plot parts.
%Recall that Theorem \ref{thm::theoretical-variances}  and Corollary \ref{corollary::strictAdditive} decompose the sampling variance of $\widehat{\tau}_F$ under split-plot designs into two parts:
%(i) that between the whole-plots,  with coefficient matrix $\mCb$, and
%(ii) that within each whole-plot, with coefficient matrix $\mCi$.
Analogous results for completely randomized designs follow from substituting Theorem \ref{thm::decompS} into formula \eqref{var_CR}:
\begin{align}
\label{eq::expansion_CR}
\var_\CR(\widehat{\tau}_F)
&=\, 4^{-1} (W-1)M\,\vg_F^{\T}(\mCf \circ \mSb)\vg_F\\
&\quad + 4^{-1}W(M-1)\, \vg_F^{\T}(\mCf \circ \mSi)\vg_F\quad(F \in \mathcal{F})\nonumber\,.
\end{align}
%with $\mCf$ being the common coefficient matrix for both parts.
%This, coupled with Theorem \ref{thm::theoretical-variances}, allows us to quantify the difference in sampling variance between a $2^2$ split-plot design and its completely randomized counterpart in the most explicit way.
Contrasting this with Theorem \ref{thm::theoretical-variances} yields Corollary \ref{corollary::diff}.

\begin{corollary}
\label{corollary::diff}
Assume common treatment arm sizes \eqref{eq::sizes}.
The sampling variance of $\widehat{\tau}_F$ under a $2^2$ split-plot design \textsc{(s-p)}
differs from that under a $2^2$ completely randomized design \textsc{(c-r)} by
\begin{equation*}
\var_{\textsc{s-p}}\left( \widehat{\tau}_F \right) - \var_\CR(\widehat{\tau}_F) \,=\,C_0 \,\vg_F^{\T}\left\{ (\mCb- \mCi) \circ\left(\mSb - {\mSi}/{M}\right) \right\}\vg_F\,,
\end{equation*}
where $C_0$ is a positive constant.
\end{corollary}

The difference in Corollary \ref{corollary::diff} informs us of not only the relative efficiency of split-plot designs with regard to each $F\in \mathcal{F}$, but also the discrepancy in variance estimation when a split-plot experiment is wrongfully analyzed as a completely randomized one.

\begin{corollary}
\label{corollary::strictAdditive_CR}
For strictly additive potential outcomes,
the sampling variance under $2^2$ completely randomized design in \eqref{eq::expansion_CR} reduces to
\begin{equation}
\label{eq::var_CR_strictAdd}
\var_\CR(\widehat{\tau}_F)
\,=\,
\frac{\gamma_A\gamma_B }{4(N-1)}\left(
\frac{W-1}{W}  S^2_{\textnormal{btw}} +
\frac{M-1}{M}  S^2_{\textnormal{in}}\right)\quad (F \in \mathcal{F})\,.
\end{equation}
\end{corollary}

\begin{corollary}
\label{corollary::intuition}
For strictly additive potential outcomes,
the differences in Corollary \ref{corollary::diff} reduce to
\begin{align*}
\var_{\textsc{s-p}}( \widehat{\tau}_A )
-
\var_{\CR}( \widehat{\tau}_A )
&= C_1 (  S^2_{\textnormal{btw}}-{S^2_{\textnormal{in}}}/{M}), \\
%%%%%%%%%%%%%%%%%%%
\var_{\textsc{s-p}}( \widehat{\tau}_B )
-
\var_{\CR}( \widehat{\tau}_B)
&= \var_{\textsc{s-p}}( \widehat{\tau}_{AB})
-
\var_{\CR}( \widehat{\tau}_{AB})
=
- C_2( S^2_{\textnormal{btw}}- {S^2_{\textnormal{in}}}/{M}),
\end{align*}
where $C_1$ and $C_2$ are two positive constants.
\end{corollary}
%%%%%%%%%%%%%%%%%%%%%%%%%
% COROLLARY: Comparison of variances
%%%%%%%%%%%%%%%%%%%%%%%%%
%\begin{corollary}
%\label{corollary::comparison}
%When the \textsc{pom} is strictly additive,  the respective randomization-based variances of $\widehat{\tau}_F$ under a $2^2$ split-plot design and its completely randomized counterpart satisfy
%
%\end{corollary}
%%%%%%%%%
%Thus, for strictly additive potential outcomes,
%we have $
%\var_{\textsc{s-p}}( \widehat{\tau}_A )
%>
%\var_{\CR}( \widehat{\tau}_A )$,
%%%%%%%%%%%%%%%%%%%%
%$\var_{\textsc{s-p}}( \widehat{\tau}_B )
%<
%\var_{\CR}( \widehat{\tau}_B)$, and
%$\var_{\textsc{s-p}}( \widehat{\tau}_{AB})
%<
%\var_{\CR}( \widehat{\tau}_{AB})
%$
%%a split-plot design yields less precise $\widehat{\tau}_A$ for whole-plot factor $A$ yet more precise $\widehat{\tau}_B, \widehat{\tau}_{AB}$ for sub-plot factor $B$ and interaction than its completely randomized counterpart
%if $S^2_{\textnormal{btw}}-S^2_{\textnormal{in}} /M> 0$, and
%vice versa if $S^2_{\textnormal{btw}}-S^2_{\textnormal{in}}/M < 0$.

With the same discriminant $S^2_{\textnormal{btw}}-S^2_{\textnormal{in}}/M$ as that in \eqref{diffInPrecision}, the intuition from Remark \ref{remark::BlockEffect} translates into Corollary \ref{corollary::intuition} with almost no need for change:
Assume super-population model \eqref{eq::intuition_lme},
it follows from
%\begin{center}
$E^*(S^2_{\textnormal{btw}} - S^2_{\textnormal{in}}/{M}) \,=\,\sigma_\eta^2\,\,\geq\, 0$
%\end{center}
in formula  \eqref{eq::RMK} %in Remark \ref{remark::BlockEffect} ---
 that
\begin{center}
$E^*\{\var_{\textsc{s-p}}( \widehat{\tau}_A ) \}
\geq
E^*\{\var_{\CR}( \widehat{\tau}_A ) \}\,, \quad E^*\{\var_{\textsc{s-p}}( \widehat{\tau}_B ) \}
\leq
E^*\{\var_{\CR}( \widehat{\tau}_B ) \}\,.
$
\end{center}
The inequalities are strict unless $\sigma^2_\eta=0$, in which case \eqref{eq::intuition_lme} degenerates to a simple linear model that admits no random block effects.

%%%%%%%%%%%%%%%%%%%%%%%%%%%%%
%Consider the toy \textsc{pom} in which all four columns equal the $\vY(1)$ in Table \ref{t:Potential outcome}. The strict additivity assumption is satisfied with $S^2_{\textnormal{in}} = 0.43$ and $S^2_{\textnormal{btw}} = 2.53$.
%Given $S^2_{\textnormal{btw}} - S^2_{\textnormal{in}}/S > 0$,
%we would anticipate
%$\var_{\textsc{s-p}}\left( \widehat{\tau}_A \right)
%>
%\var_{\CR}\left( \widehat{\tau}_A \right)$,
%$\var_{\textsc{s-p}}\left( \widehat{\tau}_B \right) <\var_{\CR}\left( \widehat{\tau}_B \right)$, and $\var_{\textsc{s-p}}\left( \widehat{\tau}_{AB} \right) <\var_{\CR}\left( \widehat{\tau}_{AB} \right)$.
%The numerical results computed from \eqref{eq::strictAdd} and \eqref{eq::strictAdditive_CR} corroborate our `educated guess':
%\begin{center}
%\begin{tabular}{l|p{1.2cm}p{1.2cm}p{1.2cm}}\hline
%& $\var(\widehat{\tau}_A)$ & $\var(\widehat{\tau}_B)$ & $\var(\widehat{\tau}_{AB})$   \\\hline
%Split-plot  & 2.53 & 0.21 & 0.21\\
%Completely randomized & 1.21  & 1.21&1.21
%\\\hline
%\end{tabular}
%\end{center}
%
%
%

%%%%%%%%%%%%%%%%%%%%%%%%%
\subsection{Simplified expressions under balanced designs}
%%%%%%%%%%%%%%%%%%%%%%%%%
%We present in this section some alternative expressions for the sampling variances under balanced designs in terms of .
%%The resulting paraphrases are particularly neat under balanced designs with
%%The paraphrasing allows for expressions of quite neat forms under balanced designs with $r_A=r_B=1$.
%
Recall $S^2_{F\text{-btw}}$ and $S^2_{F\text{-in}}$ from \eqref{eq::SinF_element} as the between- and within-block variances of $\tau_{(wm)\text{-}F}$.
Define analogously
$$
S^2_{\mu\text{-btw}} = \frac{1}{W-1} \sum_{w=1}^W ({\mu}_{(w)} - \mu)^2\,,\,\,
S^2_{\mu\text{-in}}= \frac{1}{W} \sum_{w=1}^W\left\{\frac{1}{M-1} \sum_{m=1}^M (\mu_{(wm)} - \mu_{(w)})^2 \right\}
$$
for $\mu_{(wm)} = 4^{-1} \sum_{k=1}^4 Y_{(wm)}(k)$,
with
$\mu_{(w)} =  M^{-1}\sum_{m=1}^M \mu_{(wm)}$ and
$\mu =
N^{-1} \sum_{(wm)} \mu_{(wm)}
$
being the block and population averages respectively.
%gives the average of all $4M$ potential outcomes in whole-plot $w$, and
%\begin{equation}
%\label{eq::mu}
%\mu = W^{-1}\sum_{w=1}^W\mu_{(w)} =N^{-1} \sum_{(wm)} \mu_{(wm)} = (4N)^{-1} \sum_{(wm), k}Y_{(wm)}(k)
%\end{equation}
%gives the average of all $4N$ entries in the potential outcome matrix $\mY$.
%\bigskip
%\bigskip
%%%%%%%%%%%%%%%%%%%%%%%%%
% COROLLARY: BALANCED
%%%%%%%%%%%%%%%%%%%%%%%%%
\begin{corollary}
\label{corollary::balanced}
Under a balanced $2^2$ split-plot design with $W_{-1}=W_{+1}$ and $M_{-1}=M_{+1}$,
the sampling variances % of $\widehat{\tau}_F$
in Theorem \ref{thm::theoretical-variances} reduce to
\begin{align*}
\var_{\textsc{s-p}}\left( \widehat{\tau}_A \right)
&=\,
4W^{-1} S^2_{\mu\text{-btw}} + N^{-1} (S^2_{B\text{-in}} + S^2_{AB\text{-in}})\,,\\
%%%%%%%%%%%%%%
\var_{\textsc{s-p}}( \widehat{\tau}_B)
&=\,
W^{-1} S^2_{AB\text{-btw}}  +
N^{-1} (4S^2_{\mu\text{-in}}+ S^2_{A\text{-in}})\,,\\
%%%%%%%%%%%%%
\var_{\textsc{s-p}}( \widehat{\tau}_{AB})
&=\,
W^{-1} S^2_{B\text{-btw}}  +
N^{-1}(4S^2_{\mu\text{-in}} + S^2_{A\text{-in}})\,.
\end{align*}
\end{corollary}

Analogous results for balanced complete randomizations follow from letting $N_1=N_2=N_3=N_4$
in \eqref{var_CR}:
$$
\var_\CR(\widehat{\tau}_F)
\,=\,
N^{-1} \vg_F^{\T} \left(\mP_4 \circ \mS \right) \vg_F
\,=\,
N^{-1} \sum_{k=1}^4 S^2(k,k) - N^{-1}S^2_F\quad (F \in \mathcal{F})\,.
$$
This is the exact form of Theorem 2 in \cite{Tir2015} when the number of factors equals two.
%%%%%%%%%%%%%%%%%%%%%%%%%
% COROLLARY: Balanced design, between-block (whole-plot) additive \textsc{pom}
%%%%%%%%%%%%%%%%%%%%%%%%%
\begin{corollary}
For within-block additive potential outcomes,
the sampling variances of $\widehat{\tau}_F$ under a balanced $2^2$ split-plot design reduce from Corollary \ref{corollary::balanced} to
\begin{align*}
\var_{\textsc{s-p}}\left( \widehat{\tau}_A \right)
&=\,
W^{-1} S^2_{\mu\text{-btw}}\,,\\
%%%%%%%%%%%%%%
\var_{\textsc{s-p}}\left( \widehat{\tau}_B \right)
&=\,
W^{-1} S^2_{AB\text{-btw}}  +
4N^{-1} S^2_{\mu\text{-in}}\,,\\
%%%%%%%%%%%%%
\var_{\textsc{s-p}}\left( \widehat{\tau}_{AB} \right)
&=\,
W^{-1} S^2_{B\text{-btw}}  +
4N^{-1} S^2_{\mu\text{-in}}\,.
\end{align*}
\end{corollary}

%%%%%%%%%%%%%%%%%%%%%%%%%
% COROLLARY: Balanced design, between-block (whole-plot) additive \textsc{pom}
%%%%%%%%%%%%%%%%%%%%%%%%%
\begin{corollary}
For between-block additive potential outcomes,
the sampling variances of $ \widehat{\tau}_B$ and $ \widehat{\tau}_{AB}$ under a balanced $2^2$ split-plot design reduce from Corollary \ref{corollary::balanced}  to
\begin{align*}
\var_{\textsc{s-p}}\left( \widehat{\tau}_B \right)
\,=\, \var_{\textsc{s-p}}\left( \widehat{\tau}_{AB} \right)
\,=\,
N^{-1} \left(4S^2_{\mu\text{-in}} + S^2_{A\text{-in}} \right).
\end{align*}
\end{corollary}

%%%%%%%%%%%%%%%%%%%%%%%%%%%%%%%%%
% ==============================================
\section{Estimating the sampling variances}
% ===============================================
\label{section::varHat}
%%%%%%%%%%%%%%%%%%%%%%%%%%%%%%%%%
The sampling variances by formula  \eqref{eq::theoreticalVar} are in practice unobservable.
We address in this section their estimation, and use the results %resulting estimates
to construct Neymanian interval estimates.

Recall from Definition \ref{def::SP} that the whole-plot randomization assigns $W_{-1}$ whole-plots to $-1_A$ level of factor $A$ and the rest $W_{+1}$ to $+1_A$ level.
Let
\begin{center}
$\mathcal{W}_{-1} = \{w: \text{whole-plot $w$ is assigned to $-1_A$ level}\}$,
\end{center}
\begin{center}
$\mathcal{W}_{+1} = \{w: \text{whole-plot $w$ is assigned to $+1_A$ level}\}$.
\end{center}
%%%%%%%%%%%%%%%%%
For each $w \in \mathcal{W}_{-1}$, whole-plot $w$ {\color{blue}ends up with --- maybe `sees'??} $M_{-1}$ of its $M$ sub-plots in treatment arm 1 and the rest $M_{+1}$ in treatment arm 2.
Define for such whole-plots
\begin{center}
$Y_{(w)}^\obs(1) = M_{-1}^{-1} \sum_{m:T_{(wm)}=1} Y_{(wm)}^\obs$,\quad$
Y_{(w)}^\obs(2) = M_{+1}^{-1} \sum_{m:T_{(wm)}=2} Y_{(wm)}^\obs$
\end{center}
as the sample versions of $Y_{(w)}(1)$ and $Y_{(w)}(2)$ respectively.
%,
%where $T_{(wm)}$ index the treatment assignment of sub-plot $(wm)$.
%%%%%%%%%%%%%%%%%
Assume $|\mathcal{W}_{-1}| = W_{-1} \geq 2$,
%%%%%%%%%%%%%
$$
s^2_{\textnormal{btw}}(k,l)\,=\,\left(W_{-1}-1\right)^{-1}\sum_{w \in \mathcal{W}_{-1} }
\{Y_{(w)}^\obs(k) - \bar{Y}^\obs(k) \}
\{Y_{(w)}^\obs(l) - \bar{Y}^\obs(l)\}\,,
$$
as the covariance of $Y_{(w)}^\obs(k)$ and $Y_{(w)}^\obs(l)$ over $w \in \mathcal{W}_{-1}$,
defines a sensible sample version of the between-whole-plot covariance
%
%$$Y_{(w)}^\obs(1) = M_{-1}^{-1} \sum_{m:T_{(wm)}=1} Y_{(wm)}^\obs,\quad
%Y_{(w)}^\obs(2) = M_{+1}^{-1} \sum_{m:T_{(wm)}=2} Y_{(wm)}^\obs\,,$$
%where $T_{(wm)}=k$ if sub-plot $(wm)$ receives treatment $k$,
%as the sample versions of $Y_{(w)}(1)$ and $Y_{(w)}(2)$, respectively.
%with the average of $Y_{(w)}^{\obs}(k)$ over all $w$ that receive $-1_A$ being $\bar{Y}^\obs(k)$, and that of $Y_{(w)}^{\obs}(l)$ being $Y_{(w)}^{\obs}(l)$,
$$S^2_{\textnormal{btw}}(k,l)
\,=\,
(W-1)^{-1} \sum_{w=1}^W
\{Y_{(w)}(k) - \bar{Y}(k) \}
\{Y_{(w)}(l) - \bar{Y}(l)\}
$$
 for  $k,l = 1,2$.
%%%%%%%%%%%%%%%%%%%%%
%%%%%%%%%%%%%%%%%%%%%%
Likewise, define
 \begin{center}
$Y_{(w)}^\obs(3) = M_{-1}^{-1} \sum_{m:T_{(wm)}=3} Y_{(wm)}^\obs$,\quad$
Y_{(w)}^\obs(4) = M_{+1}^{-1} \sum_{m:T_{(wm)}=4} Y_{(wm)}^\obs$
\end{center}
for each $w \in \mathcal{W}_{+1}$,
now that whole-plots in this set {\color{blue}end up with} $M_{-1}$ of its $M$ sub-plots in treatment arm 3 and the rest $M_{+1}$ in treatment arm 4.
The corresponding
$$
s^2_{\textnormal{btw}}(k,l)\,=\,\left(W_{+1}-1\right)^{-1}\sum_{w \in \mathcal{W}_{+1} }
\{Y_{(w)}^\obs(k) - \bar{Y}^\obs(k) \}
\{Y_{(w)}^\obs(l) - \bar{Y}^\obs(l)\}
$$
defines a sensible sample version of
$S^2_{\textnormal{btw}}(k,l)$ for $k,l \in \{3,4\}$.

\begin{lemma}
\label{lem::Es2}
Under the $2^2$ split-plot design qualified by Definition \ref{def::SP},
the sampling expectations of $s^2_{\textnormal{btw}}(k,l)$ satisfy
\begin{align*}
\E\begin{pmatrix}
s^2_{\textnormal{btw}}(1,1)&s^2_{\textnormal{btw}}(1,2) \\
s^2_{\textnormal{btw}}(2,1)&s^2_{\textnormal{btw}}(2,2)
\end{pmatrix} \,=\,\,&
\begin{pmatrix}
S^2_{\textnormal{btw}}(1,1) & S^2_{\textnormal{btw}}(1,2)\\
S^2_{\textnormal{btw}}(2,1)&S^2_{\textnormal{btw}}(2,2)
\end{pmatrix}\\
& +
M^{-1}
\begin{pmatrix}
r_B & - 1\\
- 1&r_B^{-1}
\end{pmatrix}
\circ
\begin{pmatrix}
S^2_{\textnormal{in}}(1,1) & S^2_{\textnormal{in}}(1,2)\\
S^2_{\textnormal{in}}(2,1)&S^2_{\textnormal{in}}(2,2)
\end{pmatrix},\\
%\end{align*}
%and the sampling expectations of $s^2_{\textnormal{btw}}(k,l)$ ( $(k,l) \in \{3,4\} \times \{3,4\}$)  satisfy
%\begin{align*}
\E\begin{pmatrix}
s^2_{\textnormal{btw}}(3,3)&s^2_{\textnormal{btw}}(3,4) \\
s^2_{\textnormal{btw}}(4,3)&s^2_{\textnormal{btw}}(4,4)
\end{pmatrix} \,=\,\,&
\begin{pmatrix}
S^2_{\textnormal{btw}}(3,3) & S^2_{\textnormal{btw}}(3,4)\\
S^2_{\textnormal{btw}}(4,3)&S^2_{\textnormal{btw}}(4,4)
\end{pmatrix}\\
& +
M^{-1}
\begin{pmatrix}
r_B & - 1\\
- 1&r_B^{-1}
\end{pmatrix}
\circ
\begin{pmatrix}
S^2_{\textnormal{in}}(3,3) & S^2_{\textnormal{in}}(3,4)\\
S^2_{\textnormal{in}}(4,3)&S^2_{\textnormal{in}}(4,4)
\end{pmatrix}.
\end{align*}
%\begin{align*}
%\E(\mathbf{s^\text{$2$}_{\textnormal{btw}}}) &=&
%\begin{pmatrix}
%1&1&0& 0 \\
%1&1&0& 0 \\
%%%%%%%%%
%0& 0& 1&1\\
%0& 0& 1&1
%\end{pmatrix}
%\circ \mathbf{S^\text{$2$}_{\textnormal{btw}}}
%+
%M^{-1}
%\begin{pmatrix}
%r_B&-1&0&0\\
%-1&r_B^{-1}&0&0\\
%0&0&r_B&-1\\
%0&0&-1&r_B^{-1}
%\end{pmatrix}
%\circ \mathbf{S^\text{$2$}_{\textnormal{in}}}\,.
%\end{align*}
%%%%%%%%%%%%%%%%%%%%%%%%%%
%% Balanced case
%%%%%%%%%%%%%%%%%%%%%%%%%
%When the design is balanced, \eqref{eq::s2btw} reduces to
%\begin{eqnarray*}
%\E_{\textsc{s-p}}\left\{ s^2_{\textnormal{btw}}(k,l)\right\}
%&=&
%S^2_{\textnormal{btw}}({k,l})
%	 + M^{-1}(-1)^{I_{\{k\neq l\}}} S^2_{\textnormal{in}}({k,l})\,.
%\end{eqnarray*}
\end{lemma}

\noindent
As illustrated by Lemma \ref{lem::Es2}, the sampling expectations of $s^2_{\textnormal{btw}}(k,l)$ contain not only their  `potential outcomes prototypes' $S^2_{\textnormal{btw}}(k,l)$ but also the within-whole-plot covariances $S^2_{\textnormal{in}}(k,l)$.
This renders them `self-sufficient' for estimating the $\var_{\SP}(\widehat{\tau}_F)$ in \eqref{eq::theoreticalVar}, requiring no extra help from the not-yet-defined `$s^2_{\textnormal{in}}(k,l)$.'
%%%%%%%%%%%%%%%
%Let
%$$
%\mathbf{s^\text{$2$}_{\textnormal{btw}}}(-1_A)=
%\begin{pmatrix}
%s^2_{\textnormal{btw}}(1,1)  & s^2_{\textnormal{btw}}(1,2)\\
%s^2_{\textnormal{btw}}(2,1)  & s^2_{\textnormal{btw}}(2,2)
%\end{pmatrix},
%\,\,
%\mathbf{s^\text{$2$}_{\textnormal{btw}}}(+1_A)=
%\begin{pmatrix}
%s^2_{\textnormal{btw}}(3,3)  & s^2_{\textnormal{btw}}(3,4)\\
%s^2_{\textnormal{btw}}(4,3)  & s^2_{\textnormal{btw}}(4,4)
%\end{pmatrix}.
%$$
%%%%%%%%%%%%%%%
% ====================
% thm::estimated-variances
%%%%%%%%%%%%%%%
\begin{theorem}
\label{thm::estimated-variances}
Under the $2^2$ split-plot design qualified by Definition \ref{def::SP},
the sampling variance of $\widehat{\tau}_F$ can be conservatively estimated by
$$
\widehat{V}_F
=4^{-1} \vg_F^{\T}
 \begin{pmatrix} W_{-1}^{-1}
\begin{pmatrix}
s^2_{\textnormal{btw}}(1,1)&s^2_{\textnormal{btw}}(1,2)\\
s^2_{\textnormal{btw}}(2,1)&s^2_{\textnormal{btw}}(2,2)
\end{pmatrix}
 & \mathbf{0} \\  \mathbf{0} & W_{+1}^{-1}
\begin{pmatrix}
s^2_{\textnormal{btw}}(3,3)&s^2_{\textnormal{btw}}(3,4)\\
s^2_{\textnormal{btw}}(4,3)&s^2_{\textnormal{btw}}(4,4)
\end{pmatrix}
\end{pmatrix}
\vg_F
%\\
%%&=\,
%4^{-1} \vg_F^{\T}
% \begin{pmatrix}
%W_{-1}^{-1}\mathbf{s^\text{$2$}_{\textnormal{btw}}}(-1_A)  & \mathbf{0}\\
% \mathbf{0} & W_{+1}^{-1}\mathbf{s^\text{$2$}_{\textnormal{btw}}}(+1_A)
%\end{pmatrix}
%\vg_F
%%\,,
$$
in the sense that
\begin{equation*}
\var_{\SP}(\widehat{\tau}_F) - \E_{\SP}( \widehat{V}_F)
\,=\,
-(4W)^{-1}S^2_{F\text{-btw}}
\,\,\leq\,\,0\,.
\end{equation*}
The last inequality is strict unless the block average factorial effects $\tau_{(w)\text{-}F}$ are constant across all $w = 1, \ldots, W$, i.e., $S^2_{F\text{-btw}} = 0$.
\end{theorem}

%\begin{remark}
%The entrywise product  in the definition of $\widehat{V}_F$ in \eqref{eq::VhatF} evaluates to
%\begin{align*}
%&\begin{pmatrix}
%W_{-1}^{-1}\mJ_2  & \mathbf{0}\\
% \mathbf{0} & W_{+1}^{-1}\mJ_2
%\end{pmatrix}
%\circ
%\mathbf{s^\text{$2$}_{\textnormal{btw}}}\\
%&=\,
% \begin{pmatrix} W_{-1}^{-1}
%%
%\begin{pmatrix}
%s^2_{\textnormal{btw}}(1,1)&s^2_{\textnormal{btw}}(1,2)\\
%s^2_{\textnormal{btw}}(2,1)&s^2_{\textnormal{btw}}(2,2)
%\end{pmatrix}
%%
% & \mathbf{0} \\  \mathbf{0} & W_{+1}^{-1}
%%
%\begin{pmatrix}
%s^2_{\textnormal{btw}}(3,3)&s^2_{\textnormal{btw}}(3,4)\\
%s^2_{\textnormal{btw}}(4,3)&s^2_{\textnormal{btw}}(4,4)
%\end{pmatrix}
%%
%\end{pmatrix},
%\end{align*}
%rendering the arbitrary zeros in $\mathbf{s^\text{$2$}_{\textnormal{btw}}}$ irrelevant whatsoever.
%\end{remark}

\begin{remark}
Whereas we left ${s^\text{$2$}_{\textnormal{btw}}}(k,l)$ undefined for treatment pairs that can never be observed together within the same whole-plot,
the definition of `${s^\text{$2$}_{\textnormal{in}}}(k,l)$,' as some sensible sample version of ${S^\text{$2$}_{\textnormal{in}}}(k,l)$, could be even more `selective.'
In particular, any candidate of the form
\begin{center}
$\sum\{Y_{(wm)}(k) -Y_{(w)}^\obs(k)\}
\{Y_{(wm)}(l) -Y_{(w)}^\obs(l)\}\,,$
\end{center}
would require $Y_{(wm)}(k)$ and $Y_{(wm)}(l)$ to be both observed for at least some sub-plots. This is possible if and only if $k=l \in \{1,2,3,4\}$, leaving the rest $4\times4-4=12$ pairs of $(k,l)$ indefinite.
\end{remark}

Let $q_{1-\alpha/2}$ be the $100(1-\alpha/2)\%$ quantile of standard normal distribution.
It follows from the finite population central limit theorem \citep{Hajek1960} that
interval
\begin{equation}
\label{eq::intEst}
\left[\widehat{\tau}_F - q_{1-\alpha/2}\widehat{V}_F^{1/2},\widehat{\tau}_F + q_{1-\alpha/2} \widehat{V}_F^{1/2}\right]
\end{equation}
will cover $\tau_F$ with at least $100(1-\alpha)\%$ long-run relative frequency as $W$ and $M$ approach infinity.
%over a long series of independent $2^2$ split-plot randomizations as qualified by Definition \ref{def::SP}, intervals of the form
%\begin{equation}
%\label{eq::intEst}
%\left[\widehat{\tau}_F - q_{1-\alpha/2}\widehat{V}_F^{1/2},\widehat{\tau}_F + q_{1-\alpha/2} \widehat{V}_F^{1/2}\right]
%\end{equation}
%cover $\tau_F$ with at least $100(1-\alpha)\%$ long run relative frequency  as $W$ and $M$ approach infinity.
We thus define \eqref{eq::intEst} as the $100(1-\alpha)\%$ \emph{Neymanian split-plot interval estimate} of $\tau_F$, intending for approximate exact-coverage under between-block or stricter additivity and over-coverage if otherwise at reasonably large $W$ and $M$.
This completes the estimation procedure.

%%%%%%%%%%%%%%%%%%%%%
% ========================================
\section{Based on randomization vs. based on model\label{section::model}}
% ======================================
%%%%%%%%%%%%%%%%%%%%%
Before turning to performance evaluation of the proposed procedure,
let us take a brief detour in this section, and discuss some of the key features that set this randomization-based approach apart from existing model-based alternatives.% in the context of $2^2$ split-plot experiments.

Recall $\vg_A=(-1,-1,+1,+1)$, $\vg_B=(-1,+1,-1,+1)$, and $\vg_{AB}=\vg_A\circ\vg_B$,
such that the $k$th entry in $\vg_F$ equals the level of factor $F$ in treatment $k$.
Let
$
\mD
\,=\, 2^{-1}( \vones_4,\vg_A,\vg_B, \vg_{AB})
%\,=\, 2^{-1}
%\begin{pmatrix}
%+1 & -1 & -1 & +1\\
%+1 & -1 & +1 & -1\\
%+1 & +1 & -1 & -1\\
%+1 & +1 & +1 & +1
%\end{pmatrix}
$
be the design matrix,
and let $g_F(k)$ be the $k$th entry in $\vg_F$.
It follows from the orthogonality of $\mD$ that
%We write the vector of potential outcomes of sub-plot $(w,m)$ can be written in terms of $\tau_{(wm)\text{-}F}$ and $\mu_{(wm)} = 4^{-1}\vones_4^\T \vY_{(wm)}$ as
\begin{align}
\vY_{(wm)}
&=\,  \mD\mD^\T\vY_{(wm)} \,=\, \mD \left\{2^{-1}\!\left( \vones_4,\vg_A,\vg_B, \vg_{AB}\right)^\T\vY_{(wm)}\right\} \label{eq::YasTau}\\
%\,=\,  \mD\left( \mD^\T\vY_{(wm)}\right)
&=\,
\mD \left( 2^{-1}\vones_4^\T\vY_{(wm)}, 2^{-1}\vg_A^\T\vY_{(wm)}, 2^{-1}\vg_B^\T\vY_{(wm)}, 2^{-1}\vg_{AB}^\T\vY_{(wm)}\right)^\T\nonumber\\
&=\,
\mD  \left(2\mu_{(wm)},  \tau_{(wm)\text{-}A},\tau_{{(wm)\text{-}B}}, \tau_{(wm)\text{-}AB}\right)^\T, \nonumber
\end{align}
with
\begin{align}
Y_{(wm)}(k)
&=
2^{-1}\left(1, g_A(k), g_B(k), g_{AB}(k)\right)\label{eq::decomp_PO_k}
\left(2\mu_{(wm)},  \tau_{(wm)\text{-}A},\tau_{{(wm)\text{-}B}}, \tau_{(wm)\text{-}AB}\right)^\T
\\
&=\mu_{(wm)} +  \sum_{F\in \mathcal{F}} 2^{-1}g_F(k) \tau_{(wm)\text{-}F} \nonumber
\end{align}
in the $k$th row.
%%%%%%%%%%%%%%%%%%%
Averaging \eqref{eq::decomp_PO_k} over all $(wm)$ yields
\begin{align}
%\bar{\vY}&=\left(\bar Y(k), \bar Y(k) , \bar Y(k) , \bar Y(k) \right)^\T \,=\, 2^{-1}\mD (2\mu,\tau_A,\tau_B,\tau_{AB})^\T \,\\
\bar Y(k)
&=\,  \mu + \sum_{F\in \mathcal{F}}  2^{-1}g_F(k) \tau_{F}
%2^{-1}\left(1, g_A(k), g_B(k), g_{AB}(k) \right)
%(2\mu,\tau_A,\tau_B,\tau_{AB})^\T
\label{eq::decomp_PO_popk}.
%\\
%\,=\,&
%\mu_{(wm)} + 2^{-1}\tau_{{(wm)\text{-}A}}A_w + 2^{-1}\tau_{{(wm)\text{-}B}}B_{(wm)} + 2^{-1}\tau_{(wm)\text{-}AB}A_wB_{(wm)}\,, \nonumber
\end{align}

Recall that $T_{(wm)} = k$ if sub-plot $(wm)$ is assigned to treatment $k$.
We have $Y_{(wm)}^\obs = Y_{(wm)}(T_{(wm)})$.
The \emph{derived linear model} \citep{Hinkelmann2008} treats the population average $\bar Y(T_{(wm)})$  as the part in $Y_{(wm)}(T_{(wm)})$ explainable by the treatment,
and
decomposes the observed outcomes as
%, and treat whatever remains as unit-level variation.
\begin{align}
\label{eq::decomp_PO_resid}
Y_{(wm)}^\obs  \,=\, Y_{(wm)}(T_{(wm)})
&=\,\bar Y(T_{(wm)})+\epsilon_{(wm)}\\
&=\,\mu +  \sum_{F\in \mathcal{F}} 2^{-1}g_F(T_{(wm)}) \tau_{F}+\epsilon_{(wm)}\nonumber\,,
\end{align}
where
$\epsilon_{(wm)} = Y_{(wm)}^\obs - \bar Y(T_{(wm)})$ are the unit-level random errors,
 and the last equality follows from letting $k = T_{(wm)}$ in \eqref{eq::decomp_PO_popk}.
Let
\begin{center}
$\delta_{(wm)\text{-}\mu}\,=\, \mu_{(wm)} - \mu\,,\quad \delta_{(wm)\text{-}F}\,=\, \tau_{(wm)\text{-}F} - \tau_F\quad(F \in \mathcal{F})$
\end{center}
be the deviations of unit-level parameters from the \emph{finite-population} averages.
Plug \eqref{eq::decomp_PO_k}, with $k$ set at $T_{(wm)}$,  into \eqref{eq::decomp_PO_resid} to see
\begin{equation}
\label{eq::PO_resid}
\epsilon_{(wm)} =
\delta_{(wm)\text{-}\mu} +  \sum_{F\in \mathcal{F}} 2^{-1}g_F(T_{(wm)}) \delta_{(wm)\text{-}F}\,.
\end{equation}

The $g_F(T_{(wm)})$ in \eqref{eq::decomp_PO_resid}, despite the compound definition as `the level of factor $F$ in the treatment $T_{(wm)}$ received by sub-plot $(wm)$,' has the straightforward interpretation as the level of factor $F$ received by sub-plot $(wm)$.
%
%Letting $k=T_{(wm)}$ in \eqref{eq::decomp_PO_k} and \eqref{eq::decomp_PO_popk} yields
%\begin{align}
%Y_{(wm)}(T_{(wm)})
%%Y_{(wm)}^\obs &=
%\label{eq::decomp_PO}\\
%&=
%\mu_{(wm)} + 2^{-1} \sum_{F\in \mathcal{F}} g_F(\mu_{(wm)} + 2^{-1} \sum_{F\in \mathcal{F}} g_F(k) \tau_{(wm)\text{-}F}) \tau_{(wm)\text{-}F}
%%2^{-1}(1, A_w, B_{(wm)}, A_wB_{(wm)})
%%\left(2\mu_{(wm)},\tau_{(wm)\text{-}A}, \tau_{(wm)\text{-}B},\tau_{(wm)\text{-}AB}\right)^\T,
%\nonumber\\
%\bar Y(T_{(wm)})
%%&=
%%2^{-1}(1, A_w, B_{(wm)}, A_wB_{(wm)})
%%\left(2\mu,\tau_A,\tau_B,\tau_{AB}\right)^\T\nonumber\\
%&= \mu + 2^{-1} \sum_{F\in \mathcal{F}} g_F(T_{(wm)}) \tau_{F}.\nonumber
%\end{align}
%
%
%On the other hand, let $T_{(wm)} \in \{1,2,3,4\}$ the treatment assignment of sub-plot $(wm)$,
%and let $A_w$ and $B_{(wm)}$ be the corresponding levels of factors $A$ and $B$.
%With $A_w=g_A(T_{(wm)})$, $B_{(wm)}=g_B(T_{(wm)})$, and $g_{AB}(T_{(wm)}) = g_A(T_{(wm)}) g_B(T_{(wm)}) =A_wB_{(wm)}$,
%\begin{align}
%Y_{(wm)}(T_{(wm)})
%%Y_{(wm)}^\obs &=
%\label{eq::decomp_PO}\\
%&=
%\mu_{(wm)} + 2^{-1}\left( A_w\tau_{(wm)\text{-}A} +B_{(wm)}\tau_{(wm)\text{-}B} + A_wB_{(wm)}\tau_{(wm)\text{-}AB}\right)
%%2^{-1}(1, A_w, B_{(wm)}, A_wB_{(wm)})
%%\left(2\mu_{(wm)},\tau_{(wm)\text{-}A}, \tau_{(wm)\text{-}B},\tau_{(wm)\text{-}AB}\right)^\T,
%\nonumber\\
%\bar Y(T_{(wm)})
%%&=
%%2^{-1}(1, A_w, B_{(wm)}, A_wB_{(wm)})
%%\left(2\mu,\tau_A,\tau_B,\tau_{AB}\right)^\T\nonumber\\
%&=\mu + 2^{-1}\left( \tau_{A}A_w +\tau_{B}B_{(wm)} + \tau_{AB}A_wB_{(wm)}\right).\nonumber
%\end{align}
%
%
%
%
%\medskip
This, together with the functional form of \eqref{eq::decomp_PO_resid}, unsurprisingly reminds us of the family of additive regression models:
\begin{equation}
\label{eq::linearModel}
Y^\obs_{(wm)}
\,=\, {\beta}_0 +  \sum_{F\in \mathcal{F}} g_F(T_{(wm)})\beta_{F}+\epsilon_{(wm)}^{\text{model}}\,.
\end{equation}
%%%%%%%%%%%%%%%%%%%%%%%%%%%%%
Despite the apparent resemblance between \eqref{eq::decomp_PO_resid} and \eqref{eq::linearModel}, however, their difference is fundamental, with the source of randomness being the  first and foremost.

The family of additive regression models \eqref{eq::linearModel}, on the one hand, conditions on the treatment assignments $T_{(wm)}$ for all its inference, and attributes the randomness in $Y^\obs_{(wm)}$ to the study population being a \emph{random sample} of some \emph{hypothetical} super-population,
reflected via $\epsilon_{(wm)}^{\text{model}}$ as the individual sampling errors.
The regression coefficients $\beta_F$ are treated as super-population causal parameters, and the linear combinations ${\beta}_0 +  \sum_{F\in \mathcal{F}} g_F(T_{(wm)})\beta_{F}$ as deterministic super-population means.
% and the residuals $\epsilon^{\text{model}}_{(wm)}$ as the sampling errors while selecting the study population.

%%%%%%%%%%%%%%%
% Randomization
%%%%%%%%%%%%%%%
The derived linear model  \eqref{eq::decomp_PO_resid}, on the other hand, conditions on the composition of study population for all its inference, and attributes the randomness in $Y_{(wm)}^\obs$ solely to the \emph{random assignment of treatments},
reflected via the joint distribution of treatment assignment variables $T_{(wm)}$.
As a result, not only the residuals $\epsilon_{(wm)}$,
but the linear combinations $\mu +  \sum_{F\in \mathcal{F}} 2^{-1}g_F(T_{(wm)}) \tau_{F}$ too,
are now stochastic via their dependence on $T_{(wm)}$ \citep{Freedman2008a,Freedman2008b,Freedman2008c},
with coefficients $\tau_F$, by definition \eqref{eq::estimands}, describing the finite study population.
See formula \eqref{eq::PO_resid} for a full specification of $\epsilon_{(wm)}$ in terms of $g_F(T_{(wm)})$.

%%%%%%%%%%%%%%%%%%%%%%%%
%Thirdly,
%while most linear-based models assume errors to be identically distributed with zero mean,
%\begin{eqnarray*}
%E_\SP(\epsilon_{(wm)})
%\,=\,
%2^{-1}\Delta2\mu_{(wm)} + 2^{-1}\delta_{{(wm)\text{-}A}}{e_A} + 2^{-1}\delta_{(wm)\text{-}B}{e_B}+   2^{-1}\delta_{(wm)\text{-}AB}{e_Ae_B}
%\end{eqnarray*}
%are in general not the same across different units, nor equal to zero.

%%%%%%%%%%%%%%%%%
%%%%%%%%%%%%%%%%%%
More quantitative comparison follows from the difference in residual covariance structure.
%, and, surprisingly, also establishes a common ground between the two methods at  exactly where they seem to diverge.
%% the search for a common ground may not actually be completely fruitless.
%%%%%%%%%%%%%%%%%%%
%
%In particular,
%%%%%%%%%%%%%
Whereas the covariances of the $\epsilon^{\text{model}}_{(wm)}$ in \eqref{eq::linearModel} are in general specified as model assumptions,
those of the $\epsilon_{(wm)}$ in \eqref{eq::decomp_PO_resid}  follow naturally from identity  \eqref{eq::PO_resid} and the joint distribution of $T_{(wm)}$ as determined by the treatment assignment mechanism.

To start with, %identity \eqref{eq::PO_resid}, when viewed in conjunction with Lemma \ref{def3},
{viewing \eqref{eq::PO_resid} in conjunction with Lemma \ref{def3}} renders the computation of $\cov_{\SP}(\epsilon_{(wm)},\epsilon_{(w^\prime m^\prime)})$ almost trivial under strict additivity:
With
$\delta_{(wm)\text{-}F} = 0$
for all $(wm)$ and $F \in \{A, B, AB\}$,
the residuals in \eqref{eq::PO_resid} reduce to constants  $\epsilon_{(wm)}=\delta_{(wm)\text{-}\mu}$, and the covariance of constants is always zero, i.e.,
$\cov_{\SP}(\epsilon_{(wm)},\epsilon_{(w^\prime m^\prime)})\,=\, 0$
for all $(wm)$ and $(w^\prime m^\prime)$  under strict additivity.
%%%%%%%%%%%

Without strict additivity, the algebra becomes tedious.
To avoid unnecessary complexity,
we defer the exact formulas for $\cov_{\SP}(\epsilon_{(wm)},\epsilon_{(w^\prime m^\prime)})$ at each finite $(W, M)$ to the online supplementary material,
and save Theorem \ref{thm::cov} for but the `punch line' in terms of \emph{finite-population asymptotics} \citep{Hajek1960}
\begin{center} $\lim_{W, M \to \infty} \cov_{\SP_{(W,M, r_A, r
_B)}}(\epsilon_{(wm)},\epsilon_{(w^\prime m^\prime)})$.
\end{center}
The asymptotic condition `$W, M \to \infty$' can be visualized as keeping adding till infinity new whole-plots to the current study population, and new sub-plots to the current whole-plots.
The covariance at each finite $(W, M)$ is computed under split-plot design `$\SP_{(W,M, r_A, r
_B)}$' as qualified by Definition \ref{def::SP} with $W_{+1}=r_A(r_A+1)^{-1}W$ and $M_{+1}= r_B(r_B+1)^{-1}M$.

%%%%%%%%%%%%%%%%%%%%%%%%%
%%%%%%%%%%%%%%%%%%%%%%%%%%

\begin{theorem}
\label{thm::cov}
Fix $r_A$ and $r_B$.
As $W$\! and $M$\! approach infinity,
the residual covariance $\cov_{\SP_{(W,M,r_A, r_B)}}(\epsilon_{(wm)},\epsilon_{(w^\prime m^\prime)})$ for sub-plots $(wm)$ and $(w^\prime m^\prime)$ in the current study population will converge to
\begin{equation*}
\frac{r_A}{(r_A+1)^2}\left\{\delta_{(wm)\text{-}A}+ \left(\frac{r_B-1}{r_B+1}\right)\delta_{(wm)\text{-}AB}\right\}
\left\{\delta_{(w^\prime m^\prime)\text{-}A}+\left(\frac{r_B-1}{r_B+1}\right)\delta_{(w^\prime m^\prime)\text{-}AB} \right\}
\end{equation*}
if the two are in the same whole-plot, and to zero if they are not.
\end{theorem}

\begin{corollary}
\label{corollary::7}
%For balanced $2^2$ split-plot designs, i.e., $r_A = r_B=1$,
% i.e., $r_A = r_B=1$,
%the asymptotic covariance $\lim_{W, M \to \infty} \cov_{\SP_{(W,M, r_A, r_B)}}(\epsilon_{(wm)},\epsilon_{(w^\prime m^\prime)})$ reduces to
%$4^{-1}\delta_{(wm)\text{-}A}\delta_{(w^\prime m^\prime)\text{-}A}$ for sub-plots $(wm)$ and $(w^\prime m^\prime)$ in the same whole-plot.
When the design series `$\SP_{(W,M, r_A, r
_B)}$' is balanced, i.e., $r_A = r_B=1$,
the asymptotic residual covariance  in Theorem \ref{thm::cov} reduces to
$4^{-1}\delta_{(wm)\text{-}A}\delta_{(w^\prime m^\prime)\text{-}A}$ for sub-plots $(wm)$ and $(w^\prime m^\prime)$ in the same whole-plot.
\end{corollary}

%%%%%%%%%%%%%%%%%%%%%%%%%%%%%%%%
Theorem \ref{thm::cov} and Corollary \ref{corollary::7}
provide an explicit account of the non-vanishing within-whole-plot correlation of $\epsilon_{(wm)}$ under $2^2$ split-plot designs \citep{Freedman2008a,Lin2013},
and thereby justify heuristically the  block-diagonal covariance structure that a linear mixed effects (\textsc{lme}) model assumes for its sampling errors.
With
\begin{center}
$
\epsilon^\LME_{(wm)} \,=\, \eta_w + \xi_{(wm)}
$
\end{center}
where %iid $\eta_w^{\LME}$ and iid $\xi^{\LME}_{(wm)}$
$\eta_w\overset{\text{iid}}{\sim} \mathcal{N}(0, \sigma_\eta^2)$ and $\xi_{(wm)}\overset{\text{iid}}{\sim} \mathcal{N}(0, \sigma_\xi^2)$
are jointly independent,
the covariance of
$\epsilon^\LME_{(wm)}$ and $\epsilon^\LME_{(w^\prime m^\prime)}$ equals $\sigma_\eta^2$ %$\var^*(\eta_w^{\LME}) $
if $w=w^\prime$, and 0 if otherwise.
%%%%%%%%%%%%%%%%%%%%%%%%%%%%%%%%%%
%%%%%%%%%%%%%%%%%%%%%%%%%%%%%%%%%%
Despite the `qualitative' similarity in structure,
two salient quantitative differences remain:
\begin{itemize}
\item[\bell] First, whereas the linear mixed effects model assumes equal covariances for all pairs of residuals from the same whole-plot,
those under the derived linear model, as is clear from Theorem \ref{thm::cov} and Corollary \ref{corollary::7}, vary from pair to pair even in the asymptotics.

\vspace{1mm}
\item[\bell] Second, whereas the linear mixed effects model assumes independence between whole-plots  at any finite  $(W,M)$,
formula \eqref{eq::PO_resid} suggests otherwise for the derived model.
Intuitively, $\epsilon_{(wm)}$ and $\epsilon_{(w^\prime m^\prime)}$ from two different whole-plots  are correlated  at any finite $W$ via their respective dependence on $T_{(wm)}$ and $T_{(w^\prime m^\prime)}$ and the mutual dependence between $T_{(wm)}$ and $T_{(w^\prime m^\prime)}$
--- given that knowing whole-plot $w$ receives one level of factor $A$ lowers the probability of whole-plot $w^\prime$ to receive the same level, the two assignment variables $T_{(wm)}$ and $T_{(w^\prime m^\prime)}$ are mutually correlated even if $w \neq w^\prime$.
See the online supplementary material for exact formulas for $\cov_{\SP}(\epsilon_{(wm)},\epsilon_{(w^\prime m^\prime)})$ at each finite $(W, M)$.
\end{itemize}

%%%%%%%%%%%%%%%%%%%%%%%%%%%%%%%%%%%%%%%%%%%%%%%%%%%%%%%%%%%%%%%%%%%
\section{Simulations\label{section::simulation}}
%%%%%%%%%%%%%%%%%%%%%%%%%%%%%%%%%%%%%%%%%%%%%%%%%%%%%%%%%%%%%%%%%%%%%%%%%%%%%%%%%%%%%%%%%%%%%%%%%%%%%%%%%%%%%%%%%%%%%%%%%%
We evaluate in this section the frequency coverage property of the proposed Neymanian split-plot interval estimates via simulation.

%%%%%%%%%%%%%%%%%%%%%%%%%%%%%%%%%%%%%%%%%%%%%%%%%%%%%%
% Normal
%%%%%%%%%%%%%%%%%%%%%%%%%%%%%%%%%%%%%%%%%%%%%%%%%%%%%%
\subsection{Generative models for the \textsc{pom}s}
%%%%%%%%%%%%%%%%%%%%%%%%%%%%%
%\textbf{\color{red}COVERAGE RATES --- results are invariant to the magnitudes of the treatment effects.}
Refer to the condition of all four $\bm Y(k) = \left( Y_{(11)}(k),\ldots,Y_{(WM)}(k) \right)$ being blockwise constant --- i.e.,  $Y_{(wm)}(k) = Y_{(w)}(k)$ for all $w$, $m$, and $k$ --- as `ultimate block effect.'
We consider here five types of potential outcomes: % (\textsc{po}):
\begin{itemize}
\item[\textsc{(i)}] binary potential outcomes without block effect,
\item[\textsc{(ii)}] binary potential outcomes with ultimate block effect,
\item[\textsc{(iii)}] continuous potential outcomes without block effect,
\item[\textsc{(iv)}] continuous potential outcomes with block effect,
\item[\textsc{(v)}] continuous potential outcomes with ultimate block effect
\end{itemize}
in combination with three types of additivity assumption: % (\textsc{adt}):
\begin{itemize}
\item[(i)] strict additivity,
\item[(ii)] between-block additivity,
\item[(iii)] no assumption about additivity.
\end{itemize}
This gives a total of $5 \times 3 = 15$ types of \textsc{pom}, from which specific \textsc{pom}s are generated in two steps:
\begin{itemize}
\item[1.] Generate $\vY(1)$ according to the designated potential outcomes type. See Table \ref{tb::POM} for details about the generative models.
\item[2.] Conditional on $\vY(1)$, generate $\vY(k)$ $(k=2,3,4)$ according to the designated additivity type. See Table \ref{tb::bb} for details about the generative models.
\end{itemize}
%%%%%%%%%%%%%%%%%%
Strict additivity for all five potential outcomes types is imposed by letting $\vY(k) = \vY(1)$ $(k=2,3,4)$, and between-block additivity by letting
\begin{equation}
\label{eq::bb}
Y_{(w)}(k) \,=\, Y_{(w)}(1)\quad (k = 2,3,4; w = 1, \ldots, W)\,,
\end{equation}
such that the resulting \textsc{pom}s satisfy Definitions \ref{def::strictAdd} and \ref{def::2} respectively with all differential constants being zero.
No generality is lost so far as the coverage rate is concerned.

%%%%%%%%%%%%%%%%%%%%%%%%%%%%
\subsection{Interval estimates and their coverage rates}
%%%%%%%%%%%%%%%%%%%%%%%%%%%%
For each realized \textsc{pom},
coverage rates of  the {proposed Neymanian split-plot interval estimates} are summarized over
1,000 independent split-plot randomizations and compared to those of the following three alternatives:
%,
%and the following three best known procedures for split-plot design analysis as benchmarks:
%we consider here three alternative procedures as benchmarks:  ...weconsiderhere...
%{\color{red}To evaluate how the proposed interval \eqref{eq::intEst} performs relative to other existing alternatives},
%the frequency coverage property of the proposed  by not only absolute but also relative standards
%To evaluate the frequency coverage property of the proposed interval  by not only absolute but also relative standards,
\begin{itemize}
\item {\textsc{glm} interval estimates}.
\begin{itemize}
\item[]The $100(1-\alpha)\%$ confidence intervals under the standard generalized linear model (\textsc{glm}) with the levels of factors $A$ and $B$ and their interaction as explanatory variables.
\end{itemize}
\item  {\textsc{glme} interval estimates}.
\begin{itemize}
\item[]The $100(1-\alpha)\%$ confidence intervals under the standard generalized linear mixed effects model (\textsc{glme}) that includes also whole-plot dummy, in addition to the levels of factors $A$ and $B$ and their interaction, as explanatory variable.
\end{itemize}
\item {{\CR} interval estimates}.
\begin{itemize}
\item[]The $100(1-\alpha)\%$ {Neymanian interval estimates} for $2^2$ completely randomized ({\CR}) design proposed by \cite{Tir2015}.
\end{itemize}
\end{itemize}
%%%%%%%%%%%%%%
All \textsc{glm}s are fitted by the standard R function `{glm},' and all \textsc{glme}s by `{glmer},' both with `{binomial}'  link for binary potential outcomes types \textsc{(i)--(ii)} and  `{identity}' link for continuous potential outcomes types \textsc{(iii)--(v)}.
We abbreviate `\textsc{glm}' to `\textsc{lm},' and `\textsc{glme}'  to `\textsc{lme}' in the latter case,
inasmuch as the identity link reduces the two generalized models to linear and linear mixed effects models, respectively.

%%%%%%%%%%%%%%%%%%%%%%%%%%%%%%%%%%%
%
\subsection{Results}
%
%%%%%%%%%%%%%%%%%%%%%%%%%%%%%%%%%%%%
We realize each of the 15 \textsc{pom} types at two sizes: $(W,M) = (40,40)$ and $(80,80)$, and construct the intervals at confidence level $\alpha=0.05$.
Results for the 15 \textsc{pom}s at $(W,M) = (40,40)$ are shown in Figure \ref{fig::res}; the overall superiority of {\SP} interval is evident.
Results at $(W,M) = (80,80)$ exhibit quite similar patterns, and are thus not included here to avoid redundancy.

The intended `approximate exact-coverage under between-block or stricter additivity and over-coverage if otherwise' is fulfilled by the proposed {\SP} interval for all but potential outcomes types \textsc{(ii)} and \textsc{(v)} under strict additivity.
Despite its undue conservativeness towards $\tau_B$ and $\tau_{AB}$ in these two cases,
the proposed {\SP} interval remains to be the only interval that `does not under-cover' --- see Table \ref{tb::blockwiseConst} for the untruncated statistics regarding the severe under-coverage of $\tau_A$ by {\LM} and {\LME} intervals.
%The fact of the $\widehat{\tau}_B$ and $\widehat{\tau}_{AB}$ in these two cases being virtually constant  at the true values ---i.e., $\widehat{\tau}_B \equiv {\tau}_B$ and $\widehat{\tau}_{AB}\equiv {\tau}_{AB}$ --- over all possible assignments, as a result of the ultimate block effect} --- may render even {\SP}'s undue conservativeness excusable.
The fact that $\widehat{\tau}_B$ and $\widehat{\tau}_{AB}$ in these two cases are virtually constant  at their respective true values ${\tau}_B$ and ${\tau}_{AB}$ over all possible assignments, as a result of the ultimate block effect, may render even {\SP}'s undue conservativeness excusable.

For potential outcomes type \textsc{(iv)} in particular, {\SP} markedly outperforms {\LM} (\CR) in all three factorial effects,
matches {\LME} in the main effect of whole-plot factor $A$,
and beats the latter in all other cases.
The fact of potential outcomes type \textsc{(iv)} being actually generated from {\LME} model accentuates {\SP}'s victory even further.

The general inadequacy of {\CR}, {\LM}, and {\textsc{glm}} intervals  for potential outcomes types \textsc{(ii),\! (iv)},\! and \textsc{(v)}, on the other hand, exemplifies the possible severe under-coverage when split-plot experiments are wrongfully analyzed as completely randomized ones,  even when the preferred randomization-based perspective is adopted.

%%%%%%%%%%%%%%%%%%%%%%%%%%%%%%%%%%%%%%
\begin{table}[ht]
\caption{\label{tb::POM}Generative models for $\vY(1)$ under potential outcomes (PO) types (I)--(V).}
\renewcommand*{\arraystretch}{0.8}
%%%%%%%%%%%%%%%%%%%%%%%%%%%%
%%%%%%%%%%%%%%%%%%%%%%%%%%%%
\begin{tabular}{ll}\hline
\multirow{2}{*}{PO Type}&\multirow{2}{*}{Generative Model for  $\vY(1) = \left( Y_{(11)}(1),\ldots,Y_{(WM)}(1) \right)$}\\\\\hline
%%%%%%%%%%%%%%%%%%%%%%%%%%%%%%%%%%%%%
%%%%%%%%%%%%%%%%%%%%%%%%%%%%%%%%%%%%%
\multirow{2}{*}{\phantom{aaa}(I)}&
\multirow{2}{*}{$Y_{(wm)}(1) \overset{\textnormal{iid}}{\sim}  \text{Bern}(0.5)$.}\\
&\\\hdashline
%%%%%%%%%%%%%%%%%%
%%%%%%%%%%%%%%%%%%%%%%%%%%%%%%%%
\multirow{2}{*}{\phantom{aaa}(II)}&
\multirow{2}{*}{$Y_{(w)}(1) \overset{\textnormal{iid}}{\sim} \text{Bern}(0.5)$, and $Y_{(wm)}(1) = Y_{(w)}(1)$.}\\
&\\\hdashline\\
%%%%%%%%%%%%%%%%%
\multirow{2}{*}{\phantom{aaa}(III)}&
$Y_{(wm)}(1)$ are independent normals with means $\mu_{(wm)} = 2(-1)^{I\{m \leq M/2\}}$ and\\

\vspace{1mm}
&variances $\left(\sigma_{(11)}^2, \ldots, \sigma_{(WM)}^2\right)$ being a
random permutation of $2(\vones^\T_{N/2}, \vzeros^\T_{N/2})$.
\\\hdashline
%%%%%%%%%%%%%%%%%%%%%%%%%
\multirow{2}{*}{\phantom{aaa}(IV)}
&
\multirow{2}{*}{$Y_{(wm)}(1) = \eta_w + \epsilon_{(wm)}$, where $\eta_w$ and $\epsilon_{(wm)}$ are iid standard normals.} \\&\\\hdashline
%%%%%%%%%%%%%%%%%%%%%%%%%%%%%%%%%%%%%
\multirow{2}{*}{\phantom{aaa}(V)} &
\multirow{2}{*}{$Y_{(w)}(1) \overset{\textnormal{iid}}{\sim} \mathcal{N}(0,1)$, and $Y_{(wm)}(1) = Y_{(w)}(1)$.}\\
&\\\hline
\end{tabular}
\end{table}

%%%%%%%%%%%%%%%%%%%%%%%%%%%%
%%%%%%%%%%%%%%%%%%%%%%%%%%%%
\begin{table}[htbp]
\caption{Generative models for $\vY(k)$ $(k = 2,3,4)$ under the 15 \textsc{pom} types as combinations of the five potential outcomes (PO) types (I)--(V) in Table \ref{tb::POM}, and the three additivity (ADT) types: (i) strict, (ii) between-block, and (iii) no assumption about additivity. \label{tb::bb}}
\renewcommand*{\arraystretch}{1.7}
\centering
\begin{tabular}{c:c:l}\hline
ADT Type&{PO Type} & \multicolumn{1}{c}{Generative Model for $\vY(k)$ $(k = 2,3,4)$} \\\hline
(i)
	&(I)\,--\,(V)
	&\begin{tabular}{l} $\vY(k) = \vY(1)$.\end{tabular} \\\hdashline
%%%%%%%%%%%%%%%%%%%%%%%%%%%%%%%%%%%%%%
\multirow{9}{*}{(ii)}
%%%%%%%%%%%%%%%%%%%
	&\multirow{3}{*}{(I)}
	&\multirow{3}{*}{\begin{tabular}{l}
	$\vY(k)$ are independent blockwise permutations of $\bm Y(1)$, \\
	such that the numbers of 1's within each block are the\\ same for $\bm Y(k)$ and $\bm Y(1)$. This ensures \eqref{eq::bb}.\end{tabular}}\\
     &&\\
     &&\\\cdashline{2-3}
%%%%%%%%%%%%%%%%%%%
	&\multirow{3}{*}{(II),\,(V)}
	&\multirow{3}{*}{\begin{tabular}{l}
	$\vY(k) = \vY(1)$. Under ultimate block effect, we have\\
	$Y_{(wm)}(1) = Y_{(w)}(1)$ and $Y_{(wm)}(k) =Y_{(w)}(k)$;\\
     \eqref{eq::bb} holds if and only if $Y_{(wm)}(k) = Y_{(wm)}(1)$.
	\end{tabular}}\\
     &&\\
     &&\\\cdashline{2-3}
%%%%%%%%%%%%%%%%%%%%%%%%%%%%
	&\multirow{3}{*}{(III),\,(IV)}
	& \multirow{3}{*}{\begin{tabular}{l}
	$Y_{(wm)}(k)=Y^{\prime}_{(wm)}(k)-\{{Y}^{\prime}_{(w)}(k) -  {Y}_{(w)}(1) \}$, \\where $\bm Y^{\prime}(k)$ are iid as $\vY(1)$.\\
	Subtracting ${Y}^{\prime}_{(w)}(k) -  {Y}_{(w)}(1)$ ensures \eqref{eq::bb}.
	\end{tabular}}\\
&&\\
&&\\\hdashline
%%%%%%%%%%%%%%%%%%%%%%%%%%
(iii)&(I)\,--\,(V)&\begin{tabular}{l}$\vY(k)$ are iid as $\vY(1)$.\end{tabular}\\\hline
\end{tabular}
\end{table}

\begin{table}[htbp]
\caption{Coverage rates (\%) averaged over 1,000 independent split-plot randomizations with $r_A=r_B=1$ at $(W, M )= (40,40)$ for potential outcomes (PO) types (II) and (V).\label{tb::blockwiseConst}}
\begin{center}
\renewcommand{\arraystretch}{1.2}
\begin{tabular}{cl:rrrr:rrr}\hline
 &&\multicolumn{4}{c:}{PO Type (II)} & \multicolumn{3}{c}{PO Type (V)}\\
&&S-P &C-R &GLM &{GLME}& S-P &	LM (C-R)	& LME\\\hline
%%%%%%%%%%%%%
\multirow{3}{*}{Strict Additivity}
&$\tau_A$&95.0 &0.0 &0.0 &32.5 &	95.0 &	22.9&	25.9\\
&$\tau_B$ & 100.0&100.0&100.0&	100.0&100.0&100.0 &100.0\\
&$\tau_{AB}$&	100.0&100.0&100.0&	100.0&100.0&100.0 &100.0 \\\hdashline
%%%%%%%%%%%%%%%%%%%
\multirow{3}{*}{\begin{tabular}{c}No Assumption\\ about Additivity\end{tabular}}
& $\tau_A$  &99.3&33.9	&34.4	&84.3&99.6 &45.4	&99.6\\
 & $\tau_B$ &99.7&42.3&44.0&	33.2&97.2&39.6&27.1\\
 & $\tau_{AB}$ & 98.7&36.8	&47.3&26.2&100.0	&65.0	&	43.8\\\hline
\end{tabular}
\end{center}
\end{table}

\begin{figure}[htbp]
{\caption{Coverage rates summarized over 1,000 independent split-plot randomizations with $r_A=r_B=1$ at $(W,M) = (40,40)$ ($\alpha = 0.05$).
All bars start from the nominal coverage rate 0.95 and grow upwards/downwards to the actual values, truncated at 0.85. Results of {\CR} and {\LM} are combined for potential outcomes (PO) types (III)--(V), since the procedure by which \cite{Tir2015} constructed the {\CR} renders it numerically identical to the {\LM}. \label{fig::res}}}
\begin{center}
\begin{tabular}{c:cccc}\cline{1-4}
{PO}&Strict & Between-Block & No Assumption \\
Type &Additivity & Additivity & about Additivity\\\cline{1-4}\\
(I) & \includegraphics[height =.15\textheight]{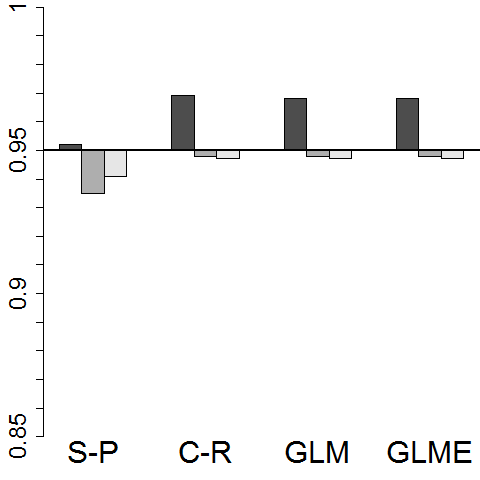}&
\includegraphics[height =.15\textheight]{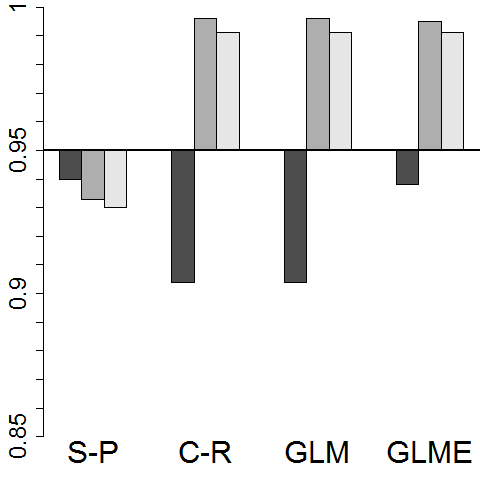}&
\includegraphics[height =.15\textheight]{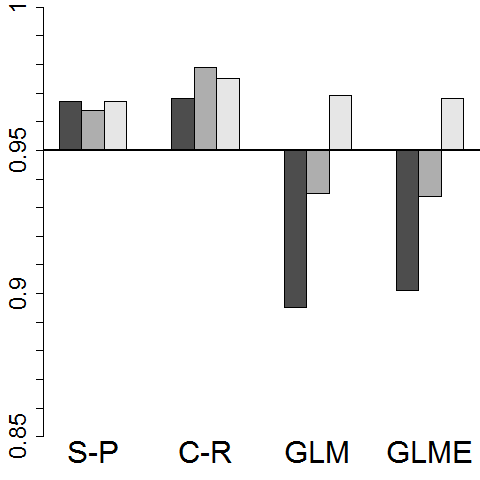}
&\includegraphics[scale=.9]{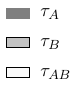}\\\cdashline{1-4}\\
%%%%%%%%%%%%%%%%%%%%%%%%
(II)& \includegraphics[height =.15\textheight]{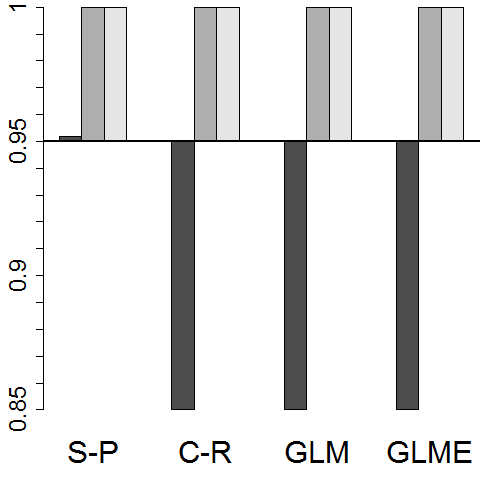}
&
&
\includegraphics[height =.15\textheight]{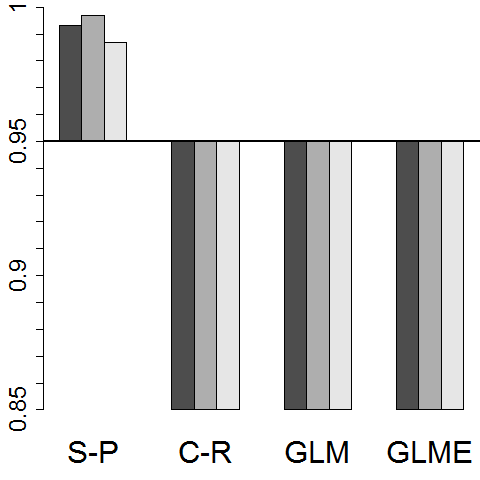}
&\\\cdashline{1-4}\\
%%%%%%%%%%%%%%%%%%%%%%%%
(III) & \includegraphics[height =.15\textheight]{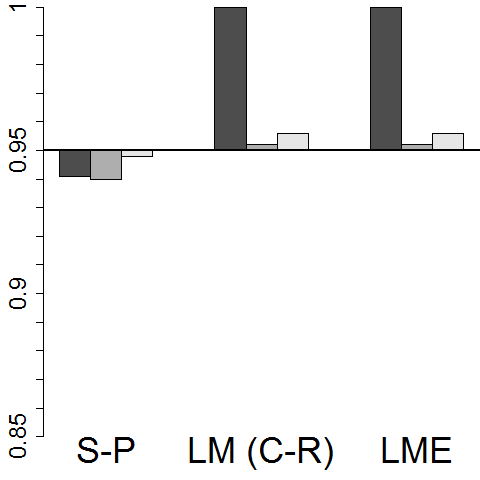}&
\includegraphics[height =.15\textheight]{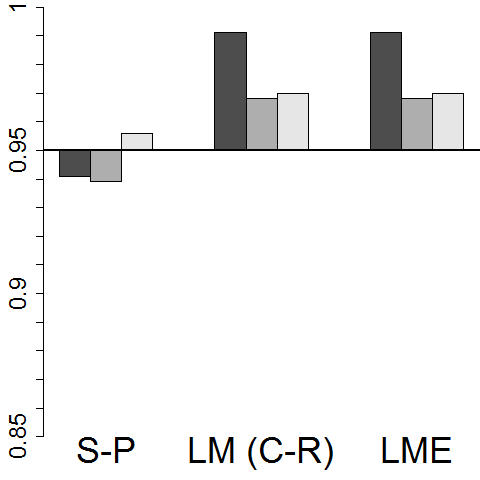}&
\includegraphics[height =.15\textheight]{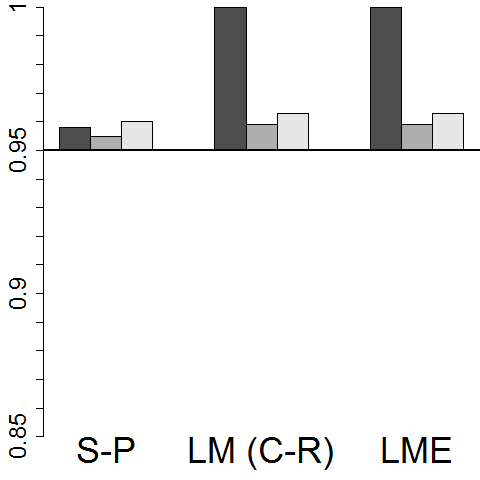}&\\\cdashline{1-4}\\
%%%%%%%%%%%%%%%%%%%%%%%%%%%
(IV) &\includegraphics[height=.15\textheight]{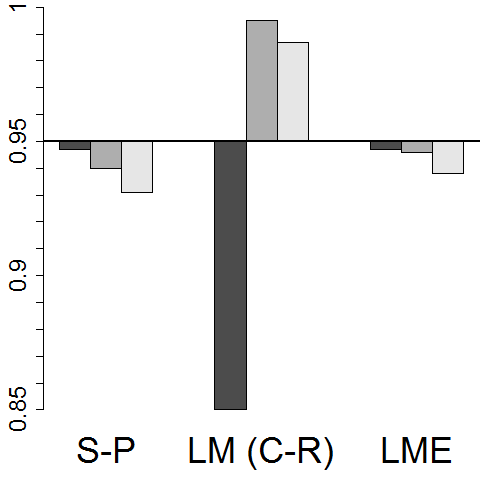}&
\includegraphics[height=.15\textheight]{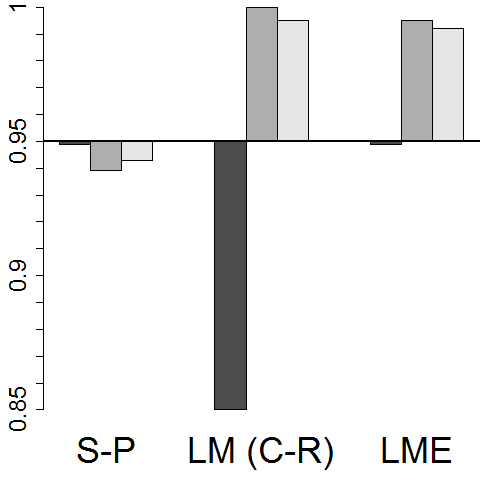}&
\includegraphics[height=.15\textheight]{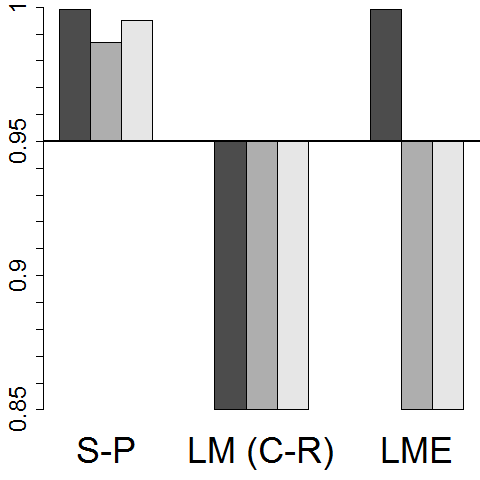}&\\\cdashline{1-4}\\
%%%%%%%%%%%%%%%%%%%%%%%%%%%
(V) & \includegraphics[height =.15\textheight]{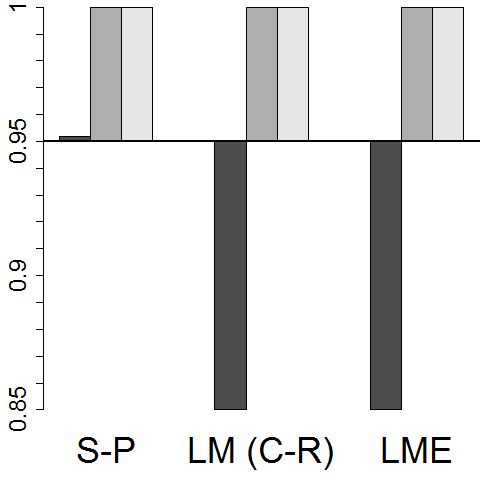}&
&
\includegraphics[height =.15\textheight]{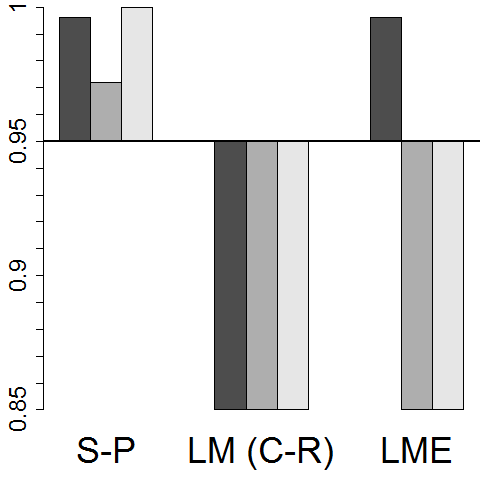}
&\\\cline{1-4}
%%%%%%%%%%%%%%%%%
%& \multicolumn{3}{c}{}
\end{tabular}
\end{center}
\end{figure}

%%%%%%%%%%%%%%%%%%%%%
%%%%%%%%%%%%%%%%%%%%%%
% -----------------------------------------------
%
%
\section{Discussion\label{section::discussion}}
%
%%%%%%%%%%%%%%%%%%%%%%
Randomization-based causal inference, originally developed by \cite{Neyman1923} and \cite{Neyman1935} in the context of completely randomized, randomized block, and Latin square designs,
(a)\,attributes the randomness in experimental data to the actual physical randomization of the experiments,
(b)\,allows for the definition of causal effects over a finite population of interest,
and
(c)\,extends the super-population notions of `unbiased' point estimates and `conservative' interval estimates to the finite-population settings.
Under this inferential framework, we proposed a new procedure for analyzing $2^2$ split-plot designs, and demonstrated its superior frequency coverage property over existing model-based alternatives.

Whereas the length limit restrains us from going any further,
the interested reader may find the following three directions, among others, worthy of future exploration.
%%%%%%%%%%%%%%%%%%%%%%%%%%%%%%%%%%%%%%%%%%
First, \cite{Rubin1978} and \cite{Tir2015} discussed Bayesian causal inference for completely randomized designs in the context of  treatment-control and $2^K$ factorial experiments respectively.
How to extend the same framework to split-plot designs in a way that also guarantees frequency properties is yet unclear.
%%%%%%%%%%%%%%%%%
Second,  \cite{Fisher1935} proposed the use of randomization test for sharp null hypotheses regarding the treatment effects at unit level.
Extension of such framework to split-plot designs should complement the current Neymanian framework's focus on the population-level parameters.
%%%%%%%%%%%%%%%%%%%%%%
Third,  extension of the current results to multi-level, multi-factor, or other more complex forms of split-plot designs, as documented in \cite{Federer2007}, would be of both theoretical and practical interest.

\bibliographystyle{imsart-nameyear}
\bibliography{Split-plot_2015}

\newpage
\appendix
\setcounter{section}{0}
\setcounter{lemma}{0}
\renewcommand {\thelemma} {A.\arabic{lemma}}

\setcounter{equation}{0}
\setcounter{section}{0}
\setcounter{figure}{0}
\setcounter{example}{0}
\setcounter{lemma}{0}
\setcounter{page}{1}

\renewcommand{\thelemma}{\Alph{section}.\arabic{lemma}}
\begin{center}
\textbf{\large SUPPLEMENTARY MATERIAL}
\end{center}
%%%%%%%%%%%%%%%%%%%%%%%%%%%
\section{Matrix Algebra}\label{sec::matrixAlgebra}
%%%%%%%%%%%%%%%%%%%%%%%%%%
%For any positive integer $p$, let
%${\vones}_p$ denote the $p$-dimensional vector with all entries equal to one,
%${\mathbf I}_p$ denote the $p$-dimensional identity matrix,
%${\mathbf J}_p = \vones_p\vones_p^{\T}$ denote the $p \times p$ matrix with all entries equal to one,
%and $\mP_p =  \mI_p - p^{-1}\mJ_p$ denote the $p$-dimensional projection matrix with column space orthogonal to $\vones_p$.
For any $p$-dimensional vector $\bm a = (a_1, \ldots, a_p)^\T$,
let $\textnormal{diag}\{\bm a\} = \textnormal{diag}\{a_1, \ldots, a_p\}$
%\,=\, \left( \begin{array}{llll}
%a_1&0&\ldots&0\\
%0&a_2&\ldots&0\\
%\vdots&\vdots&\ddots&\vdots\\
%0&0&0&a_n
%\end{array}\right)
%$$
denote the $p\times p$ diagonal matrix with $a_i$ as the $i$th diagonal entry.
For any arbitrary matrices (including vectors and scalars) $\mathbf{A}_1, \ldots,  \mathbf{A}_q$,
let
$$
\textnormal{Bdiag}\left\{ \mathbf{A}_1, \mathbf{A}_2, \ldots,  \mathbf{A}_q \right\} \,=\, \begin{pmatrix}
\mathbf{A}_1 & \mathbf{0} & \ldots & \mathbf{0}\\
\mathbf{0} & \mathbf{A}_2 & \ldots & \mathbf{0}\\
\vdots & \vdots & \ddots & \vdots \\
\mathbf{0} &\mathbf{0} &\ldots& \mathbf{A}_q
\end{pmatrix}
$$
denote the block-diagonal matrix with $\mathbf{A}_i$ as the $i$th diagonal block.

%%%%%%%%%%%%%%%%%%%%%%%%%%%%%%%%%%%%%%%%%%%%%%%%%%%%%%%%%%%%%%%%%%%%%%%%%%%%%%%%%%%%%%%%%%%%%%%%%%%%%%%%%%%%%%%%%%%%%%%%%%%%%%%%%
%
%	LEMMA_KRONECKER and Elementwise	
%
%%%%%%%%%%%%%%%%%%%%%%%%%%%%%%%%%%%%%%%%%%%%%%%%%%%%%%%%%%%%%%%%%%%%%%%%%%%%%%%%%%%%%%%%%%%%%%%%%%%%%%%%%%%%%%%%%%%%%%%%%%%%%%%%%
Before proceeding to the formal proofs,
let us first establish some properties of the Kronecker product $(\otimes)$ and the entrywise product $(\circ)$ that will be invoked repeatedly throughout this appendix.

\begin{lemma}
\label{lem::kronecker}
\begin{enumerate}
\item For any vectors $\bm a, \bm b \in \mathbb{R}^L$, and matrix $\mathbf{Q} \in \mathbb{R}^{L\times L}$,
\begin{equation}
\left(\bm a \bm b^\T\right) \circ \mathbf{Q} \,=\,
\textnormal{diag}\{\bm a\}\mathbf{Q} \textnormal{diag}\{\bm b\}\,.
 \label{eq::circToDot}
\end{equation}
%%%%%%%%%%%
\item For any random vectors $\bm X_1$ and $\bm X_2$, and constant vectors $\bm a$ and $\bm b$,
\begin{equation}
\label{eq::kronecker_cov}
\cov( \bm X_1\otimes \bm a, \bm X_2 \otimes \bm b)
\,=\,
\cov( \bm X_1,\bm X_2) \otimes ( \bm a\bm b^\T )\,.
\end{equation}
%%%%%%%%%%%%%%%%%%
\item
Let $\prod$ denote the usual matrix product. For any matrices $\mathbf{A}_i$ and $\mathbf{B}_i$ ($i = 1,\ldots,n$) such that $ \prod_{i=1}^n ( \mathbf{A}_i \otimes \mathbf{B}_i) $, $ \prod_{i=1}^n  \mathbf{A}_i$, and $\prod_{i=1}^n  \mathbf{B}_i$ are all well-defined,
\begin{equation}
\label{eq::Kronecker_prod}
\prod_{i=1}^n \left( \mathbf{A}_i \otimes \mathbf{B}_i\right)
\,=\, \left( \prod_{i=1}^n  \mathbf{A}_i \right) \otimes \left( \prod_{i=1}^n  \mathbf{B}_i \right).
\end{equation}
%%%%%%%%%%%%%%%%%%%%
\item
Given $\bm X_1, \ldots, \bm X_K \in \mathbb{R}^N$,
let $\mathbf{X} = (\bm X_1, \ldots, \bm X_K)$
be the $N\times K$ matrix with $\bm X_i$ as the $i$th column, and
$\widetilde{\mathbf{X}} = \textnormal{Bdiag}\left\{\bm X_1, \ldots, \bm X_K\right\}$
be the $(NK)\times K$ block-diagonal matrix with $\bm X_i$ as the $i$th diagonal block.
For any $K\times K$ matrix $\mathbf{A}$ and $N\times N$ matrix $\mathbf{B}$,
\begin{equation}
\label{eq::kronecker_zong}
\widetilde{\mathbf{X}}^\T \left( \mathbf{A} \otimes \mathbf{B} \right) \widetilde{\mathbf{X}}
%\,=\,
%\mathbf{A} \circ\{ \widetilde{\mathbf{X}}^\T \left( \mJ_{K}\otimes \mathbf{B} \right) \widetilde{\mathbf{X}} \}
\,=\,
\mathbf{A} \circ \left\{ \mathbf{X}^\T \mathbf{B} \mathbf{X} \right\}.
\end{equation}
%%%%%%%%%%%%%%%%%%%%%%%%%%%%%%%
%\item
%\begin{equation}
%\label{eq::Kronecker_rank}
%\mathrm{rank}(\mathbf{A} \otimes \mathbf{B}) \,=\, \mathrm{rank}\mathbf{(A)} \mathrm{rank} \mathbf{(B)}\,.
%\end{equation}
\end{enumerate}
\end{lemma}

%%%%%%%%%%%%%%%%%%%%%%%%%%%%%%%%%%%%%%%
% ========================================================
%
\section{Algebraic properties of the Science}
%
% ========================================================

%%%%%%%%%%%%%%%%%%%%%%%
\begin{proof}[Proof of Lemma \ref{def1}]
%%%%%%%%%%%%%%%%%%%%%%%
On the one hand,
strict additivity implies the existence of $c_k$  such that $\bm Y(k) = \bm Y(1) + c_k \vones_N$ ($k = 2, \ldots, K$).
We have $Y_i(k) = Y_i(1) + c_k$, $\bar Y(k) = \bar Y(1) + c_k$, and $Y_i(k) - \bar Y(k) = Y_i(1) - \bar Y(1)$.
This, coupled with the definition of $S^2(k,l)$, proves $S^2(k,l) = S^2(1,1)$ for all $(k,l)$.
%%%%%%%%%%%%%%%%%%
On the other hand, given $\mathbf{S^\text{$2$}} = S_0^2 {\mJ}_4$, we have
\begin{align*}
&\| \mP_N \{ \bm Y(k) - \bm Y(l)\} \|^2\\
&\quad=\,
\{\bm Y(k) - \bm Y(l) \}^{\T} \mP_N \{ \bm Y(k) - \bm Y(l)\} \\
&\quad=\,
\bm Y(k)^{\T} \mP_N \bm Y(k) + \bm Y(l)^{\T} \mP_N \bm Y(l) - 2 \bm Y(k)^{\T} \mP_N \bm Y(l) \\
&\quad=\,  (N-1)( S_{k,k}^2 + S_{l,l}^2 - 2 S_{k,l}^2 ) \,=\, (N-1) ( S_0^2 + S_0^2 - 2 S_0^2)
\,=\, 0
\end{align*}
for any $(k, l) \in \{1, \ldots, K\}^2$.
Therefore, $\mP_N \left\{ \bm Y(k) - \bm Y(l) \right\} = \bm 0_N$, and the difference
\begin{equation*}
 \bm Y(k) - \bm Y(l)=\mP_N\{ \bm Y(k) - \bm Y(l)\}+  \vones_N \{ \bar{Y}(k) - \bar{Y}(l)\}= \vones_N \{ \bar{Y}(k) - \bar{Y}(l)\}
\end{equation*}
equals $\bar{Y}(k) - \bar{Y}(l)$ in all dimensions.
This completes the proof.
\end{proof}

%%%%%%%%%%%%%%%%%%%%%%%%%%%%%%
%
% Block in terms of science
%
%%%%%%%%%%%%%%%%%%%%%%%%%%%%%%%
%\begin{lemma}
%\label{def2}
%A \emph{science} with block structure is between-block additive if and only if
%\begin{enumerate}
%\item ${\bm Y}_{\text{block}}(k)  = \left({Y}_{(1\cdot)}(k), \, \ldots, \, {Y}_{(W\cdot)}(k)\right)^\T$ $(k=1,\ldots, K)$ are perfectly correlated, or equivalently,
%\item all elements in $\mSb$ are the same, i.e., $\mSb= S^2_{\textrm{btw}} {\mJ}_K$.
%for some non-negative number $S^2_{\textnormal{btw}}$.
%\end{enumerate}
%A \emph{science} with block structure is within-block additive if and only if
%\begin{enumerate}
%\item $\mY_w(k)$ $(k=1,\ldots, K)$ are perfectly correlated for all $w \in \{1,\ldots, W\}$, or equivalently,
%\item all elements in $\mSi$ are the same, i.e., $\mSi= S^2_{\textrm{in}} {\mJ}_K$
%for some non-negative number $S^2_{\textnormal{in}}$.
%\end{enumerate}
%\end{lemma}

%%%%%%%%%%%%%%%%%%%%%%%
\begin{proof}[Proof of Lemma \ref{def2}]
%%%%%%%%%%%%%%%%%%%%%%%
The equivalence follows immediately from the definitions of $\mSb$ and $\mSi$.
\end{proof}

%%%%%%%%%%%%%%%%%%%%%%%%%%%%%%
%
% Strict in terms of factorial
%
%%%%%%%%%%%%%%%%%%%%%%%%%%%%%%%
%\begin{lemma}\label{def3}
%The $N\times4$ \emph{science} matrix of a $2^2$ factorial experiment is strictly additive if and only if all three unit-level factorial effects $\tau_{i\text{-}F}$ $(F \in \mathcal{F})$ are constant across all units. Mathematically,
%\begin{itemize}
%\item[]\textbf{Condition 1}. $\tau_{i\text{-}F} = \tau_F$ ($i = 1, \ldots, N$, $F = A, B, AB$)
%\item[]\textbf{Condition 2}. $S^2_F=0$ $(F \in \mathcal{F})$
%\end{itemize}
%give two equivalent sufficient and necessary conditions of strictly additivity.
%\end{lemma}

%%%%%%%%%%%%%%%%%%%%%%%
\begin{proof}[Proof of Lemma \ref{def3}]
%%%%%%%%%%%%%%%%%%%%%%%
On the one hand, under strict additivity,
there exists some constants $c_2$, $c_3$, $c_4$ such that
$Y_i(2) = Y_i(1) + c_2$, $Y_i(3)= Y_i(1) + c_3$, $Y_i(4)=Y_i(1) +c_4$
for all $i \in \{1, \ldots, N\}$, and we can write $\bm Y_i$ as
$$\bm Y_i
\,=\, \left( Y_i(1), Y_i(2), Y_i(3), Y_i(4) \right)^\T
\,=\, Y_i(1) \vones_4 + (0, c_2, c_3, c_4)^\T\,.
$$
Thus, for any $F \in \mathcal{F}$,
\begin{align*}
\bm \tau_{i\text{-}F}
&=\, 2^{-1}\vg_F^\T\vY_i \,=\, 2^{-1}\vg_F^\T\left\{Y_i(1) \vones_4 + (0, c_2, c_3, c_4)^\T\right\}\\
&=\,  2^{-1}Y_i(1) \vg_F^\T\vones_4 + 2^{-1}\vg_F^\T(0, c_2, c_3, c_4)^\T
\,\,=\,\, 2^{-1}\vg_F^\T(0, c_2, c_3, c_4)^\T
\end{align*}
are constant for all $i \in \{1, \ldots, N\}$.
This proves the necessity of the condition.

On the other hand,
given $\tau_{i\text{-}F}$ being constant across all units ($F \in \mathcal{F}$),
it follows from \eqref{eq::YasTau} in the main text that
\begin{align}
\vY_{i}
&=\, \mD (2\mu_{i},\,  \tau_{i\text{-}A},\, \tau_{{i\text{-}B}},\, \tau_{i\text{-}AB})^\T
\,=\,  \mD (2\mu_{i},\,  \tau_{A},\, \tau_{B},\, \tau_{AB})^\T\,,\label{eq::toQuote}\\
\mY
&=\,  (2\bm \mu, \, \tau_{A} \vones_N, \, \tau_{B} \vones_N, \, \tau_{AB} \vones_N)\mD^\T\,.\nonumber
\end{align}
We have
\begin{align*}
\mP_N\mY
&=\, \mP_N(2\bm \mu, \, \tau_{A} \vones_N, \, \tau_{B} \vones_N, \, \tau_{AB} \vones_N)\mD^\T\\
&=\,  (2\mP_N\bm \mu, \, \tau_{A}  \mP_N\vones_N, \, \tau_{B}  \mP_N\vones_N, \, \tau_{AB}  \mP_N\vones_N)\mD^\T\\
&=\, (2\mP_N\bm \mu, \, \vzeros_N, \,  \vzeros_N, \,  \vzeros_N)\mD^\T
\,=\,  (2\mP_N\bm \mu, \, \vzeros_N, \,  \vzeros_N, \,  \vzeros_N)(\vones_4,\, \vg_A,\, \vg_B, \, \vg_{AB})^\T\\
&=\, 2\mP_N\bm \mu \vones_4^\T\,,\\
%%%%%%%%%%%%%%%%%%%%%%%%%%%%%%%%%%%%
\mS
&=\, (N-1)^{-1} \mY^\T \mP_N \mY
\,=\,(N-1)^{-1}( \mP_N \mY)^\T  (\mP_N \mY) \\
&=\, (N-1)^{-1}( 2\mP_N\bm \mu \vones_4^\T)^\T ( 2\mP_N\bm \mu \vones_4^\T)
%&=\,  \frac{4}{N-1} (\vones_4 \bm \mu^\T \mP_N^\T)(\mP_N\bm \mu \vones_4^\T)
%\,=\, \frac{4}{N-1} (\vones_4 \bm \mu^\T \mP_N)(\mP_N\bm \mu \vones_4^\T)
%\,=\,  \frac{4}{N-1} \vones_4 (\bm \mu^\T \mP_N^\T\mP_N\bm \mu) \vones_4^\T\\
\,=\, 4(N-1)^{-1} \vones_4 (\bm \mu^\T \mP_N\bm \mu) \vones_4^\T\\
&=\,   \frac{4(\bm \mu^\T \mP_N\bm \mu)}{N-1} \vones_4 \vones_4^\T
\,=\, \frac{4(\bm \mu^\T \mP_N\bm \mu)}{N-1} \mJ_4\,.
\end{align*}
By Lemma \ref{def1}, this implies strict additivity with $S_0^2 = 4(N-1)^{-1}\bm \mu^\T \mP_N\bm \mu$.
\end{proof}

%%%%%%%%%%%%%%%%%%%%%%%%%%%%%%
%
% block in terms of factorial
%
%%%%%%%%%%%%%%%%%%%%%%%%%%%%%%%
%\begin{lemma}\label{def4}
%The $N\times4$ \emph{science} matrix of a $2^2$ factorial experiment with block structure is between-block additive if and only if all three block average factorial effects $\tau_{(w)\text{-}F}$ $(F \in \mathcal{F})$ are constant across all blocks. Mathematically,
%\begin{itemize}
%\item[]\textbf{Condition 1}. $\tau_{(w)\text{-}F} = \tau_F$ ($ w = 1, \ldots, W$, $F = A, B, AB$).
%\item[]\textbf{Condition 2}. $S^2_{A\text{-btw}} = S^2_{B\text{-btw}} = S^2_{AB\text{-btw}} = 0$.
%\end{itemize}
%give two equivalent sufficient and necessary conditions of between-block additivity.
%\end{lemma}
%%%%%%%%%%%%%%%%%%%%%%%%%%%%%%%%%%%%%%
%\begin{lemma}\label{def5}
%The $N\times4$ \emph{science} matrix of a $2^2$ factorial experiment with block structure is within-block additive if and only if all three unit-level factorial effects $\tau_{(wm)\text{-}F}$ $(F \in \mathcal{F})$ are constant within each block. Mathematically,
%\begin{itemize}
%\item[]\textbf{Condition 1}. $\tau_{(wm)\text{-}F} = \tau_{(w)\text{-}F}$ ($ w = 1, \ldots, W$, $F = A, B, AB$).
%\item[]\textbf{Condition 2}. $S^2_{A\text{-in}} = S^2_{B\text{-in}} = S^2_{AB\text{-in}} = 0$.
%\end{itemize}
%give two equivalent sufficient and necessary conditions of within-block additivity.
%\end{lemma}

%%%%%%%%%%%%%%%%%%%%%%%
\begin{proof}[Proof of Lemma \ref{def4}]
%%%%%%%%%%%%%%%%%%%%%%%
Treating block $w$ as unit $i$ in Lemma \ref{def3} proves the between-block part,
whereas the within-block part follows straightforwardly from
\begin{align*}
\vY_{(wm)}
&=\, \mD (2\mu_{(wm)}, \tau_{(wm)\text{-}A},\tau_{{(wm)\text{-}B}}, \tau_{(wm)\text{-}AB})^\T\\
&=\, \mD (2\mu_{(wm)}, \tau_{(w)\text{-}A}, \tau_{(w)\text{-}B}, \tau_{(w)\text{-}AB})^\T
\end{align*}
as a modification of \eqref{eq::toQuote}.
\end{proof}

%%%%%%%%%%%%%%%%%%%%%%%%%%%%%%%%%%%%%%%
% ========================================================
%
\section{Sampling moments of the assignment vectors}
%
% ========================================================
%%%%%%%%%%%%%%%%%%%%%%%%%%%%%%%%%%%%%%%%%%
% =============================
%	LEMMA_COMPLETE-RANDOM-VAR
% =============================
\begin{lemma}
\label{lem::I}
For a completely randomized design with $N$ experimental units, $K$ different treatments, and planned treatment arm sizes $N_k$ with $\sum_{k=1}^K N_k = N$,
the treatment assignment vectors
${\bm Z}(k) = ( {I}_{\{T_1 =k \}}, \ldots, {I}_{\{ T_N = k \}})^{\T}$
satisfy
\begin{equation}
\label{eq::keql}
\E_{\textsc{c-r}}\{{\bm Z}(k)\} \,=\, \frac{N_k}{N}\vones_{N}\,,\,\,\,
\cov_{\textsc{c-r}}\{{\bm Z}(k)\}  \,=\, \frac{N_k(N-N_k)}{N(N-1)} \mP_{N}\,.
\end{equation}
\end{lemma}

\begin{proof}
For any given unit $i$, the probability of it receiving treatment $k$ equals $\pr_{\textsc{c-r}}(T_i =k) = {N_k}/{N}$.
The indicator ${I}_{\{T_i =k \}}$ thus follows a Bernoulli(${N_k}/{N}$) distribution with $\E_{\textsc{c-r}}( {I}_{\{T_i =k \}} ) = {N_k}/{N}$ --- from which follows immediately the first equality in \eqref{eq::keql} --- and
\begin{align}
\label{eq::EVar_CR}
\var_{\textsc{c-r}}({I}_{\{T_i =k \}}) \,=\,\frac{N_k}{N}\left(1-\frac{N_k}{N}\right).
%\,=\, \frac{N_k(N-N_k)}{N(N-1)} - \frac{N_k(N-N_k)}{N^2(N-1)}
\end{align}
%For any two different units $i$ and $j$ $(i \neq j)$, the probability of them both receiving treatment $k$ equals
%\begin{align*}
%\pr_{\textsc{c-r}}( T_i =T_j = k)   \,=\, \frac{N_k(N_k-1)}{N(N-1)}\quad (i, j = 1, \ldots, N)\,.
%\end{align*}
The covariances of any two dimensions $i$ and $j$ $(i \neq j)$ in ${\bm Z}(k)$ satisfy
\begin{align}
\label{eq::cov_CR}
&\cov_{\textsc{c-r}}({I}_{\{T_i =k \}}, {I}_{\{T_j =k \}})\\
&\quad=\,
  \E_{\textsc{c-r}}({I}_{\{T_i =k \}}{I}_{\{T_j =k \}}) - \E_{\textsc{c-r}}({I}_{\{T_i =k \}})\E_{\textsc{c-r}}({I}_{\{T_j =k \}})\nonumber\\
&\quad=\,
%  \E_{\textsc{c-r}}({I}_{\{T_i =T_j =k \}})  -\frac{N_k^2}{N^2}\,=\,
\pr_{\textsc{c-r}}(T_i =T_j =k) -\frac{N_k^2}{N^2}\,=\, \frac{N_k(N_k-1)}{N(N-1)} -\frac{N_k^2}{N^2}\nonumber\\
&\quad=\, - \frac{N_k(N-N_k)}{N^2(N-1)}\,.\nonumber
\end{align}
Given the right-hand side of \eqref{eq::EVar_CR} satisfies
$$\frac{N_k}{N}\left(1-\frac{N_k}{N}\right) \,=\, \frac{N_k(N-N_k)}{N^2}
\,=\,  \frac{N_k(N-N_k)}{N(N-1)} - \frac{N_k(N-N_k)}{N^2(N-1)}\,,$$
organizing \eqref{eq::EVar_CR} and \eqref{eq::cov_CR} into variance-covariance matrix form yields
\begin{align*}
\cov_{\textsc{c-r}}\left\{{\bm Z}(k)\right\} \,=\, \frac{N_k(N-N_k)}{N(N-1)} \mI_{N} - \frac{N_k(N-N_k)}{N^2(N-1)} \mJ_{N} \,=\, \frac{N_k(N-N_k)}{N(N-1)} \mP_{N}\,.
\end{align*}
This completes the proof.
\end{proof}
%%%%%%%%%%%%%%%%%%%%
% =========================
% Cov_CR(Z)
% ===========================
%%%%%%%%%%%%%%%%%%%%
%\begin{lemma}
%\label{lem::covZ_CR}
%For a $2^2$ completely randomized design, we have
%\begin{eqnarray*}
%\E_{\textsc{c-r}}\left( \bm Z^* \right) \,=\,  N^{-1} \vones_{4N}\,,\quad
%\cov_{\textsc{c-r}} \left( \bm Z^* \right)  \,=\, \mCf \otimes \mP_N\,,
%\end{eqnarray*}
%where the coefficient matrix
%\begin{eqnarray*}
% \mCf  &=&
%\frac{1}{N(N-1)}\left( \textnormal{diag}\left\{ \frac{N}{N_1}, \frac{N}{N_2}, \frac{N}{N_3}, \frac{N}{N_4}\right\} - \mJ_4\right).
%\end{eqnarray*}
%When the design is balanced, the coefficient matrix simplifies to
%$$\mCf \,=\,\frac{4}{N(N-1)}\mP_4\,.$$
%\end{lemma}

%%%%%%%%%%%%%%%%%%%%%
\begin{proof}[Proof of Lemma \ref{lem::covZ_CR}]
%%%%%%%%%%%%%%%%%%%%%%
Given any two treatments $k$ and $l$ $(k \neq l)$,
the covariance of the entries of $\bm Z(k)$ and $\bm Z(l)$ can be computed as
\begin{align*}
\cov_{\CR}( I_{\{T_i= k \}},I_{\{T_j = l \}})
&=\, \E_{\CR} (I_{\{ T_i= k, \, T_j = l \}})  - \E_{\CR} ( I_{\{T_i= k \}})\E_{\CR} ( I_{\{T_j= l \}})\\
&=\, \left\{
\begin{array}{ll}
%- \frac{N_k N_l }{N^2},
-N_k N_l/N^2, & \text{if $i = j$,}\\
%\frac{N_k N_l}{N^2(N-1)},
N_k N_l/\{N^2(N-1)\},&\text{if $i\neq j$.}
\end{array}
\right.
 \end{align*}
The variance-covariance matrix of $\bm Z(k)$ and $\bm Z(l)$ is thus
\begin{equation*}
\cov_{\CR}\{ \bm Z(k),\bm Z(l) \} \,=\, - \frac{N_k N_l}{N(N-1)}\mP_{N}\quad (k \neq l)\,.
\end{equation*}
This, together with \eqref{eq::keql}, yields
\begin{align*}
\cov_{\CR} \left\{ N_k^{-1}\bm Z(k) , N_k^{-1}\bm Z(k)  \right\}
&=\, \frac{1}{N(N-1)}\left( \frac{N}{N_k}-1\right) \mP_{N}\,,\\
\cov_{\CR} \left\{ N_k^{-1}\bm Z(k) , N_l^{-1}\bm Z(l)  \right\}
&=\, - \frac{1}{N(N-1)} \mP_{N} \quad (k \neq l)
\end{align*}
and
\begin{align*}
\cov_{\textsc{c-r}} \left( \bm Z^* \right)
&=\,
\cov_{\CR} \left\{ ( N_1^{-1}\vZ(1)^\T, N_2^{-1}\vZ(2)^\T, N_3^{-1}\vZ(3)^\T, N_4^{-1}\vZ(4)
) \right\}\\
&=\,
\frac{1}{N(N-1)}
\begin{pmatrix}
\frac{N}{N_{1}}-1&-1&-1&-1\\
-1&\frac{N}{N_{2}}-1&-1&-1\\
-1&-1&\frac{N}{N_{3}}-1&-1\\
-1&-1&-1&\frac{N}{N_{4}}-1
\end{pmatrix}
\otimes \mP_N \\
&=\,
\frac{1}{N(N-1)}\left( \textnormal{diag}\left\{ \frac{N}{N_1}, \frac{N}{N_2}, \frac{N}{N_3}, \frac{N}{N_4}\right\} - \mJ_4\right) \otimes \mP_N
%\,=\,\mCf \otimes \mP_N
\,.
\end{align*}
This completes the proof.
\end{proof}
%%%%%%%%%%%%%%%%%%%%%%%%%%%%%%%%%%

Under the $2^2$ split-plot design qualified by Definition \ref{def::SP},
let $A_w$ be the level of factor $A$ for whole-plot $w$ in the whole-plot randomization,
and let $B_{(wm)}$ be the level of factor $B$ for sub-plot $(wm)$ in the sub-plot randomization.
Recall from the main text that $T_{(wm)}=k$ if sub-plot $(w, m)$ receives treatment $k$, and that $g_A(T_{(wm)})$ and $g_B(T_{(wm)})$ indicate the levels of factors $A$ and $B$ in treatment $T_{(wm)}$ respectively.
We have
\begin{equation}
\label{eq::transition}
A_w \,=\,g_A(T_{(wm)})\,,\quad B_{(wm)}\,=\,g_B(T_{(wm)}) \,.
\end{equation}

For $z \in \{-1,+1\}$, define
$\vZ_{A}(z) = ( I_{\{A_1= z \}}, \ldots, I_{\{A_W = z \}})^\T \in \{0,1\}^W$
and $\vZ_{B}(z) = ( I_{\{B_{(11)}= z \}}, \ldots, I_{\{B_{(WM)}= z \}})^\T \in \{0,1\}^N$ as the factorial analogues of $\vZ(k)$ for the whole-plot and sub-plot randomizations respectively.

For treatment $k$ with $g_{A}(k) \in \{-1,+1\}$ level of factor $A$ and $g_{B}(k) \in \{-1,+1\}$ level of factor $B$,
introduce shorthand notations $\vZ_{A}(k) = \vZ_{A}\{g_{A}(k)\}$  to indicate the whole-plots that receive $g_{A}(k)$ level of factor $A$, and
$\vZ_{B}(k) = \vZ_{B}\{g_{B}(k)\}$ to indicate the sub-plots that receive $g_{B}(k)$ level of factor $B$.
Further define
\begin{itemize}
\item $W_{(1)} = W_{(2)} = W_{-1}$, $W_{(3)} = W_{(4)} =W_{+1}$ such that $W_{(k)}$ indicates the total number of whole-plots in which treatment $k$ will be observed, and
\item $M_{(1)} = M_{(3)} = M_{-1}$, $M_{(2)} = M_{(4)} = M_{+1}$ such that, for each whole-plot that receives $g_A(k)$ level of factor $A$ in the whole-plot randomization, $M_{(k)}$ of its $M$ sub-plots end up in treatment arm $k$.
\end{itemize}
\noindent The following lemma gives the covariance structures of $\vZ_{A}(k)$ and $\vZ_{B}(k)$ as a central building block for the proof of Theorem \ref{lem::covZ_SP}.
\begin{lemma}
\label{lem::lemForProvingCovZ_SP}
$\bm Z(k)$ can be expressed  as
\begin{equation}
\bm Z(k)
\,=\, \left\{ \vZ_{A}(k) \otimes \vones_M\right\}\circ \vZ_{B}(k)
\,=\,  \left[ \textnormal{diag}\left\{ \vZ_{A}(k) \right\}\otimes \mIs\right] \vZ_{B}(k)\,,\label{eq::Z_decomp0}
\end{equation}
where $\{\vZ_{A}(k)\}_{k=1}^4$ and $\{\vZ_{B}(k)\}_{k=1}^4$ are mutually independent with expectations and covariances
\begin{eqnarray}
&&\cov_{\SP} \left\{ \vZ_{A}(k), \vZ_{A}(l)\right\}
\,=\, g_{A}(k)g_{A}(l)\frac{W_{+1}W_{-1}}{W(W-1)}  \mPw\label{eq::Cov_ZAk}, \\
&&\E_{\SP}\left\{\vZ_{B}(k)\right\} \,=\, \frac{M_{(k)}}{M} \vones_N\,,\label{eq::E_ZBk}\\
&&
\cov_{\SP}\left\{\vZ_{B}(k), \vZ_{B}(l) \right\}\,=\, g_{B}(k)g_{B}(l) \frac{M_{+1}M_{-1}}{M(M-1)}\mPi\,. \label{eq::Cov_ZBk}
\end{eqnarray}
\end{lemma}

%%%%%%%%%%%%%%%%%%%%%%%%%%%%%%%%%
\begin{proof}[Proof of Lemma \ref{lem::lemForProvingCovZ_SP}]
It follows from identities
\begin{equation*}
I_{\{T_{(wm)} = k \}}= I_{\{A_{w} = g_{A}(k)\}} I_{\{B_{(wm)}= g_{B}(k) \}}\quad(w=1,\ldots,W; m=1,\ldots,M)
\end{equation*}
that
\begin{align*}
\bm Z(k)
&=\, (I_{\{T_{(11)} = k \}}, \ldots, I_{\{T_{(WM)} = k \}})^\T \nonumber\\
&=\, (I_{\{A_{1} = g_{A}(k) \}} \vones_M^\T, \ldots, I_{\{A_{W} = g_{A}(k) \}}\vones_M^\T)^\T
\circ  (I_{\{B_{(11)} = g_{B}(k) \}}, \ldots, I_{\{B_{(WM)} = g_{B}(k) \}})^\T\nonumber\\
&=\, \left\{ \vZ_{A}\{g_{A}(k)\} \otimes \vones_M\right\} \circ\vZ_{B}\{g_{B}(k)\}
\,=\, \left\{ \vZ_{A}(k) \otimes \vones_M\right\}\circ \vZ_{B}(k)\\
&=\, \left[ \textnormal{diag}\left\{ \vZ_{A}(k) \otimes \vones_M \right\}\right]  \vZ_{B}(k)
\,=\, \left[ \textnormal{diag}\left\{ \vZ_{A}(k) \right\}\otimes \mIs\right] \vZ_{B}(k)\,.
\end{align*}
This proves \eqref{eq::Z_decomp0}.

Applying Lemma \ref{lem::I} to the whole-plot randomization yields
\begin{equation*}
\label{lem::11k}
\E_{\SP}\{ \vZ_{A}(k) \}\,=\,\frac{W_{(k)}}{W} \vones_W,\,\,
\cov_{\SP}\left\{\vZ_{A}(k) \right\} \,=\, \frac{W_{+1}W_{-1}}{W(W-1)}\mPw,
\end{equation*}
and it follows immediately from
$\vZ_{A}(-1)  = \vones_W - \vZ_{A}(+1)$ that
\begin{equation*}
\label{lem::11kl}
\cov_{\SP}\left\{\vZ_{A}(-1),  \vZ_{A}(+1)\right\}=
- \cov_{\SP}\left\{\vZ_{A}(+1)\right\} =- \frac{W_{+1}W_{-1}}{W(W-1)}\mPw\,.
\end{equation*}
As a result, we have
\begin{align*}
&\cov_{\SP}\{ \vZ_{A}(k), \vZ_{A}(l)\}
\,=\,
\cov_{\SP}[ \vZ_{A}\{g_{A}(k)\},\vZ_{A}\{g_{A}(l)\}]\nonumber\\
&\quad =\,
(-1)^{I_{\{g_{A}(k) \neq g_{A}(l)\}}} \frac{W_{+1}W_{-1}}{W(W-1)}  \mPw
\,=\, g_{A}(k)g_{A}(l)\frac{W_{+1}W_{-1}}{W(W-1)}\,,
\end{align*}
where the last equality holds because $g_{A}(k) \neq g_{A}(l)$ implies $\{g_{A}(k), g_{A}(l)\}=\{-1, +1\}$. This proves \eqref {eq::Cov_ZAk}.

Last, but not least,
to better understand the covariance structure of $\vZ_{B}(k)$, let us introduce
$\vZ_{B, (w)}(k) = ( I_{\{B_{(w1)}= g_{B}(k) \}}, \ldots, I_{\{B_{(wS)}= g_{B}(k) \}})^\T$
as the $M$-dimensional sub-vector of $\vZ_{B}(k)$ that corresponds to whole-plot $w$.
For any fixed $k \in \{1,2,3,4\}$, the sub-plot randomization mechanism renders $\vZ_{B, (w)}(k)$ ($w = 1,\ldots,W$) iid with
\begin{equation*}
\E_{\SP}\{\vZ_{B, (w)}(k)\}\,=\, \frac{M_{(k)}}{M} \vones_M\,, \quad
\cov_{\SP}\{\vZ_{B, (w)}(k)\} \,=\, \frac{M_{+1}M_{-1}}{M(M-1)} \mP_M
\end{equation*}
as follows from Lemma \ref{lem::I}. The expectation and covariance of $
\vZ_{B}(k)$ can thus be computed block by block as
\begin{equation*}
\E_{\SP}\{\vZ_{B}(k)\}
\,=  \left( \E_{\SP}\{\vZ_{B, (1)}(k)\}^\T, \ldots, \E_{\SP}\{\vZ_{B, (W)}(k)\}^\T \right)^\T \,=\, \frac{M_{(k)}}{M} \vones_N\,,
\end{equation*}
which proves \eqref{eq::E_ZBk}, and
\begin{align*}
&\cov_{\SP}\{\vZ_{B}(k) \}
\,=\,
\textnormal{Bdiag}\left[\cov_{\SP}\{\vZ_{B, (1)}(k) \}, \, \ldots, \, \cov_{\SP}\{\vZ_{B, (W)}(k)\}\right]\\
&\quad=\, \mIw \otimes \cov_{\SP}\{\vZ_{B, (1)}(k)\}
\,=\, \frac{M_{+1}M_{-1}}{M(M-1)}\mIw \otimes\mP_M
\,=\,  \frac{M_{+1}M_{-1}}{M(M-1)}\mPi\,.
\end{align*}
Thus,
\begin{equation}
\label{eq::Cov_ZBk_1}
\cov_{\SP}\left\{\vZ_{B}(-1)\right\} \,=\, \cov_{\SP}\left\{\vZ_{B}(+1) \right\} \,=\,  \frac{M_{+1}M_{-1}}{M(M-1)}\mPi\,,
\end{equation}
and it follows immediately from identity $\vZ_{B}(+1) = \vones_N - \vZ_{B}(-1)$ that
\begin{equation}
\label{eq::Cov_ZBk_2}
\cov_{\SP}\{\vZ_{B}(-1),  \vZ_{B}(+1) \}
\,=\, - \frac{M_{+1}M_{-1}}{M(M-1)}\mPi\,.
\end{equation}
Finally, the fact that $g_{B}(k)g_{B}(l)$ equals 1 if $g_{B}(k) = g_{B}(l)$ and equals $-1$ if $g_{B}(k) \neq g_{B}(l)$ allows us to unify  \eqref{eq::Cov_ZBk_1} and \eqref{eq::Cov_ZBk_2} into one formula  as
\begin{align*}
\cov_{\SP}\{\vZ_{B}(k), \vZ_{B}(l) \}
&=\, \cov_{\SP}[\vZ_{B}\{g_{B}(k)\},\vZ_{B}\{g_{B}(l)\}]\\
&=\, g_{B}(k)g_{B}(l) \frac{M_{+1}M_{-1}}{M(M-1)}\mPi\,,
\end{align*}
which proves \eqref{eq::Cov_ZBk}.
This completes the proof of Lemma \ref{lem::lemForProvingCovZ_SP}.
\end{proof}
%%%%%%%%%%%%%%%%%%%%%%%%%%%%%%%%%%

\begin{proof}[Proof of Theorem \ref{lem::covZ_SP}]
We again approach the mean and variance-covariance matrix of
$\bm Z^*$
%$
%\bm Z^* =( N_1^{-1}\vZ(1)^\T,  N_2^{-1}\vZ(2)^\T,  N_3^{-1}\vZ(3)^\T, N_4^{-1}\vZ(4)^\T)^\T
%$
from those of the $\vZ(k)$ ($k=1,2,3,4$).

In particular, let $\mathcal{Z}_A = \{ \vZ_A(k) \}_{k=1}^4$.
The law of iterated expectations allows us to decompose
the covariance of $\vZ(k)$ and $\vZ(l)$  into
\begin{align}
\label{eq::ECCE}
\cov_{\SP} \left\{ \vZ(k),\, \vZ(l)\right\}
=&\,\,
\cov_{\SP} \left[
  \E_{\SP}\left\{ \vZ(k) \mid \mathcal{Z}_A \right\},
  \E_{\SP}\left\{ \vZ(l) \mid \mathcal{Z}_A\right\} \right]\\
&+\,
\E_{\SP} \left[ \cov_{\SP} \left\{ \vZ(k), \vZ(l) \mid \mathcal{Z}_A\right\}\right]\nonumber.
\end{align}
Refer to the two components on the right as the \emph{covariance of expectations} and the \emph{expectation of covariance}, respectively.
Given
\begin{align*}
  \E_{\SP}\left\{ \vZ(k) \mid \mathcal{Z}_A \right\}
\,\overset{\eqref{eq::Z_decomp0}}{=}& \E_{\SP} \left[ \{\vZ_{A}(k) \otimes \vones_M\}\circ \vZ_{B}(k) \mid \mathcal{Z}_A\right]\nonumber\\
%%%%%%%%%%%%%%%%%%
\,\overset{\phantom{\eqref{eq::Cov_ZBk}}}{=}&
\{\vZ_{A}(k) \otimes \vones_M\} \circ \E_{\SP} \{\vZ_{B}(k) \mid \mathcal{Z}_A\}\nonumber\\
%%%%%%%%%%%%%%%%%%
\,\overset{\phantom{\eqref{eq::Cov_ZBk}}}{=}&  \{\vZ_{A}(k) \otimes \vones_M\} \circ \E_{\SP} \{\vZ_{B}(k)\}\nonumber\\
%%%%%%%%%%%%%%%%%%%%%%
\,\overset{\eqref{eq::E_ZBk}}{=}& \{\vZ_{A}(k) \otimes \vones_M\} \circ \left(\frac{M_{(k)}}{M} \vones_N\right)\,=\, \frac{M_{(k)}}{M} \{\vZ_{A}(k) \otimes \vones_M\}\,,
\end{align*}
we have
\begin{align}
\cov_{\SP}[
  &\E_{\SP}\{ \vZ(k) \mid \mathcal{Z}_A \},
  \E_{\SP}\{ \vZ(l) \mid \mathcal{Z}_A\}]\label{eq::firstTerm}\\
&\overset{\phantom{\eqref{eq::Cov_ZBk}}}{=}\frac{M_{(k)}M_{(l)}}{M^2}  \cov_{\SP}\{ \vZ_{A}(k)\otimes \vones_M,\,\vZ_{A}(l)\otimes \vones_M\}\nonumber\\
%%%%%%%%%%%%%%%%%%%%
&\overset{\phantom{\eqref{eq::Cov_ZBk}}}{=}
\frac{M_{(k)}M_{(l)}}{M^2} \cov_{\SP} \left\{ \vZ_{A}(k),\, \vZ_{A}(l)\right\} \otimes \mJs\nonumber\\
%%%%%%%%%%%%%%%%%%%%
&\overset{\eqref{eq::Cov_ZAk}}{=}\,
\frac{M_{(k)}M_{(l)}}{M^2} g_{A}(k)g_{A}(l) \frac{W_{+1}W_{-1}}{W(W-1)}\mPw \otimes \mJs
\nonumber\\
%%%%%%%%%%%%%%%%%%%%
&\overset{\phantom{\eqref{eq::Cov_ZBk}}}{=}\,
g_{A}(k)g_{A}(l) \frac{W_{+1}W_{-1}M_{(k)}M_{(l)}}{N(W-1)}\mPb\,.\nonumber
\end{align}
This gives the covariance of expectations component of \eqref{eq::ECCE}.
%%%%%%%%%%%%%%%%%%%%%%%
%%%%%%%%%%%%%%%%%%%%%%%
%%%%%%%%%%%%%%%%%%%%%%%
Likewise, given
\begin{align*}
&\cov_{\SP} \left\{ \vZ(k), \vZ(l) \mid \mathcal{Z}_A\right\} \nonumber \\
\overset{\eqref{eq::Z_decomp0}}{=}&\cov_{\SP} \left\{\left[ \textnormal{diag}\left\{ \vZ_{A}(k) \right\}\otimes \mIs\right] \vZ_{B}(k), \, \left[ \textnormal{diag}\left\{ \vZ_{A}(l) \right\}\otimes \mIs\right] \vZ_{B}(l) \mid \mathcal{Z}_A\right\}\nonumber\\
\,\overset{\phantom{\eqref{eq::Cov_ZBk}}}{=}& \left[ \textnormal{diag}\left\{ \vZ_{A}(k) \right\}\otimes \mIs\right] \cov_{\SP} \left\{\vZ_{B}(k), \vZ_{B}(l)  \mid\mathcal{Z}_A \right\}\left[ \textnormal{diag}\left\{ \vZ_{A}(l) \right\}\otimes \mIs\right]\nonumber\\
%%%%%%%%
\,\overset{\phantom{\eqref{eq::Cov_ZBk}}}{=}&  \left[ \textnormal{diag}\left\{ \vZ_{A}(k) \right\}\otimes \mIs\right] \cov_{\SP} \left\{\vZ_{B}(k), \vZ_{B}(l) \right\}\left[ \textnormal{diag}\left\{ \vZ_{A}(l) \right\}\otimes \mIs\right]\nonumber\\
%%%%%%%%%%%%%%%%%%%
\,\overset{\eqref{eq::Cov_ZBk}}{=}&
 \left[ \textnormal{diag}\left\{ \vZ_{A}(k) \right\} \otimes \mIs \right]
\left\{ g_{B}(k)g_{B}(l)\frac{M_{+1}M_{-1}}{M(M-1)}\mIw \otimes\mP_M\right\}
\left[ \textnormal{diag}\left\{ \vZ_{A}(l) \right\} \otimes \mIs \right]\nonumber\\
\,\overset{\phantom{0}\eqref{eq::Kronecker_prod}}{=}&
g_{B}(k)g_{B}(l)\frac{M_{+1}M_{-1}}{M(M-1)}[ \textnormal{diag}\left\{ \vZ_{A}(k) \right\} \otimes \mIs ]
( \mIw \otimes\mP_M)
[ \textnormal{diag}\left\{ \vZ_{A}(l) \right\} \otimes \mIs]\,,
\end{align*}
we have
\begin{align*}
&\E_{\SP} \left[\left\{g_{B}(k)g_{B}(l)\frac{M_{+1}M_{-1}}{M(M-1)}\right\}^{-1} \cov_{\SP} \left\{ \vZ(k), \vZ(l) \mid \mathcal{Z}_A\right\}\right]\nonumber\\
%&=\, \E_{\SP} \left[ \cov_{\SP} \left\{\left[ \textnormal{diag}\left\{ \vZ_{A}(k) \right\}\otimes \mIs\right] \vZ_{B}(k), \, \left[ \textnormal{diag}\left\{ \vZ_{A}(l) \right\}\otimes \mIs\right] \vZ_{B}(l) \mid \mathcal{Z}_A\right\}\right]\nonumber\\
%&=\,  \E_{\SP} \left( \left[ \textnormal{diag}\left\{ \vZ_{A}(k) \right\}\otimes \mIs\right] \cov_{\SP} \left\{\vZ_{B}(k), \vZ_{B}(l)  \mid\mathcal{Z}_A \right\}\left[ \textnormal{diag}\left\{ \vZ_{A}(l) \right\}\otimes \mIs\right]\right)\nonumber\\
%%%%%%%%%%%%%%%%%%%%%%
%&=\,  \E_{\SP} \left( \left[ \textnormal{diag}\left\{ \vZ_{A}(k) \right\}\otimes \mIs\right] \cov_{\SP} \left\{\vZ_{B}(k), \vZ_{B}(l) \right\}\left[ \textnormal{diag}\left\{ \vZ_{A}(l) \right\}\otimes \mIs\right]\right)\nonumber\\
%%%%%%%%%%%%%%%%%%%%%%%%
\,\overset{\phantom{\eqref{eq::Cov_ZBk}}}{=}&
\E_{\SP} \left( \left[ \textnormal{diag}\left\{ \vZ_{A}(k) \right\} \otimes \mIs \right]
\left( \mIw \otimes\mP_M\right)
\left[ \textnormal{diag}\left\{ \vZ_{A}(l) \right\} \otimes \mIs \right] \right)\nonumber\\
%%%%%%%%%%%%%%%%%%%%%%%%%%%%%%%%%
\,\overset{\phantom{\eqref{eq::Cov_ZBk}}}{=}&
\E_{\SP}\left(
\left[ \textnormal{diag}\left\{ \vZ_{A}(k) \right\}\mIw \textnormal{diag}\left\{ \vZ_{A}(l) \right\} \right] \otimes\left(\mIs \mP_M \mIs\right) \right)\nonumber\\
%%%%%%%%%%%%%%%%%%%%%%%%%%%%
\,\overset{\phantom{\eqref{eq::Cov_ZBk}}}{=}&
\E_{\SP} \left[\textnormal{diag}\left\{ \vZ_{A}(k) \circ \vZ_{A}(l) \right\}\right] \otimes \mP_M\nonumber\\
%%%%%%%%%%%%%%%%%%%%%%%%%%%%
\,\overset{\phantom{\eqref{eq::Cov_ZBk}}}{=}&
\{\E_{\SP}(I_{\{Z_1 = g_{A}(k)\}}I_{\{Z_1 = g_{A}(l)\}}
)\cdot \mIw\}\otimes \mP_M \nonumber\\
\,\overset{\phantom{\eqref{eq::Cov_ZBk}}}{=}&
\E_{\SP}(I_{\{Z_1 = g_{A}(k)= g_{A}(l)\}}
)
\left( \mIw\otimes \mP_M \right) \,=\,
\left(I_{\{g_{A}(k)= g_{A}(l)\}}\frac{W_{(k)}}{W} \right)
\left( \mIw\otimes \mP_M \right)\nonumber\\
%%%%%%%%%%%%%%%%%%%%%%%%%%%%
\,\overset{\phantom{\eqref{eq::Cov_ZBk}}}{=}&
I_{\{g_{A}(k)= g_{A}(l)\}} \frac{W_{(k)}}{W(M-1)}
\mPi\,.
\end{align*}
Multiply both sides by $g_{B}(k)g_{B}(l){M_{+1}M_{-1}}/\{M(M-1)\}$ to have
\begin{equation}
\E_{\SP} \left[\cov_{\SP} \left\{ \vZ(k), \vZ(l) \mid \mathcal{Z}_A\right\}\right]
=
I_{\{g_{A}(k)= g_{A}(l)\}}g_{B}(k)g_{B}(l) \frac{W_{(k)}M_{+1}M_{-1}}{N(M-1)}
\mPi\,. \label{eq::secondTerm}
\end{equation}
This gives the expectation of covariance component of \eqref{eq::ECCE}.
%Alternatively, identity $I_{\{g_{A}(k)= g_{A}(l)\}} = 2^{-1}(g_{A}(k) g_{A}(l)+1)$
%yields an equivalent expression
%\begin{eqnarray*}
%I_{\{g_{A}(k)= g_{A}(l)\}}g_{B}(k)g_{B}(l)\frac{W_{(k)}M_{+1}M_{-1}}{N(M-1)}
%&=&
%2^{-1}(g_{A}(k) g_{A}(l)+1)g_{B}(k)g_{B}(l)\frac{W_{(k)}M_{+1}M_{-1}}{N(M-1)}
%\mIw\otimes \mP_M\\
%%%%%%%%%%%%%%%%%%%%
%&=&
%2^{-1}(g_{AB}(k) g_{AB}(l)+g_{B}(k)g_{B}(l)) \frac{W_{(k)}M_{+1}M_{-1}}{N(M-1)}\mPi.
%\end{eqnarray*}
%%%%%%%%%%%%%%%%%%%
%%%%%%%%%%%%%%%%%%%%%%%%%%%%%%%%%%%%%%%%%%%%%%%%%%%%%%%%%%%%%%%%%%%%%%%%%%%%%%%%%%%%%%%%%%%%%%%%%%%%%%%%%%%%%%%%%%%%%%%%
%%%%%%%%%%%%%%%%%%%%%%%%%%%%%%%%%%%%%%%%%%%%%%%%%%%%%%%%%%%%%%%%%
%
% EC+CE
%
%%%%%%%%%%%%%%%%%%%%%%%%%%%%%%%%%%%%%%%%%%%%%%%%%%%%%%%%%%%%%%%%%
Substituting \eqref{eq::firstTerm} and \eqref{eq::secondTerm} into \eqref{eq::ECCE} yields
\begin{align}
\label{eq::covZkZl_SP}
\cov_{\SP}\{ \vZ(k), \vZ(l)\}
\,=\,\,&
g_{A}(k)g_{A}(l) \frac{W_{+1}W_{-1}M_{(k)}M_{(l)}}{N(W-1)}\mPb \\
&+\, I_{\{g_{A}(k)= g_{A}(l)\}}g_{B}(k)g_{B}(l)\frac{W_{(k)}M_{+1}M_{-1}}{N(M-1)}\mPi\,, \nonumber
%\\
%%%%%%
%&=& g_{A}(k)g_{A}(l) \frac{W_{+1}W_{-1}M_{(k)}M_{(l)}}{N(W-1)}\mPb
%+
%2^{-1}(g_{AB}(k) g_{AB}(l)+g_{B}(k)g_{B}(l)) \frac{W_{(k)}M_{+1}M_{-1}}{N(M-1)}\mPi
\end{align}
and
\begin{align}
\label{eq::ECCE_res}
\cov_{\SP} \left\{ N_k^{-1}\vZ(k), N_l^{-1}\vZ(l) \right\}
&=\,
\left(W_{(k)}M_{(k)}W_{(l)}M_{(l)}\right)^{-1} \cov_{\SP} \left\{ \vZ(k), \vZ(l) \right\}\\
&=\,
 C^{{\textnormal{btw}}}_{k,l}\mPb + C^{{\textnormal{in}}}_{k,l}\mPi\,,\nonumber
%\\
%&=&
%\left\{
%\begin{array}{ll}
% \frac{W_{-g_{A}(k)}}{W_{(k)}}\cdot \frac{1}{N(W-1)} \cdot\mPb +
%W_{(k)}^{-1}\cdot \frac{S_{-g_{B}(k)}}{S_{g_{B}(k)}}\cdot\frac{1}{N(M-1)}\cdot \mPi\,,& \text{if $k=l$,}\\
%%%%%%%%%%%%%%%%
% \frac{W_{-z}}{W_z}\cdot \frac{1}{N(W-1)} \cdot\mPb-
%W_{z}^{-1}\cdot\frac{1}{N(M-1)}\cdot \mPi\,, & \text{if $k\neq l$ yet $g_{A}(k) = g_{A}(l) = z$,}\\
%%%%%
%-\frac{1}{N(W-1)} \cdot\mPb\,, &\text{if $g_{A}(k) \neq g_{A}(l)$.}
%\end{array}\right. \\\\
%&=:& \mC^{\SP}_{k,l}
\end{align}
where
\begin{align*}
C^{{\textnormal{btw}}}_{k,l} \,=\,  \frac{ g_{A}(k)g_{A}(l)}{W_{(k)}W_{(l)}}\frac{W_{+1}W_{-1}}{N(W-1)}\,,\,\,\,
C^{{\textnormal{in}}}_{k,l}\, = \,\frac{I_{\{g_{A}(k)= g_{A}(l)\}}g_{B}(k)g_{B}(l)}{W_{(l)}M_{(k)}M_{(l)}} \frac{M_{+1}M_{-1}}{N(M-1)}\,.
\end{align*}
It is straightforward to verify that $C^{{\textnormal{btw}}}_{k,l}$ is the $(k,l)$th entry of $\mCb$ and $C^{{\textnormal{in}}}_{k,l}$ is the $(k,l)$th entry of $\mCi$.
This, coupled with \eqref{eq::ECCE_res}, completes the proof.
\end{proof}
%%%%%%%%%%%%%%%%%%%%%%%%%%
%%%%%%%%%%%%%%%%%%%%%%%%%%
An interesting observation is that the three coefficient matrices $\mCb$, $\mCi$, and $\mCf$ satisfy
\begin{equation}
\label{eq::interesting}
(W-1)\mCb + W(M-1)  \mCi \,=\, (N-1) \mCf\,.
\end{equation}

\begin{proof}[Proof of Identity \eqref{eq::interesting}]
The result follows from
\begin{align}
&N(W-1)\mCb +
NW(M-1)  \mCi \nonumber\\
&\quad=\,
\begin{pmatrix} r_A &-1\\ -1& r^{-1}_A \end{pmatrix} \otimes \mJ_2
+\left\{
\begin{pmatrix} r_A &-1\\	-1& r^{-1}_A \end{pmatrix}
\otimes
\begin{pmatrix} r_B&-1\\-1& r_B^{-1}\end{pmatrix}
+
\mJ_2
\otimes
\begin{pmatrix} r_B&-1\\-1& r_B^{-1}\end{pmatrix}
\right\}\nonumber\\
&\quad=\,
\left\{\begin{pmatrix} r_A &-1\\ -1& r^{-1}_A \end{pmatrix} + \mJ_2 \right\}
\otimes
\left\{
\mJ_2 +
\begin{pmatrix} r_B&-1\\-1& r_B^{-1}\end{pmatrix}
\right\}- \mJ_4
%\,=\,
%\begin{pmatrix} r_A + 1 &0\\ 0& r^{-1}_A+1 \end{pmatrix}
%\otimes
%\begin{pmatrix} r_B+1&0\\0& r_B^{-1}+1\end{pmatrix}
%- \mJ_4
\nonumber\\
&\quad=\,
\left\{\begin{pmatrix} W/ W_{-1} &0\nonumber\\ 0& W/ W_{+1}\end{pmatrix}\right\}
\otimes
\left\{\begin{pmatrix} M/ M_{-1} &0\nonumber\\0& M/ M_{+1}\end{pmatrix} \right\}
- \mJ_4\nonumber\\
&\quad=\,
 \textnormal{diag}\left\{\frac{N}{N_1}, \frac{N}{N_2},\frac{N}{N_3},\frac{N}{N_4}\right\} - \mJ_4\,=\, N(N-1) \mCf\,.\nonumber
\end{align}
\end{proof}

%%%%%%%%%%%%%%%%%%%%%%%%%%%%%%%%%%
% ===============================================
%
\section{Sampling variances of the estimators}
%
% ==========================================
%%%%%%%%%%%%%%%%%%%%%%%%%%%%%%
%
%\begin{lemma}
%\label{lem::zong}
%Given any 4-treatment-level assignment mechanism \textsc{am},
%the randomization-based means and variances of the resulting $\widehat{\tau}_F$ defined in \eqref{eq::pointEstimator} can be computed as
%\begin{eqnarray*}
%\E_{\textsc{AM}}(\widehat{\tau}_F) \,=\,  2^{-1}\vg_F^{\T} \Yt^{\T}\E_{\textsc{AM}}\left( \bm Z^* \right)\,,
%\quad
%\var_{\textsc{AM}}(\widehat{\tau}_F) \,=\, 4^{-1} \vg_F^{\T} \Yt^{\T}\cov_{\textsc{AM}}\left( \bm Z^* \right)\Yt \vg_F
%\quad(F = A, B, AB)\,,
%\end{eqnarray*}
%where all expectations and (co)variances are taken with regards to the probability distribution induced by \textsc{am} over all possible assignments.
%\end{lemma}

\begin{proof}[Proof of Lemma \ref{lem::zong}]
Straightforward.
\end{proof}

% ============================
% COROLLARY: CONNECTION TO TIR2015
% ==============================

%%%%%%%%%%%%%%%%%%%%%%%%%%%%%%%%%%%%%%%%%%%%%%%%%%%%%%%%%%%%%%%%%%%%%%%%%%%%%%%%%%%%%%%%%%%%%%%%%%%%%%%%%%%%%%%%%%%%%%%%%%%%%%%%%
%
%	PROOF OF Theoretical-VAR factorial 	
%
%%%%%%%%%%%%%%%%%%%%%%%%%%%%%%%%%%%%%%%%
%%%%%%%%%%%%%%%%%%%%%%%
%\begin{theorem}
%\label{thm::factorial}
%For a $2^2$ completely randomized design,
%$\widehat{\tau}_F$ are unbiased estimators of factorial effects ${{\tau}}_F$ with variances
%\begin{eqnarray}
%\label{eq::var_CR}
%\var_{\textsc{c-r}}(\widehat{\tau}_F)
%&=& 4^{-1}(N-1) \vg_F^{\T} \left(\mCf \circ \mS \right) \vg_F\,.
%\end{eqnarray}
%\end{theorem}

\begin{proof}[Proof of Theorem \ref{thm::factorial}]
By Lemmas \ref{lem::covZ_CR} and \ref{lem::zong} we have
\begin{align*}
\E_{\textsc{c-r}}(\widehat{\tau}_F)
&=\,  2^{-1}\vg_F^{\T} \Yt^{\T}\E_{\textsc{c-r}}( \bm Z^* )
\,=\,  2^{-1}\vg_F^{\T} \Yt^{\T} (N^{-1}\vones_{4N})\\
&=\, 2^{-1}\vg_F^{\T}( N^{-1}\Yt^{\T}\vones_{4N})
\,=\,
2^{-1}\vg_F^{\T} \bar{\bm Y} \,=\, \tau_F
\end{align*}
and
\begin{align*}
\var_{\textsc{c-r}}(\widehat{\tau}_F)
&=\, 4^{-1} \vg_F^\T \widetilde{\mY}^\T \cov_{\textsc{c-r}}( \bm Z^*) \widetilde{\mY}\vg_F
\,=\, 4^{-1} \vg_F^\T \widetilde{\mY}^\T\left(\mCf \otimes \mP_N\right)\widetilde{\mY}\vg_F\\
%%%%%%%%%%%%%%%%%%%%%%
&=\, 4^{-1} \vg_F^\T \{ \widetilde{\mY}^\T \left(\mCf \otimes \mP_N\right) \Yt\} \vg_F
%%%%%%%%%%%%%%%%%%%%%
\overset{\eqref{eq::kronecker_zong}}{=} 4^{-1} \vg_F^\T \left\{ \mCf \circ \left( \mY^\T \mP_N\mY \right) \right\}\vg_F\,.
\end{align*}
This, coupled with $\mS = (N-1)^{-1} \mY^\T \mP_N\mY$, proves \eqref{var_CR}.
%%%%%%%%%%%%%%%%%%%%%%%%%
%%%%%%%%%%%%%%%%%%%%%%%%%

When the design is balanced,
the coefficient matrix $\mCf$ in Lemma \ref{lem::covZ_CR} reduces to
$\mCf = (4\mI_4-\mJ_4)/\{N(N-1)\} = 4\mP_4/\{N(N-1)\}$.
Substituting this simplified version into \eqref{var_CR} yields
%\begin{align*}
%&N \cdot \var_{\textsc{c-r}}(\widehat{\tau}_F)
%\,=\,  \vg_F^{\T} \left( \mP_4 \circ \mathbf{S^\text{$2$}} \right) \vg_F
%\,=\,  \vg_F^{\T} \left\{ (\mI_4 - \mJ_4/4) \circ \mathbf{S^\text{$2$}} \right\} \vg_F \\
%&\quad =\,  \vg_F^{\T}(\mI_4\circ \mathbf{S^\text{$2$}} ) \vg_F - 4^{-1} \vg_F^{\T} ( \mJ_4 \circ \mathbf{S^\text{$2$}}) \vg_F\\
%&\quad =\, \vg_F^{\T} \, \textnormal{diag}\{S^2(1,1),S^2(2,2),S^2(3,3),S^2(4,4)\} \vg_F - 4^{-1} \vg_F^{\T} \mathbf{S^\text{$2$}} \vg_F\\
%&\quad=\,  \sum_{k=1}^4 S^2(k,k)  - 4^{-1} \vg_F^{\T} \mY^\T \mP_N \mY\vg_F
%\,=\,
% \sum_{k=1}^4 S^2(k,k) - S^2_F\,.
%\end{align*}
\begin{align*}
N \cdot \var_{\textsc{c-r}}(\widehat{\tau}_F)
&=\,  \vg_F^{\T} \left( \mP_4 \circ \mathbf{S^\text{$2$}} \right) \vg_F
\,=\,  \vg_F^{\T} \left\{ (\mI_4 - \mJ_4/4) \circ \mathbf{S^\text{$2$}} \right\} \vg_F \\
&=\,  \vg_F^{\T}(\mI_4\circ \mathbf{S^\text{$2$}} ) \vg_F - 4^{-1} \vg_F^{\T} ( \mJ_4 \circ \mathbf{S^\text{$2$}}) \vg_F\\
&=\, \vg_F^{\T} \, \textnormal{diag}\{S^2(1,1),S^2(2,2),S^2(3,3),S^2(4,4)\} \vg_F - 4^{-1} \vg_F^{\T} \mathbf{S^\text{$2$}} \vg_F\\
&=\,  \sum_{k=1}^4 S^2(k,k)  - 4^{-1} \vg_F^{\T} \mY^\T \mP_N \mY\vg_F
\,=\,
 \sum_{k=1}^4 S^2(k,k) - S^2_F\,.
\end{align*}
This completes the proof.
\end{proof}

\begin{proof}[Proof of Theorem \ref{thm::theoretical-variances}]
By Theorem \ref{lem::covZ_SP} and Lemma \ref{lem::zong}, we have
\begin{align*}
\E_{\SP}(\widehat{\tau}_F)
&\,\,\,=\,\,\,\, 2^{-1}\vg_F^{\T} \Yt^{\T}\E_{\SP}\left( \bm Z^* \right)
\,=\,  2^{-1}\vg_F^{\T} \Yt^{\T} (N^{-1}\vones_{4N}) \\
&\,\,\,=\,\,\,\,  2^{-1}\vg_F^{\T} ( N^{-1}\Yt^{\T}\vones_{4N} )
\,=\,
2^{-1}\vg_F^{\T} \bar{\bm Y} \,=\, \tau_F\,,\\
%%%%%%%%%%
\vspace{3mm}
\var_\SP(\widehat{\tau}_F)
&\,\,\,=\,\,\,\, 4^{-1} \vg_F^\T \widetilde{\mY}^\T \cov_{\SP} \left( \bm Z^* \right) \widetilde{\mY}\vg_F\\
&\,\,\,=\,\,\,\,4^{-1} \vg_F^\T \widetilde{\mY}^\T\left(\mCb\otimes \mPb + \mCi\otimes \mPi\right)\widetilde{\mY}\vg_F\\
%%%%%%%%%%%%%%%%%%%%%%
&\,\,\,=\,\,\,\, 4^{-1} \vg_F^\T \{ \widetilde{\mY}^\T \left(\mCb\otimes \mPb \right) \Yt + \Yt \left(\mCi\otimes \mPi\right) \Yt \} \vg_F\\
%%%%%%%%%%%%%%%%%%%%%
&\overset{\eqref{eq::kronecker_zong}}{=} 4^{-1} \vg_F^\T \left\{ \mCb \circ \left( \mY^\T \mPb\mY \right) + \mCi \circ \left( \mY^\T \mPi  \mY \right)\right\}\vg_F\,.
\end{align*}
This, coupled with the definitions of $\mSb$ and $\mSi$ in \eqref{mSbmSi},
 completes the proof.
\end{proof}

\begin{proof}[Proof of Corollary \ref{corollary::balanced}]
When the design is balanced, we have $r_A = r_B = 1$.
and the coefficient matrices $\mCb$ and $\mCi$ in \eqref{eq::theoreticalVar} reduce to
\begin{equation*}
\mCb
\,=\, \frac{1}{N(W-1)} \vg_A \vg_A^{\T}\,,\quad
%%%%%%%%%%%%%%%%%%%%%%
\mCi
\,=\, \frac{1}{NW(M-1)}\left( \vg_B \vg_B^\T + \vg_{AB}  \vg_{AB}^\T\right).
\end{equation*}
Substituting these simplified versions, together with the definitions of $\mSb$ and $\mSi$ in \eqref{mSbmSi}, into  \eqref{eq::theoreticalVar} yields
\begin{align}
\label{eq::balanced_decomposition}
\var_{\textsc{S-P}}( \widehat{\tau}_F)
\,=&\frac{\vg_F^{\T}\left\{ (\vg_A \vg_A^{\T}) \circ (\mY^\T\mPb\mY) \right\} \vg_F}{4N (W-1)}+\frac{ \vg_F^{\T}\left\{ ( \vg_B \vg_B^\T + \vg_{AB}  \vg_{AB}^\T)  \circ (\mY^\T\mPi\mY) \right\} \vg_F}{4NW(M-1)}\\
%%%%%%%%%%
\,=&\,\frac{\vg_F^{\T}\left\{ (\vg_A \vg_A^{\T}) \circ (\mY^\T\mPb\mY) \right\} \vg_F}{4N (W-1)}\nonumber\\
&+\frac{\vg_F^{\T}\left\{ ( \vg_B \vg_B^\T)  \circ (\mY^\T\mPi\mY) \right\} \vg_F}{4NW(M-1)}
+\frac{\vg_F^{\T}\left\{ (\vg_{AB}  \vg_{AB}^\T)  \circ (\mY^\T\mPi\mY)\right\} \vg_F}{4NW(M-1)} \,.\nonumber
\end{align}
%%%%%%%%%%%%%%%%%%%%%%%
Introduce shorthand notations for the entrywise products in \eqref{eq::balanced_decomposition}:
\begin{equation*}
\mHb =
(\vg_A \vg_A^{\T}) \circ (\mY^\T\mPb\mY)\,,\,\,
\mathbf{H}_{\text{in-}F}= ( \vg_F \vg_F^\T) \circ (\mY^\T\mPi\mY)
\,\,(F = B, AB)\,.
\end{equation*}
%to paraphrase \eqref{eq::balanced_decomposition0} as
%\begin{align}
%\label{eq::balanced_decomposition}
%\var_{\textsc{S-P}}( \widehat{\tau}_F)\,=&\,\{4N (W-1)\}^{-1}\vg_F^{\T}\mHb \vg_F\\
%&+\,\{ 4NW(M-1)\}^{-1} \left( \vg_F^{\T} \mathbf{H}_{\text{in-}B}\vg_F + \vg_F^{\T} \mathbf{H}_{\text{in-}AB} \vg_F\right),\nonumber
%\end{align}
%where $\vg_F^\T \mHb \vg_F $, as
It follows from
\begin{align*}
\mHb &=\, (\vg_A \vg_A^{\T}) \circ (\mY^\T\mPb\mY)
\overset{\eqref{eq::circToDot}}{=}
\textnormal{diag}\{\vg_A\} (\mY^\T\mPb\mY)\textnormal{diag}\{\vg_A\} \\
&=\,  [\mY \textnormal{diag}\{\vg_A\}]^\T \mPb[\mY\textnormal{diag}\{\vg_A\}]
\end{align*}
%%%%%%%%%%%%%%%%
%%%%%%%%%%%%%%%%
that
\begin{align}
\label{eq::quadratic1_general}
\vg_F^\T \mHb \vg_F
&=\,
[  \mY \textnormal{diag}\{\vg_A\}\vg_F]^\T \mPb [\mY\textnormal{diag}\{\vg_A\} \vg_F ]\\
&=\,
\{ \mY( \vg_A\circ\vg_F)\}^\T  \mPb \{ \mY (\vg_A\circ\vg_F)\}\,. \nonumber
\end{align}
This, coupled with
$
\mY (\vg_A\circ\vg_A)= \mY \vones_4= 4\bm \mu$,
$\mY (\vg_A\circ\vg_B) = \mY\vg_{AB} = 2\bm\tau_{AB}$, and $\mY (\vg_A\circ\vg_{AB}) = \mY\vg_{B} =2\bm\tau_B$,
yields
\begin{align}
\label{eq::quadraticForms_mHb}
\vg_A^\T &\mHb \vg_A
\,=\, 16 \bm\mu^\T  \mPb \bm\mu
\,=\,  \left\{16 (W-1)M\right\} S^2_{\mu\text{-btw}}\,,\\
%%%%%%%%%%%%%%%%%%
\vg_B^\T&\mHb \vg_B
\,=\, 4\bm{\tau}_{AB}^\T\mPb\bm{\tau}_{AB}
\,=\, \left\{ 4(W-1)M\right\} S^2_{AB\text{-btw}}\,,\nonumber\\
%%%%%%%%%%%%%%%%%%%%%%%%%%
\vg_{AB}^\T &\mHb \vg_{AB}
\,=\, 4\bm{\tau}_{B}^\T \mPb\bm{\tau}_{B}
\,=\,  \left\{ 4(W-1)M\right\} S^2_{B\text{-btw}}
\,.\nonumber
\end{align}
%%%%%%%%%%%%%%%%%%%%%%%
%%%%%%%%%%%%%%%%%%%%%%%
%%%%%%%%%%%%%%%%%%%%%%%
Analogues of \eqref{eq::quadratic1_general} for $\mathbf{H}_{\text{in-}B}$ and $\mathbf{H}_{\text{in-}AB}$ follow from similar algebra as
\begin{align*}
\vg_F^\T \mathbf{H}_{\text{in-}B} \vg_F
&=\,
%\vg_F^\T( [\mY \textnormal{diag}\{\vg_B\}]^\T \mPi [\mY\textnormal{diag}\{\vg_B\} ] ) \vg_F
\{ \mY( \vg_B\circ\vg_F)\}^\T  \mPi \{ \mY (\vg_B\circ\vg_F)\},\\
%%%%%%%%%%%%%
\vg_F^\T \mathbf{H}_{\text{in-}AB} \vg_F
&=\,
%\vg_F^\T( [\mY \textnormal{diag}\{\vg_{AB}\}]^\T \mPi [\mY\textnormal{diag}\{\vg_{AB}\} ] ) \vg_F\,=\,
\{ \mY( \vg_{AB}\circ\vg_F)\}^\T  \mPi \{ \mY (\vg_{AB}\circ\vg_F)\}.
\end{align*}
This, coupled with
\begin{equation*}
\begin{array}{lll}
\mY (\vg_B\circ\vg_A) \,=\, 2\bm\tau_{AB}\,,& \mY (\vg_B\circ\vg_B) \,=\,4\bm\mu\,,&
\mY (\vg_B\circ\vg_{AB})\,=\, 2\bm\tau_{A}\,,\\
\mY (\vg_{AB} \circ\vg_A) \,=\, 2\bm\tau_B\,,&\mY (\vg_{AB} \circ\vg_B) \,=\, 2\bm\tau_{A}\,,&
\mY (\vg_{AB} \circ\vg_{AB})\,=\,4\bm\mu
\end{array}
\end{equation*}
yields
\begin{align}
\label{eq::quadraticForms_mHi}
\vg_A^{\T}  &\mathbf{H}_{\text{in-}B} \vg_A  \,=\, 4W(M-1)S^2_{AB\text{-in}}\,,&
\vg_A^{\T} &\mathbf{H}_{\text{in-}AB} \vg_A \,=\,
4W(M-1)S^2_{B\text{-in}}\,, \\
%%%%%%%%%%%%
\vg_B^{\T}&\mathbf{H}_{\text{in-}B} \vg_B
\,=\, 16 W(M-1) S^2_{\mu\text{-in}}\,,&
\vg_B^{\T}&\mathbf{H}_{\text{in-}AB} \vg_B
\,=\, 4W(M-1)S^2_{A\text{-in}}\,,\nonumber\\
%%%%%%%%%%%%%%
\vg_{AB}^{\T}&\mathbf{H}_{\text{in-}B} \vg_{AB}
\,=\,  4W(M-1)S^2_{A\text{-in}}\,,&
\vg_{AB}^{\T}&\mathbf{H}_{\text{in-}AB} \vg_{AB}
\,=\,16 W(M-1) S^2_{\mu\text{-in}}\,.\nonumber
\end{align}
%%%%%%%%%%%%%%%%%%
Substituting \eqref{eq::quadraticForms_mHb} and \eqref{eq::quadraticForms_mHi} into
\eqref{eq::balanced_decomposition} completes the proof.
\end{proof}

%%%%%%%%%%%%%%%%%%%%%%%%%%%%%%%%%%

\begin{proof}[Proof of Corollary \ref{corollary::strictAdditive}]
Under strict additivity, we have $\mSb = \Sb\mJ_4$ and $\mSi = \Si \mJ_4$.
Formula \eqref{eq::theoreticalVar} simplifies to
\begin{equation}
\label{eq::var_strictAdd}
\var_{\textsc{S-P}}( \widehat{\tau}_F )
= 4^{-1}\{ (W-1)M\,\Sb\}\, \vg_F^{\T} \mCb\vg_F + 4^{-1}\{W(M-1) \Si \}\, \vg_F^\T \mCi \vg_F\,.
\end{equation}
To prove  Corollary \ref{corollary::strictAdditive} thus reduces to computing the quadratic forms  $\vg_F^{\T} \mCb\vg_F$ and $\vg_F^{\T} \mCi\vg_F$. Starting with $F = A$, direct application of \eqref{eq::Kronecker_prod} to
\begin{align*}
\mCb &=\, \frac{1}{N(W-1)} \begin{pmatrix} r_A &-1\\ -1& r^{-1}_A \end{pmatrix} \otimes \mJ_2\,,\\
\mCi &=\,
\frac{1}{NW(M-1)}\left\{
\begin{pmatrix} r_A &-1\\	-1& r^{-1}_A \end{pmatrix} + \mJ_2\right\}
\otimes
\begin{pmatrix} r_B&-1\\-1& r_B^{-1}\end{pmatrix}\\
&=\,
\frac{1}{NW(M-1)}
\begin{pmatrix} r_A &-1\\	-1& r^{-1}_A \end{pmatrix}\otimes
\begin{pmatrix} r_B&-1\\-1& r_B^{-1}\end{pmatrix}
+
\frac{1}{NW(M-1)}\mJ_2
\otimes
\begin{pmatrix} r_B&-1\\-1& r_B^{-1}\end{pmatrix},\\
%%%%%%%%%
\vg_A &=\, (-1,1)^\T \otimes \vones_2
%,
%\, \,\,
%\vg_B \,=\,  \vones_2 \otimes (-1,1)^\T,\,\,\, \vg_{AB}\,=\, (-1, 1)^\T \otimes (-1,1)^\T.
\end{align*}
yields
\begin{align*}
\vg_A^\T  \{N(W-1)\mCb\} \vg_A
&\overset{\phantom{\eqref{eq::Kronecker_prod}}}{=}
\left\{(-1, 1) \otimes \vones_2^\T\right\}
\left\{ \begin{pmatrix} r_A &-1\\ -1& r^{-1}_A \end{pmatrix} \otimes \mJ_2 \right\}
\left\{ (-1, 1)^\T \otimes \vones_2 \right\} \nonumber\\
&\overset{\eqref{eq::Kronecker_prod}}{=}
\left\{(-1, 1)\begin{pmatrix} r_A &-1\\ -1& r^{-1}_A \end{pmatrix} \begin{pmatrix}-1\nonumber\\1\end{pmatrix}
\right\}
\otimes
\left(\vones_2^\T\mJ_2\vones_2\right)\nonumber\\
&\overset{\phantom{\eqref{eq::Kronecker_prod}}}{=} \left\{\gamma_A\otimes4\right\}
\,\,\,=\,\,\,
4 \gamma_A\,,
\end{align*}
\vspace{-10mm}
\begin{align*}
\vg_A^\T  \{NW(M&-1)\mCi\}  \vg_A\nonumber\\
\overset{\phantom{\eqref{eq::Kronecker_prod}}}{=}&
 \vg_A^\T
\left\{
\begin{pmatrix} r_A &-1\\	-1& r^{-1}_A \end{pmatrix}
\otimes
\begin{pmatrix} r_B&-1\\-1& r_B^{-1}\end{pmatrix}
\right\}
\vg_A
+
\vg_A^\T
\left\{
\mJ_2
\otimes
\begin{pmatrix} r_B&-1\\-1& r_B^{-1}\end{pmatrix}
\right\}
\vg_A\nonumber\\
%%%%%%%%%%%%%%%%%%%%%
\overset{\phantom{\eqref{eq::Kronecker_prod}}}{=}&
\left\{(-1, 1) \otimes \vones_2^\T\right\}
\left\{
\begin{pmatrix} r_A &-1\\	-1& r^{-1}_A \end{pmatrix}
\otimes
\begin{pmatrix} r_B&-1\\-1& r_B^{-1}\end{pmatrix}
\right\}
\{ (-1, 1)^\T \otimes \vones_2\} \nonumber\\
&+\,
\{(-1, 1) \otimes \vones_2^\T\}
\left\{
\mJ_2
\otimes
\begin{pmatrix} r_B&-1\\-1& r_B^{-1}\end{pmatrix}
\right\}
\left\{ (-1, 1)^\T \otimes \vones_2 \right\}\nonumber\\
%%%%%%%%%%%%%%%%%%%%%
\overset{\eqref{eq::Kronecker_prod}}{=}&
\left\{
(-1, 1)
\begin{pmatrix} r_A &-1\\ -1& r^{-1}_A \end{pmatrix}
\begin{pmatrix}-1\\1\end{pmatrix}
\right\}
\otimes
\left\{
\vones_2^\T
\begin{pmatrix} r_B&-1\\-1& r_B^{-1}\end{pmatrix}
\vones_2
\right\}\nonumber\\
&+\,
\left\{
(-1, 1)
\mJ_2
\begin{pmatrix}-1\\1\end{pmatrix}
\right\}
\otimes
\left\{\vones_2^\T
\begin{pmatrix} r_B&-1\\-1& r_B^{-1}\end{pmatrix}
\vones_2\right\}\nonumber\\
%%%%%%%%%%%%%%%%%%%
\overset{\phantom{\eqref{eq::Kronecker_prod}}}{=}&
\gamma_A\otimes (\gamma_B-4 ) + 0
\,\,\,=\,\,\,\gamma_A(\gamma_B-4)\,.\nonumber
\end{align*}
Thus
\begin{equation}
\label{eq::quadratic_Abtw}
\vg_A^\T  \mCb  \vg_A
=
\frac{4 \gamma_A}{N(W-1)}\,,\quad
\vg_A^\T  \mCi  \vg_A
=
\frac{\gamma_A(\gamma_B-4)}{NW(M-1)}\,.
\end{equation}
Substituting $\vg_A$ with $\vg_B=  \vones_2 \otimes (-1,1)^\T$ and $\vg_{AB}= (-1, 1)^\T \otimes (-1,1)^\T$ respectively in the above computation yields
\begin{align}
\vg_B^\T  \mCb  \vg_B
&\overset{\eqref{eq::Kronecker_prod}}{\propto}
\left\{\vones_2^\T\begin{pmatrix} r_A &-1\\ -1& r^{-1}_A \end{pmatrix}
\vones_2\right\}
\otimes
\left\{\left(-1, 1\right)\mJ_2\begin{pmatrix}-1\\1\end{pmatrix}\right\}
\,=\,0\,,\\
%%%%%%%%%%%%%
\vg_{AB}^\T  \mCb  \vg_{AB}
&\overset{\eqref{eq::Kronecker_prod}}{\propto}
\left\{(-1, 1)\begin{pmatrix} r_A &-1\\ -1& r^{-1}_A \end{pmatrix} \begin{pmatrix}-1\\1\end{pmatrix}\right\}
\otimes
\left\{(-1, 1)\mJ_2
\begin{pmatrix}-1\\1\end{pmatrix}
\right\}\,=\, 0
\end{align}
and
\begin{align*}
&\vg_B^\T  \{NW(M-1)\mCi \} \vg_B\nonumber\\
&\quad=\,
\{\vones_2^\T \otimes (-1,1)\}
\left\{
\begin{pmatrix} r_A &-1\\	-1& r^{-1}_A \end{pmatrix}
\otimes
\begin{pmatrix} r_B&-1\\-1& r_B^{-1}\end{pmatrix}
\right\}
\left\{ \vones_2 \otimes (-1,1)^\T\right\}\nonumber \\
&\quad\quad+\,
\{\vones_2^\T \otimes (-1,1)\}
\left\{\mJ_2\otimes\begin{pmatrix} r_B&-1\\-1& r_B^{-1}\end{pmatrix}\right\}
\left\{ \vones_2 \otimes (-1,1)^\T\right\} \nonumber\\
&\quad=\,
\left(\gamma_A-4\right)
\otimes
\gamma_B
 +
4\otimes\gamma_B
\,\,=\,\,
\gamma_A\gamma_B\,,\nonumber
\\\vspace{-8mm}\\
%\end{align}
%%%%%%%%%%%%%%%%%%%%%%%%%%%
%\vspace{-8mm}
%%%%%%%%%%%%%%%%%%%%%%%%%%%%
%\begin{align}
&\vg_{AB}^\T  \{NW(M-1)\mCi\}  \vg_{AB}\nonumber\\
&\quad=\,
\{(-1, 1) \otimes (-1,1)\}
\left\{
\begin{pmatrix} r_A &-1\\	-1& r^{-1}_A \end{pmatrix}
\otimes
\begin{pmatrix} r_B&-1\\-1& r_B^{-1}\end{pmatrix}
\right\}
\left\{(-1, 1)^\T \otimes (-1,1)^\T\right\}\nonumber \\
&\quad\quad+\,
\{(-1, 1) \otimes (-1,1)\}
\left\{\mJ_2\otimes\begin{pmatrix} r_B&-1\\-1& r_B^{-1}\end{pmatrix}
\right\}
\left\{(-1, 1)^\T \otimes (-1,1)^\T\right\} \nonumber\\
&\quad=\,
%\frac{1}{NW(M-1)}
%\left\{
%(-1, 1)
%\begin{pmatrix} r_A &-1\\ -1& r^{-1}_A \end{pmatrix}
%\begin{pmatrix}-1\\1\end{pmatrix}\right\}\otimes
%\left\{(-1, 1)
%\begin{pmatrix} r_B&-1\\-1& r_B^{-1}\end{pmatrix}
%\begin{pmatrix}-1\\1\end{pmatrix}
%\right\}\\
%&+& \frac{1}{NW(M-1)}
%\left\{(-1, 1)\mJ_2\begin{pmatrix}-1\\1\end{pmatrix}\right\}
%\otimes
%\left\{(-1, 1)\begin{pmatrix} r_B&-1\\-1& r_B^{-1}\end{pmatrix}
%\begin{pmatrix}-1\\1\end{pmatrix}\right\}\\
%%%%%%%%%%%%%%%%%%
%&=&
\gamma_A
\otimes
\gamma_B
 +
0
\,\,=\,\,\gamma_A\gamma_B\,.\nonumber
\end{align*}
Thus
\begin{equation}
\label{eq::quadratic_ABbtw}
\vg_{AB}^\T \mCi  \vg_{AB}\,=\, \vg_{AB}^\T  \mCi  \vg_{AB}\,=\,\frac{
\gamma_A\gamma_B}{NW(M-1)}
\,.
\end{equation}
Substituting
\eqref{eq::quadratic_Abtw}\,--\,\eqref{eq::quadratic_ABbtw}
into \eqref{eq::var_strictAdd} completes the proof.
\end{proof}

%%%%%%%%%%%%%%%
%\begin{corollary}
%\label{corollary::strictAdditive_CR}
%When the \emph{science} is strictly additive,
%the variances of $\widehat{\tau}_F$ $(F \in \mathcal{F})$ over all possible assignments of a $2^2$  completely randomized design with treatment arm sizes \eqref{eq::sizes} are
%\begin{eqnarray}
%\label{eq::var_CR_strict}
%\var_{\textsc{c-r}}(\widehat{\tau}_F)
%&=&
%\frac{\gamma_A\gamma_B }{4(N-1)}\left(
%\frac{W-1}{W}  S^2_{\textnormal{btw}} +
%\frac{M-1}{M}  S^2_{\textnormal{in}}\right) \quad (F \in \mathcal{F}).
%\end{eqnarray}
%\end{corollary}
%

\begin{proof}[Proof of Corollary \ref{corollary::strictAdditive_CR}]
Under strict additivity, we have $\mSb = \Sb\mJ_4$, $\mSi = \Si \mJ_4$, and
\begin{equation*}
\mS
\,=\, \frac{(W-1)M}{N-1} \mSb + \frac{W(M-1)}{N-1}\mSi
\,=\, \left\{\frac{(W-1)M}{N-1}\Sb  + \frac{W(M-1)}{N-1} \Si  \right\} \mJ_4\,.
\end{equation*}
Substituting this simplified expression for $\mS$  into \eqref{eq::expansion_CR} yields
\begin{equation}
\label{eq::theoreticalVar_CR_strict}
\var_{\textsc{c-r}}(\widehat{\tau}_F)
\,=\, 4^{-1}\left\{ (W-1)M\, \Sb  + W(M-1) \Si \right\} \vg_F^{\T}\mCf\vg_F\,,
\end{equation}
where, by identities \eqref{eq::interesting} and \eqref{eq::quadratic_Abtw}\,--\,\eqref{eq::quadratic_ABbtw},
\begin{align*}
\vg_F^\T\mCf\vg_F = \frac{W-1}{N-1}\vg_F^\T\mCb\vg_F +\frac{W(M-1)}{N-1} \vg_F^\T\mCi\vg_F=
\frac{\gamma_A\gamma_B}{N(N-1)}\quad (F \in \mathcal{F})
\,.
\end{align*}
This, coupled with \eqref{eq::theoreticalVar_CR_strict}, proves \eqref{eq::var_CR_strictAdd}.
It then follows from \eqref{eq::var_CR_strictAdd} and Corollary \ref{corollary::strictAdditive} that,
for $F = A$,
\begin{align*}
\var_{\textsc{s-p}}( \widehat{\tau}_A)
&-
\var_{\textsc{c-r}}\left( \widehat{\tau}_A \right) \\
=\, & \frac{\gamma_A}{W} S^2_{\textnormal{btw}}
+
\frac{\gamma_A(\gamma_B-4)}{4N}S^2_{\textnormal{in}} - \frac{\gamma_A\gamma_B }{4(N-1)}\left(
\frac{W-1}{W}  S^2_{\textnormal{btw}} +
\frac{M-1}{M}  S^2_{\textnormal{in}}\right)\\
=\, & \frac{\gamma_A}{4(N-1)W}\left\{ 4(N-1) - (W-1)\gamma_B\right\}S^2_{\textnormal{btw}}\\
& +\,  \frac{\gamma_A}{4N(N-1)} \left\{(N-1)(\gamma_B-4) -(N-W) \gamma_B \right\}S^2_{\textnormal{in}}\\
=\, &  \frac{\gamma_A}{4(N-1)W}\left\{ 4(N-W) - (W-1)(\gamma_B-4)\right\}S^2_{\textnormal{btw}}\\
& +\,  \frac{\gamma_A}{4N(N-1)} \left\{(W-1)(\gamma_B-4)  -4(N-W)\right\}S^2_{\textnormal{in}} \\
=\, &\frac{\gamma_A}{4(N-1)W} \left\{ 4(N-W)- (W-1)(\gamma_B-4)\right\}\left(  S^2_{\textnormal{btw}}- \frac{S^2_{\textnormal{in}}}{M} \right)
\end{align*}
where
\begin{align*}
& 4(N-W)- (W-1)(\gamma_B-4)\,= \,
 4W(M-1) - (W-1)(r_B+ r^{-1}_B - 2)\\
\geq\,\, &4W(M-1)- (W-1)\left(\frac{M-1}{1}+ \frac{1}{M-1}- 2\right) \\
\geq\,\, &4W(M-1)- (W-1)({M-1}) \,= \, (3W+1)(M-1) \,> \, 0\,,
\end{align*}
and, for $F = B$ and $AB$,
\begin{align*}
\var_{\textsc{s-p}}\left( \widehat{\tau}_F \right)
-
\var_{\textsc{c-r}}\left( \widehat{\tau}_F \right)
=\, &
\frac{\gamma_A\gamma_B}{4N} S^2_{\textnormal{in}}- \frac{\gamma_A\gamma_B }{4(N-1)}\left(
\frac{W-1}{W}  S^2_{\textnormal{btw}} +
\frac{M-1}{M}  S^2_{\textnormal{in}}\right)\\
=\, &
- \frac{\gamma_A\gamma_B(W-1) }{4(N-1)W}
S^2_{\textnormal{btw}}
+ \frac{\gamma_A\gamma_B(W-1)}{4N(N-1)} S^2_{\textnormal{in}}\\
=\, &- \frac{\gamma_A\gamma_B(W-1) }{4(N-1)W}\left( S^2_{\textnormal{btw}}-\frac{S^2_{\textnormal{in}}}{M}\right).
\end{align*}
This completes the proof.
\end{proof}

%%%%%%%%%%%%%%%%%%%%%%%%%%%

%%%%%%%%%%%%%%%%%%%%%%%%%%
%%%%%%%%%%%%%%%%%%%%%%%%%%%
\section{Variance Estimation}
%%%%%%%%%%%%%%%%%%%%%%
%%%%%%%%%%%%%%%%%%%%%%%
%%%%%%%%%%%%%%%%%%%%%%%%%%%%
%\begin{lemma}
%\label{lem::Es2}
%For a $2^2$ split-plot design, we have
%\begin{eqnarray*}
%\E_{\SP}(\mathbf{s^\text{$2$}_{\textnormal{btw}}}) &=&
%\begin{pmatrix}
%1&1&0& 0 \\
%1&1&0& 0 \\
%%%%%%%%%
%0& 0& 1&1\\
%0& 0& 1&1
%\end{pmatrix}
%\circ \mathbf{S^\text{$2$}_{\textnormal{btw}}}
%+
%M^{-1}
%\begin{pmatrix}
%r_B&-1&0&0\\
%-1&r_B^{-1}&0&0\\
%0&0&r_B&-1\\
%0&0&-1&r_B^{-1}
%\end{pmatrix}
%\circ \mathbf{S^\text{$2$}_{\textnormal{in}}}\,.
%\end{eqnarray*}
%%%%%%%%%%%%%%%%%%%%%%%%%%%
%%% Balanced case
%%%%%%%%%%%%%%%%%%%%%%%%%%
%%When the design is balanced, \eqref{eq::s2btw} reduces to
%%\begin{eqnarray*}
%%\E_{\textsc{S-P}}\left\{ s^2_{\textnormal{btw}}(k,l)\right\}
%%&=&
%%S^2_{\textnormal{btw}}({k,l})
%%	 + M^{-1}(-1)^{I_{\{k\neq l\}}} S^2_{\textnormal{in}}({k,l})\,.
%%\end{eqnarray*}
%\end{lemma}

%%%%%%%%%%%%%%%%%%%%%%%
% ==========================
%	PROOF_of_LEMMA_ESTIMATED-VAR	
%
%%%%%%%%%%%%%%%%%%%%%%%
\begin{lemma}
\label{eq::EZkZl}
For treatments $k$ and $l$ with the same $z = g_{A}(k)= g_{A}(l) \in \{-1,+1\}$ level of factor $A$, we have
\begin{align*}
\E_{\SP}\{ \vZ(l) \vZ(k)^\T \}
&=\,
 \frac{W_{+1}W_{-1}M_{(l)}M_{(k)}}{N(W-1)}\mPb + g_{B}(k)g_{B}(l)\frac{W_{z}M_{+1}M_{-1}}{N(M-1)}\mPi
\nonumber\\
&\phantom{\overset{\eqref{eq::covZkZl_SP}}{=} }
+ \frac{W_{z}^2 M_{(k)}M_{(l)}}{N^2} \mJn\,.\nonumber
\end{align*}
\end{lemma}

\begin{proof}[Proof of Lemma \ref{eq::EZkZl}]
With $z = g_{A}(k)= g_{A}(l) \in \{-1,+1\}$, we have $g_{A}(k)g_{A}(l)=z^2 =1$.
Substituting this into \eqref{eq::covZkZl_SP} yields
$$\cov_{\SP}\{ \vZ(l), \vZ(k)\} \,=\,
 \frac{W_{+1}W_{-1}M_{(l)}M_{(k)}}{N(W-1)}\mPb + g_{B}(k)g_{B}(l)\frac{W_{z}M_{+1}M_{-1}}{N(M-1)}\mPi\,.
$$
This, coupled with
\begin{align*}
%%%%%%%%%%%%%%%
\E_{\SP}\{ \vZ(l) \vZ(k)^\T \}
&=\,    \cov_{\SP}\left\{ \vZ(l),\, \vZ(k) \right\} + \E_{\SP}\{ \vZ(l)\}\E_{\SP}\{\vZ(k)\}^\T\,,\\
 \E_{\SP}\{ \vZ(l)\}\E_{\SP}\{\vZ(k)\}^\T &=\,   \left( \frac{W_{z}M_{(l)}}{N}\vones_N \right) \left(\frac{W_{z}M_{(k)}}{N}\vones_N^\T\right)  \,=\, \frac{W_{z}^2 M_{(k)}M_{(l)}}{N^2} \mJn\,,
\end{align*}
completes the proof.
\end{proof}

\begin{proof}[Proof of Lemma \ref{lem::Es2}]
\noindent
Define
\begin{align}
\label{eq::defm}
m_{w}(k)
&=\, M_{(k)}^{-1}\sum_{m=1}^M Y_{(wm)}(k) I_{\{T_{(wm)} = k \}}\\
&=\, M_{(k)}^{-1}I_{\{A_w = g_A(k) \}}\sum_{m=1}^M Y_{(wm)}(k) I_{\{B_{(wm)} = g_B(k) \}}\nonumber
\end{align}
such that $m_{w}(k)$ equals $Y_{(w)}^\obs(k)$ if $A_w = g_{A}(k)$, and equals $0$ if otherwise.
Let $\bm m(k) =   ( m_{1}(k), \ldots, m_{W}(k))^\T$.
The sample between-whole-plot covariances $s^2_{\textnormal{btw}} (k,l)$  satisfy
\begin{align*}
(W_z-1)s^2_{\textnormal{btw}} (k,l) &= \sum_{w: A_w = z} \{Y_{(w)}^\obs(k)-\bar{Y}^\obs(k)\}\{ Y_{(w)}^\obs(l) - \bar{Y}^\obs(l) \}\\
&= \sum_{w: A_w = z} Y_{(w)}^\obs(k) Y_{(w)}^\obs(l) - W_{z}\bar{Y}^\obs(k)\bar{Y}^\obs(l)\nonumber\\
%&= (W_z-1)^{-1} \sum_{w=1}^W m_{w}(k) m_{w}(l) - W_{z}\bar{Y}^\obs(k)\bar{Y}^\obs(l)\nonumber\\
&=\bm m(k)^\T \bm m(l)
		- W_{z}\bar{Y}^\obs(l)\bar{Y}^\obs(k)\,,\nonumber
\end{align*}
with sampling expectations
\begin{equation}
\label{eq::1111}
(W_z-1)E_{\SP}\{s^2_{\textnormal{btw}} (k,l)\}\,=\, E_{\SP}\{\bm m(k)^\T \bm m(l) \}
		- W_{z}E_{\SP}\{\bar{Y}^\obs(l)\bar{Y}^\obs(k)\}\,.
\end{equation}
We take the divide-and-conquer strategy here, and compute the two terms on the right-hand side of \eqref{eq::1111} one at a time.

To start with, let $\vY_w(k) = ( Y_{(w1)}(k), \ldots, Y_{(wM)}(k))^\T$ be the potential outcomes vectors for whole-plot $w$, and let ${\bm Z}_w(k) = ( I_{\{T_{(w1)} = k\}}, \ldots, I_{\{T_{(wM)} = k\}})^\T$ indicate the recipients of treatment $k$ therein.
That $m_{w}(k)
= M_{(k)}^{-1}  \bm  \vY_w(k)^\T {\bm Z}_w(k)$ allows us to write $\bm m(k)$ as
\begin{align}
\label{eq::vectorm}
\bm m(k) =&\,\,  ( m_{1}(k), \ldots, m_{W}(k))^\T\nonumber\\
=&\,\, M_{(k)}^{-1}\left ( \bm  \vY_1(k)^\T {\bm Z}_1(k),  \ldots,  \bm  \vY_W(k)^\T {\bm Z}_W(k)\right)
\,=\, M_{(k)}^{-1} \mY^*(k)^\T {\bm Z}(k)\,,
\end{align}
where $\mY^* (k) = \textnormal{Bdiag}\left\{ \vY_1(k), \ldots,\vY_W(k)\right\}$ is the block-diagonal matrix with $\vY_w(k)$ as its $w$th diagonal block. Thus,
\begin{align*}
\bm m(k)^\T \bm m(l) &= \,
\{M_{(k)}^{-1} \mY^*(k)^\T {\bm Z}(k)\}^\T \{M_{(l)}^{-1} \mY^*(l)^\T {\bm Z}(l)\}\\
&=\,
M_{(k)}^{-1}M_{(l)}^{-1} \vZ(k)^\T \mY^*(k) \mY^*(l)^\T \vZ(l)\nonumber\\
	%%%%%%%%%%%%%%%%
&=\,  M_{(k)}^{-1}M_{(l)}^{-1} \textnormal{tr}\{ \mY^*(l)^\T \vZ(l) \vZ(k)^\T \mY^*(k)   \}\,,\nonumber
\end{align*}
with expectation
\begin{align}
\label{eq::E1}
E_{\SP}\{\bm m(k)^\T \bm m(l) \}&=\,
M_{(k)}^{-1}M_{(l)}^{-1}E_{\SP}\left[\textnormal{tr}\{ \mY^*(l)^\T \vZ(l) \vZ(k)^\T \mY^*(k)   \}\right]\\
&=\, M_{(k)}^{-1}M_{(l)}^{-1} \textnormal{tr}\left[ \mY^*(l)^\T \E_{\SP}\left\{ \vZ(l) \vZ(k)^\T \right\} \mY^*(k) \right].\nonumber
\end{align}
Given Lemma \ref{eq::EZkZl} and the linearity of trace function, we have
\begin{align}
\label{eq::firstTerm_Es2}
& \textnormal{tr}\left[ \mY^*(l)^\T \E_{\SP}\left\{ \vZ(l) \vZ(k)^\T \right\} \mY^*(k) \right]\\
%&\quad =\, \textnormal{tr}\left[ \mY^*(l)^\T
%\left\{  \frac{W_{+1}W_{-1}M_{(l)}M_{(k)}}{N(W-1)}\mPb + g_{B}(k)g_{B}(l)\frac{W_{z}M_{+1}M_{-1}}{N(M-1)}\mPi
%+ \frac{W_{z}^2 M_{(l)}M_{(k)}}{N^2} \mJn \right\} \mY^*(k) \right]\nonumber\\
%%%%%%
&\quad\quad =\,
\frac{W_{+1}W_{-1}M_{(l)}M_{(k)}}{N(W-1)} \textnormal{tr}\{ \mY^*(l)^\T \mPb\mY^*(k)\}\nonumber\\
&\quad\quad\quad+\,
g_{B}(k)g_{B}(l)\frac{W_{z}M_{+1}M_{-1}}{N(M-1)}\textnormal{tr}\{ \mY^*(l)^\T\mPi \mY^*(k) \}\nonumber\\
&\quad\quad\quad+\,
\frac{W_{z}^2 M_{(l)}M_{(k)}}{N^2}
\textnormal{tr}\{ \mY^*(l)^\T \mJn \mY^*(k) \}\,,\nonumber
%\text{where}& \text{, as ${\bm Y}_{\text{block}}(k) = \left(Y_{(1)}(k), \ldots, Y_{(W)}(k)\right)^\T$,}\nonumber\\
\end{align}
where it follows from straightforward yet tedious matrix algebra that
\begin{align}
\textnormal{tr}\{ \mY^*(l)^\T \mPb\mY^*(k)\}
&=\, \{W^{-1}{(W-1)M}\}\bm Y_{\text{block}}(l)^\T {\bm Y}_{\text{block}}(k),\label{eq::trace1}\\
%%%%%%%
\textnormal{tr}\{ \mY^*(l)^\T\mPi \mY^*(k) \}
&=\, W(M-1) S^2_{\textnormal{in}}(k,l)\,,\label{eq::trace2}\\
%%%%%%%%%%%%%%
\textnormal{tr}\{ \mY^*(l)^\T \mJn \mY^*(k) \}
&=\,M^2  \bm Y_{\text{block}}(l)^\T {\bm Y}_{\text{block}}(k)\label{eq::trace3}
\end{align}
with ${\bm Y}_{\text{block}}(k) = \left(Y_{(1)}(k), \ldots, Y_{(W)}(k)\right)^\T$.
We defer the algebraic details for \eqref{eq::trace1}\,--\,\eqref{eq::trace3} after the main proof.
%%%%%%%%%%%%%%%%%%%%%%%%%%%%%%
Equalities \eqref{eq::trace1}\,--\,\eqref{eq::trace3} simplify \eqref{eq::firstTerm_Es2} to
\begin{align*}
&\textnormal{tr}\left[ \mY^*(l)^\T \E_{\SP}\left\{ \vZ(l) \vZ(k)^\T \right\} \mY^*(k) \right]\nonumber\\
&\quad=\,
\frac{W_{+1}W_{-1}M_{(l)}M_{(k)}}{W^2} \bm Y_{\text{block}}(l)^\T \bm Y_{\text{block}}(k)
+
g_{B}(k)g_{B}(l)\frac{W_{z}M_{+1}M_{-1}}{M}S^2_{\textnormal{in}}(k,l)\nonumber\\
&\quad\quad+\,
\frac{W_{z}^2 M_{(l)}M_{(k)}}{W^2}
\bm Y_{\text{block}}(l)^\T \bm Y_{\text{block}}(k)\nonumber\\
&\quad=\,
\frac{W_z M_{(l)}M_{(k)}}{W} \bm Y_{\text{block}}(l)^\T \bm Y_{\text{block}}(k)
+
g_{B}(k)g_{B}(l)\frac{W_{z}M_{+1}M_{-1}}{M}S^2_{\textnormal{in}}(k,l)\,. %\label{eq::firstTerm_res_Es2}\,.
\end{align*}
%%%%%%%%%%%%%%%%%%%%%
%%%%%%%%%%%%%%%%%%%%
%%%%%%%%%%%%%%%%%%%%
Substituting this back into \eqref{eq::E1} yields
\begin{equation}
\label{eq::firstTerm_res_Es2}
E_{\SP}\{\bm m(k)^\T \bm m(l) \} \,=\,
 \frac{W_{z}}{W}  \bm Y_{\text{block}}(l)^\T \bm Y_{\text{block}}(k)
+ g_{B}(k)g_{B}(l) \frac{W_{z}M_{+1}M_{-1}}{MM_{(k)}M_{(l)}}  S^2_{\textnormal{in}}(k,l)\,.
\end{equation}
This gives the first term of \eqref{eq::1111}.
% The second term
%\begin{align}
%&\left\{W_{z}(W_z-1)M_{(k)}M_{(l)}\right\}^{-1} \vY(l)^\T \E_{\SP}\left\{  \vZ(l) \vZ(k)^\T \right\}\vY(k)\nonumber\\
%&\overset{\phantom{\eqref{eq::EZkZl}}}{=}\, \frac{W_{-z}}{W(W_z-1)}S^2_{\textnormal{btw}}(k,l)
%+ g_{B}(k)g_{B}(l)\frac{M_{+1}M_{-1}}{(W_z-1)MM_{(k)}M_{(l)}} S^2_{\textnormal{in}}(k,l)
%+
%\frac{W_{z}}{W_z-1}\bar Y(k) \bar Y(l)\,.
%\end{align}
%%%%%%%%%%%%%%%%%%%%%%%%%%%%%
For the second term of \eqref{eq::1111},
it follows from
$
\bar{Y}^\obs(k) = N_{k}^{-1} \vY(k)^\T {\bm Z}(k)=  W_{(k)}^{-1}M_{(k)}^{-1}\vY(k)^\T {\bm Z}(k)$
that
\begin{align}
\label{eq::0000}
W_z E_{\SP}\{\bar{Y}^\obs(k)\bar{Y}^\obs(l) \}
&=\,  W_z E_{\SP}\{W_{z}^{-2}M_{(k)}^{-1}M_{(l)}^{-1}\vY(l)^\T \vZ(l) \vZ(k)^\T \vY(k)\}\\
&=\, W_{z}^{-1}M_{(k)}^{-1}M_{(l)}^{-1}\vY(l)^\T E_{\SP}\{ \vZ(l) \vZ(k)^\T \}\vY(k)\,,\nonumber
\end{align}
%%%%%%%%%%%%%%%%
in which, by Lemma \ref{eq::EZkZl},
\begin{align*}
&\vY(l)^\T \E_{\SP}\left\{  \vZ(l) \vZ(k)^\T \right\}\vY(k) \nonumber\\
&\quad{=}
%\vY(l)^\T \left\{   \frac{W_{+1}W_{-1}M_{(k)}M_{(l)}}{N(W-1)}\mPb + g_{B}(k)g_{B}(l)\frac{W_{z}M_{+1}M_{-1}}{N(M-1)}\mPi
%+ \frac{W_{z}^2 M_{(k)}M_{(l)}}{N^2} \mJn \right\}\vY(k)\nonumber\\
%%%%%%%%%%%%%%%%%%%%%%%%%%%%%%%%
%\overset{\phantom{\eqref{eq::EZkZl}}}{=}\, &
\frac{W_{+1}W_{-1}M_{(k)}M_{(l)}}{N(W-1)}\vY(l)^\T  \mPb \vY(k)
+
g_{B}(k)g_{B}(l)\frac{W_{z}M_{+1}M_{-1}}{N(M-1)}\vY(l)^\T \mPi\vY(k)\nonumber \\
&\quad\quad+\,
 \frac{W_{z}^2 M_{(k)}M_{(l)}}{N^2}\vY(l)^\T \mJn \vY(k)\nonumber\\
%%%%%%%%%%%%%%%%%%%%%%%%%%%%%%%
&\quad = \frac{W_{z}W_{-z}M_{(k)}M_{(l)}}{W} S^2_{\textnormal{btw}}(k,l)
+
g_{B}(k)g_{B}(l)\frac{W_{z}M_{+1}M_{-1}}{M} S^2_{\textnormal{in}}(k,l)\nonumber\\
&\quad\quad+\,
W_{z}^2 M_{(k)}M_{(l)}\bar Y(k) \bar Y(l)\,.
\end{align*}
Substituting this last expression into the right-hand side of \eqref{eq::0000} equates $W_z E_{\SP}\{\bar{Y}^\obs(k)\bar{Y}^\obs(l) \}$ to
%%%%%%%%%%%%%%%%%%%%%%%%
%%%%%%%%%%%%%%%%%%%%%%%%%
\begin{align}
\frac{W_{-z}}{W}S^2_{\textnormal{btw}}(k,l)
+ \frac{g_{B}(k)g_{B}(l)M_{+1}M_{-1}}{MM_{(k)}M_{(l)}} S^2_{\textnormal{in}}(k,l)
+
W_{z}\bar Y(k) \bar Y(l).\label{eq::secondTerm_Es2}
\end{align}
%%%%%%%%%%%%%%%%%%%
%%%%%%%%%%%%%%%%%%%%
%as it follows from \eqref{eq::EZkZl},
%the expression of $E \{s^2_{\textnormal{btw}}(k,l)\}$ %when $g_{A}(k) = g_{A}(l)$ as
Substituting \eqref{eq::firstTerm_res_Es2} and \eqref{eq::secondTerm_Es2} into \eqref{eq::1111} yields
\begin{align*}
&(W_z-1)\E_{\SP}\{ s^2_{\textnormal{btw}}(k,l)\} \\
	&\quad=\,  \left\{\frac{W_{z}}{W}  \bm Y_{\text{block}}(l)^\T \bm Y_{\text{block}}(k)
+  \frac{g_{B}(k)g_{B}(l)W_{z}M_{+1}M_{-1}}{MM_{(k)}M_{(l)}}  S^2_{\textnormal{in}}(k,l)\right\}\\
	&\quad\quad-\, \left\{\frac{W_{-z}}{W}S^2_{\textnormal{btw}}(k,l)
+ \frac{g_{B}(k)g_{B}(l)M_{+1}M_{-1}}{MM_{(k)}M_{(l)}} S^2_{\textnormal{in}}(k,l)
+
W_{z}\bar Y(k) \bar Y(l)\right\}\\	%%%%%%%%%%%%%%%%%%%%%%%%%%%%%555
	&\quad=\,  \frac{W_{z}}{W} \left\{ \bm Y_{\text{block}}(l)^\T \bm Y_{\text{block}}(k)	- W\bar{Y}(k) \bar{Y}(l)\right\} -  \frac{W_{-z}}{W}S^2_{\textnormal{btw}}(k,l)\nonumber\\
&\quad\quad+\, g_{B}(k)g_{B}(l) \frac{(W_z-1)M_{+1}M_{-1}}{MM_{(k)}M_{(l)}}  S^2_{\textnormal{in}}(k,l)\\	%%%%%%%%%%%%%%%%%%%%%%%%%%%%%555
	&\quad=\,   \frac{W_{z}}{W}\left\{ (W-1)S^2_{\textnormal{btw}}(k,l)\right\}
-  \frac{W - W_{z}}{W}S^2_{\textnormal{btw}}(k,l)
+ g_{B}(k)g_{B}(l) \frac{(W_z-1)M_{+1}M_{-1}}{MM_{(k)}M_{(l)}}  S^2_{\textnormal{in}}(k,l)\\	%%%%%%%%%%%%%%%%%%%%%%%%%%%%
&\quad=\,  (W_z-1)S^2_{\textnormal{btw}}(k,l)
	 + (W_z-1)g_{B}(k)g_{B}(l)\frac{M_{+1}M_{-1}}{MM_{(k)}M_{(l)}}S^2_{\textnormal{in}}(k,l)\,,
%& \E_{\SP}\{ s^2_{\textnormal{btw}}(k,l)\}
%= S^2_{\textnormal{btw}}(k,l)
%	 + g_{B}(k)g_{B}(l)\frac{M_{+1}M_{-1}}{MM_{(k)}M_{(l)}}S^2_{\textnormal{in}}(k,l)
\end{align*}
with
\begin{eqnarray*}
g_{B}(k)g_{B}(l)\frac{M_{+1}M_{-1}}{MM_{(k)}M_{(l)}}
\,=\,
\left\{
\begin{array}{ll}
M^{-1}r_B & \text{if $(k,l)=(1,1), (3,3)$,}\\
M^{-1}r^{-1}_B & \text{if $(k, l) = (2,2), (4,4)$,}\\
- M^{-1} & \text{if $(k,l) = (1,2), (2,1), (3,4), (4,3)$.}\\
\end{array}
\right.
\end{eqnarray*}
Dividing both sides by $(W_z-1)$ completes the proof.
\end{proof}

\noindent
We give the algebraic details for \eqref{eq::trace1}\,--\,\eqref{eq::trace3} below.
\begin{proof}
Equality \eqref{eq::trace1} follows from
%%%%%%%%%%%%%%%%%%%%%%%%%%%%%
\begin{align*}
&\textnormal{tr}\{ \mY^*(l)^\T \mPb\mY^*(k)\}
\,=\,
 \textnormal{tr}\left[ \mY^*(l)^\T \left\{ \mPw \otimes \left(M^{-1}\mJ_M \right) \right\} \mY^*(k) \right]\\
=\,&
 \textnormal{tr}\left[ \mY^*(l)^\T \left\{ \left(\mI_W - W^{-1}\mJ_W\right) \otimes \left(M^{-1}\mJ_M \right) \right\} \mY^*(k) \right]\nonumber\\
=\,&
M^{-1} \textnormal{tr}\left\{\mY^*(l)^\T \left( \mI_W\otimes\mJ_M  \right) \mY^*(k) \right\}
- N^{-1} \textnormal{tr}\left\{ \mY^*(l)^\T \left(\mJ_W \otimes \mJ_M \right) \mY^*(k) \right\}\nonumber\\
=\,&M^{-1}\textnormal{tr}\left\{ \begin{pmatrix} \vY_1(l)^\T&\ldots&\vzeros\\ \vdots&\ddots&\vdots  \\ \vzeros&\ldots& \vY_W(l)^\T\end{pmatrix}
\begin{pmatrix} \mJs& \ldots&\mathbf{0}\\ \vdots&\ddots&\vdots  \\\mathbf{0}&\ldots& \mJs\end{pmatrix}
 \begin{pmatrix} \vY_1(k)&\ldots&\vzeros\\ \vdots&\ddots&\vdots  \\ \vzeros&\ldots& \vY_W(k)\end{pmatrix} \right\}\nonumber\\
&\,\,-\, N^{-1}\textnormal{tr}\left\{ \begin{pmatrix} \vY_1(l)^\T&\ldots&\vzeros\\ \vdots&\ddots&\vdots  \\ \vzeros&\ldots& \vY_W(l)^\T\end{pmatrix}
\begin{pmatrix} \mJs& \ldots&\mJs\\ \vdots&\ddots&\vdots  \\\mJs&\ldots& \mJs\end{pmatrix}
 \begin{pmatrix} \vY_1(k)&\ldots&\vzeros\\ \vdots&\ddots&\vdots  \\ \vzeros&\ldots& \vY_W(k)\end{pmatrix} \right\}\nonumber\\
=\,& M^{-1} \sum_{w=1}^W \vY_w(l)^\T\mJs\vY_w(k) -  N^{-1} \sum_{w=1}^W \vY_w(l)^\T\mJs\vY_w(k)\nonumber\\
=\,& \frac{W-1}{N} \sum_{w=1}^W \vY_w(l)^\T\mJs\vY_w(k)
\,=\, \frac{W-1}{N}  \sum_{w=1}^W\{ \vY_w(l)^\T\vones_M \}\{\vones_M^\T \vY_w(k)\} \nonumber\\
=\,& \frac{W-1}{N} \sum_{w=1}^W \{Y_{(w)}(l)M  \}\{{Y}_{(w)}(k)M \}
\,=\,  \frac{(W-1)M}{W}\bm Y_{\text{block}}(l)^\T {\bm Y}_{\text{block}}(k)\nonumber\,.
\end{align*}
%%%%%%%%%%%%%%%%%%%%%%
Equality \eqref{eq::trace2} follows from
%%%%%%%%%%%%%%%%%%%%%%%%
\begin{align*}
&\textnormal{tr}\{ \mY^*(l)^\T\mPi \mY^*(k) \}
\,=\,\textnormal{tr}[ \mY^*(l)^\T ( \mIw \otimes \mPs )  \mY^*(k) ]\\
=\,&\textnormal{tr}\left\{ \begin{pmatrix} \vY_1(l)^\T&\ldots&\vzeros\\ \vdots&\ddots&\vdots  \\ \vzeros&\ldots& \vY_W(l)^\T\end{pmatrix}
\begin{pmatrix} \mPs& \ldots&\mathbf{0}\\ \vdots&\ddots&\vdots  \\\mathbf{0}&\ldots& \mPs\end{pmatrix}
 \begin{pmatrix} \vY_1(k)&\ldots&\vzeros\\ \vdots&\ddots&\vdots  \\ \vzeros&\ldots& \vY_W(k)\end{pmatrix} \right\}\nonumber\\
=\,&\sum_{w=1}^W \vY_w(l)^\T \mPs \vY_w(k)
	= W(M-1) S^2_{\textnormal{in}}(k,l)\nonumber\,.
\end{align*}
%%%%%%%%%%%%%%%%%%%
Equality \eqref{eq::trace3} follows from
\begin{align*}
&\textnormal{tr}\left\{ \mY^*(l)^\T \mJn \mY^*(k) \right\}\\
=\,& \textnormal{tr}\left\{ \begin{pmatrix} \vY_1(l)^\T&\ldots&\vzeros\\ \vdots&\ddots&\vdots  \\ \vzeros&\ldots& \vY_W(l)^\T\end{pmatrix}
\begin{pmatrix} \mJs& \ldots&\mJs\\ \vdots&\ddots&\vdots  \\\mJs&\ldots& \mJs\end{pmatrix}
 \begin{pmatrix} \vY_1(k)&\ldots&\vzeros\\ \vdots&\ddots&\vdots  \\ \vzeros&\ldots& \vY_W(k)\end{pmatrix} \right\}\nonumber \\
=\,&M^2  \bm Y_{\text{block}}(l)^\T {\bm Y}_{\text{block}}(k)\nonumber\,.
\end{align*}
\end{proof}
%%%%%%%%%%%%%%%%%%%
%
%	A statement in the main text
%
%%%%%%%%%%%%%%%%%%%%%%%%%%%%%%%
\begin{lemma}
\label{corollary::EVF}
Under the $2^2$ split-plot design qualified by Definition \ref{def::SP}, the sampling expectation of $\widehat{V}_F$ equals
\begin{equation*}
%\label{eq::whyVhat}
\E_{\SP}(\widehat{V}_F)
\,=\,
4^{-1} \vg_F^{\T} \left\{
 \begin{pmatrix}
W_{-1}^{-1}\mJ_2  & \mathbf{0}\\
 \mathbf{0} & W_{+1}^{-1}\mJ_2
\end{pmatrix}
\circ \mathbf{S^\text{$2$}_{\textnormal{btw}}}
+
W(M-1) \mCi
\circ \mathbf{S^\text{$2$}_{\textnormal{in}}}
\right\}\vg_F\,. %\quad (F \in \mathcal{F})
\end{equation*}
%The coefficient matrix for $ \mathbf{S^\text{$2$}_{\textnormal{in}}}$, namely $W(M-1) \mCi$,
%matching exactly that in \eqref{eq::theoreticalVar}.
\end{lemma}

\begin{proof}[Proof of Lemma \ref{corollary::EVF}]
Rewrite $\mCi$ as
\begin{align*}
\mCi
&=\,
\frac{1}{NW(M-1)} \begin{pmatrix}
(1+r_A)r_B & - (1+r_A) & 0&0\\
- (1+r_A) &(1+r_A)r^{-1}_B &0&0\\
0& 0& (1+r^{-1}_A)r_B & - (1+r^{-1}_A)\\
0& 0& - (1+r^{-1}_A) &(1+r^{-1}_A)r^{-1}_B \end{pmatrix}\\
&=\, \frac{1}{NW(M-1)} \begin{pmatrix}
1+r_A &1+r_A& 0 & 0\\
1+r_A&1+r_A &0&0\\
0& 0&1+r^{-1}_A &1+r^{-1}_A\\
0& 0& 1+r^{-1}_A &1+r^{-1}_A \end{pmatrix}
\circ
\begin{pmatrix}
r_B & - 1 & 0&0\\
- 1 &r^{-1}_B &0&0\\
0& 0&r_B & - 1\\
0& 0& - 1 &r^{-1}_B \end{pmatrix}\\
&=\,
\frac{1}{N(M-1)}  \begin{pmatrix}
W_{-1}^{-1}\mJ_2  & \mathbf{0}\\
 \mathbf{0} & W_{+1}^{-1}\mJ_2
\end{pmatrix}
\circ
\begin{pmatrix}
r_B & - 1 & 0&0\\
- 1 &r^{-1}_B &0&0\\
0& 0&r_B & - 1\\
0& 0& - 1 &r^{-1}_B \end{pmatrix}.
\end{align*}
The result follows from identity
\begin{align*}
&
 \begin{pmatrix}
W_{-1}^{-1}\mJ_2  & \mathbf{0}\\
 \mathbf{0} & W_{+1}^{-1}\mJ_2
\end{pmatrix}
\circ
\E_{\SP}(\mathbf{s^\text{$2$}_{\textnormal{btw}}})\\
=\,&
 \begin{pmatrix}
W_{-1}^{-1}\mJ_2  & \mathbf{0}\\
 \mathbf{0} & W_{+1}^{-1}\mJ_2
\end{pmatrix}
\circ\left\{
 \begin{pmatrix}
\mJ_2  & \mathbf{0}\\
 \mathbf{0} & \mJ_2
\end{pmatrix}
\circ \mathbf{S^\text{$2$}_{\textnormal{btw}}}+
M^{-1}
\begin{pmatrix}
r_B&-1&0&0\\
-1&r_B^{-1}&0&0\\
0&0&r_B&-1\\
0&0&-1&r_B^{-1}
\end{pmatrix}
\circ \mathbf{S^\text{$2$}_{\textnormal{in}}}
\right\}\\
%%%%%%%%%%%%%%%%%%%%%%%%%%%
%%%%%%%%%%%%%%%%%%%%%%%%%%%
=\,&
 \begin{pmatrix}
W_{-1}^{-1}\mJ_2  & \mathbf{0}\\
 \mathbf{0} & W_{+1}^{-1}\mJ_2
\end{pmatrix}
\circ \mathbf{S^\text{$2$}_{\textnormal{btw}}}
+
W(M-1) \mCi
\circ \mathbf{S^\text{$2$}_{\textnormal{in}}}\,.
\end{align*}
\end{proof}
%%%%%%%%%%%%%%%%%%%%%%%%%%%
% ===================================
% PROOF OF Estimated variance being conservative
% ===================================
%%%%%%%%%%%%%%%%%%%%%%%%%%%%%%%%
%\begin{theorem}
%\label{thm::estimated-variances}
%The variance estimators $\widehat{V}_F$ defined in \eqref{eq::VhatF} are \emph{on-average conservative estimates} of $\var_{\SP}\left( \widehat{\tau}_F \right)$
%in the sense that
%\begin{eqnarray*}
%\var_{\SP}(\widehat{\tau}_F) - \E_{\SP}( \widehat{V}_F)
%\,=\,
%\frac{-S^2_{F\text{-btw}}}{4W}
%\,\,\leq\,\,0\quad (F \in \mathcal{F})\,,
%\end{eqnarray*}
%and the equality holds if and only if $S^2_{F\text{-btw}} = 0$.
%\end{theorem}

%%%%%%%%%%%%%%%%%%%%%%%%%%%%%%%%%%
\begin{proof}[Proof of Theorem \ref{thm::estimated-variances}]
%%%%%%%%%%%%%%%%%%%
It follows from Theorem \ref{thm::theoretical-variances} and Lemma \ref{corollary::EVF} that
\begin{align*}
&\var_{\textsc{S-P}}\left( \widehat{\tau}_F \right) - \E_{\SP}( \widehat{V}_F) \\
=\,&  4^{-1}  (W-1)M\,\vg_F^{\T}(\mCb \circ \mSb)\vg_F +4^{-1} W(M-1)\vg_F^\T(\mCi \circ \mSi)\vg_F\\
&-\, 4^{-1} \vg_F^{\T} \left\{
 \begin{pmatrix}
W_{-1}^{-1}\mJ_2  & \mathbf{0}\\
 \mathbf{0} & W_{+1}^{-1}\mJ_2
\end{pmatrix}
\circ \mathbf{S^\text{$2$}_{\textnormal{btw}}}\right\}\vg_F
- 4^{-1}W(M-1) \vg_F^\T ( \mCi
\circ \mathbf{S^\text{$2$}_{\textnormal{in}}})
\vg_F\\
%%%%%%%%%%%%
=\,&  4^{-1} \vg_F^{\T} \left[\left \{ (W-1)M \,\mCb - \begin{pmatrix}
W_{-1}^{-1}\mJ_2  & \mathbf{0}\\
 \mathbf{0} & W_{+1}^{-1}\mJ_2
\end{pmatrix}\right\} \circ \mSb \right]\vg_F\\
%%%%%%%%%%%%%%%
=\,&  4^{-1} \vg_F^{\T} \left[\left\{W^{-1}\begin{pmatrix}
r_A \mJ_2 &-\mJ_2\\
-\mJ_2 & r^{-1}_A\mJ_2
\end{pmatrix}
- W^{-1}\begin{pmatrix}
(1+r_A) \mJ_2  & \mathbf{0}\\
 \mathbf{0} & (1+r_A^{-1})\mJ_2
\end{pmatrix}\right\} \circ \mSb \right] \vg_F\\
%%%%%%%%%%%%
=\,&   4^{-1}\vg_F^{\T}\left\{ (-W^{-1}\mJ_4 )\circ \mSb\right\} \vg_F
 \,=\,
-(4W)^{-1}\vg_F^{\T}(\mJ_4 \circ \mSb) \vg_F\\
=\,&
-(4W)^{-1}\vg_F^{\T}\mSb \vg_F
 \,=\,  -(4W)^{-1}S^2_{F\text{-btw}} \,.
\end{align*}
This completes the proof.
\end{proof}

%%%%%%%%%%%%%%%%%%%%%%%%%%%
% ===================================
% PROOF OF COVARIANCE
% ===================================
%%%%%%%%%%%%%%%%%%%%%%%%%
%
%
\section{Covariance structure of residuals from the derived linear model}
Recall from \eqref{eq::transition} that $g_A(T_{(wm)}) = A_w$, $g_B(T_{(wm)}) = B_{(wm)}$, and  thus $g_{AB}(T_{(wm)}) = A_wB_{(wm)}$ for all $(wm)$.
This allows us to rewrite formula \eqref{eq::PO_resid} of the main text as
\begin{equation}
\label{epsilon.appendix}
\epsilon_{(wm)} =
\delta_{(wm)\text{-}\mu} +  2^{-1}\delta_{(wm)\text{-}A}A_w
+2^{-1}\delta_{(wm)\text{-}B}B_{(wm)}
+2^{-1}\delta_{(wm)\text{-}{AB}}A_wB_{(wm)}\,.
\end{equation}

%%%%%%%%%%%%%%%%%%%%%%%%%%%
\begin{proof}[Proof of Theorem \ref{thm::cov}]
%%%%%%%%%%%%%%%%%%%%%%%%%%
Let $\mathcal{A} = \{A_{w}\}_{w=1}^W$.
The law of iterated expectations allows us to write the covariance of $\epsilon_{(wm)}$ and $\epsilon_{(w^\prime m^\prime)}$ as
\begin{align}
\label{eq::ECCE_thmcov}
\cov_{\SP}(\epsilon_{(wm)},\, \epsilon_{(w^\prime m^\prime)})
=&\,\, \cov_{\SP}\left\{
\E_{\SP}( \epsilon_{(wm)}\mid \mathcal{A}),
\E_{\SP}( \epsilon_{(w^\prime m^\prime)} \mid \mathcal{A}) \right\}\\
&+
\E_{\SP}\left\{\cov_{\SP}( \epsilon_{(wm)}, \epsilon_{(w^\prime m^\prime)} \mid \mathcal{A}) \right\}. \nonumber
\end{align}
Refer to the first term on the right-hand side of \eqref{eq::ECCE_thmcov} as the \emph{covariance of expectations}, and the second the \emph{expectation of covariance}.

Let $e_B = E_{\SP}(B_{(wm)}) = (r_B-1)/(r_B+1)$ be the common expectation of the identically distributed $\{B_{(wm)}\}$.
With $\epsilon_{(wm)}$ given by \eqref{epsilon.appendix},
it follows from the joint independence of $B_{(wm)}$ and $\mathcal{A}$ that
\begin{align*}
&E_{\SP}(\epsilon_{(wm)} \mid \mathcal{A} )
\,=\,
\delta_{(wm)\text{-}\mu}+  2^{-1} e_B\delta_{(wm)\text{-}B} +   2^{-1}\left( \delta_{(wm)\text{-}A}+e_B\delta_{(wm)\text{-}AB}\right) A_w\,.
\end{align*}
This expression for $\E_{\SP}(  \epsilon_{(wm)} \mid \mathcal{A})$ allows us to compute the {covariance of expectations} term in \eqref{eq::ECCE_thmcov} as
\begin{align}
&\cov_{\SP}\left\{\E_{\SP}(  \epsilon_{(wm)} \mid \mathcal{A}), \E_{\SP}( \epsilon_{(w^\prime m^\prime)} \mid \mathcal{A}) \right\}\label{eq::CE_thmcov}\\
&=\, \cov_{\SP}\left\{2^{-1}\left( \delta_{(wm)\text{-}A}+e_B\delta_{(wm)\text{-}AB}\right) A_w, \,
2^{-1}\left( \delta_{(w^\prime m^\prime)\text{-}A}+e_B\delta_{(w^\prime m^\prime)\text{-}AB}\right)A_{w^\prime}\right\}\nonumber\\
&=\,  4^{-1}
\left( \delta_{(wm)\text{-}A}+e_B\delta_{(wm)\text{-}AB}\right)
\left( \delta_{(w^\prime m^\prime)\text{-}A}+e_B\delta_{(w^\prime m^\prime)\text{-}AB}\right)
\cov_{\SP}(A_w, A_{w^\prime}) \nonumber\\
&=\,  4^{-1}
(\delta_{(wm)\text{-}A}, \, \delta_{(wm)\text{-}AB})
%\begin{pmatrix}1&e_B\\e_B& e_B^2 \end{pmatrix}
\begin{pmatrix} 1\\ e_B \end{pmatrix}
( 1, e_B)
\left(\begin{array}{l}
\delta_{(w^\prime m^\prime)\text{-}A} \\ \delta_{(w^\prime m^\prime)\text{-}AB}
\end{array}\right)
\cov_{\SP}(A_w,A_{w^\prime})\,, \nonumber
\end{align}
where
\begin{equation*}
\cov_{\SP}\left(A_w,A_{w^\prime}\right) \,=\, \var_{\SP}(A_w)
\,=\, \frac{4W_{+1}W_{-1}}{W^2} = \frac{4r_A}{(1+r_A)^2}\,,
\end{equation*}
if {$w = w^\prime$},
and
\begin{equation*}
\cov_{\SP}(A_w,A_{w^\prime})
\,=\, - \frac{4W_{+1}W_{-1}}{W^2(W-1)} =-\frac{4r_A}{(1+r_A)^2(W-1)}\,,
\end{equation*}
if {$w\neq w^\prime$} by Lemma \ref{lem::I}.
%%%%%%%%%%%%%%%%%%%%%%%a
Similarly, by \eqref{epsilon.appendix} and the joint independence of $B_{(wm)}$ and $\mathcal{A}$,
\begin{align*}
%%%%%%%%%%%%%%%%%%%%%%%%%%%%%%
&\cov_{\SP}( \epsilon_{(wm)},\, \epsilon_{(w^\prime m^\prime)} \mid \mathcal{A})%\label{eq::plugin2}
\\
\quad=\,\,&\cov_{\SP}\left\{ 2^{-1}( \delta_{(wm)\text{-}B}+\delta_{(wm)\text{-}AB}A_w)  B_{(wm)}, 2^{-1}(\delta_{(w^\prime m^\prime)\text{-}B} +  \delta_{(w^\prime m^\prime)\text{-}AB}A_{w^\prime})B_{(w^\prime m^\prime)}\right\}\nonumber\\
\quad=\,\,&4^{-1}( \delta_{(wm)\text{-}B}+\delta_{(wm)\text{-}AB}A_w) (\delta_{(w^\prime m^\prime)\text{-}B} +  \delta_{(w^\prime m^\prime)\text{-}AB}A_{w^\prime}) \cdot
\cov_{\SP}(  B_{(wm)}, B_{(w^\prime m^\prime)})\nonumber\\
\quad=\,\,&4^{-1}( \delta_{(wm)\text{-}B}, \delta_{(wm)\text{-}AB})
\begin{pmatrix} 1\\ A_w \end{pmatrix}
( 1, A_{w^\prime})
%\begin{pmatrix} 1& A_w\\A_w &1 \end{pmatrix}
\left(\begin{array}{l}\delta_{(w^\prime m^\prime)\text{-}B}\\ \delta_{(w^\prime m^\prime)\text{-}AB}\end{array}\right)\cov_{\SP}(  B_{(wm)}, B_{(w^\prime m^\prime)})\nonumber\,.
%&=\,
%4^{-1} \left( \delta_{(wm)\text{-}B}, \,\delta_{(wm)\text{-}AB}\right) \begin{pmatrix} 1\\ A_w \end{pmatrix}
%\left( 1,\, A_{w^\prime} \right)\left(\begin{array}{l}\delta_{(w^\prime m^\prime)\text{-}B}\\ \delta_{(w^\prime m^\prime)\text{-}AB}\end{array}\right)\cov_{\SP}\left( B_{(wm)}, B_{(w^\prime m^\prime)}\right)\\
%&=\,
%4^{-1} \left( \delta_{(wm)\text{-}B}, \,\delta_{(wm)\text{-}AB}\right) \begin{pmatrix} 1& A_{w^\prime}\\ A_w& A_wA_{w^\prime} \end{pmatrix}\left(\begin{array}{l}\delta_{(w^\prime m^\prime)\text{-}B}\\ \delta_{(w^\prime m^\prime)\text{-}AB}\end{array}\right)\cov_{\SP}\left( B_{(wm)}, B_{(w^\prime m^\prime)}\right) \\
%&=\,
%\left\{ \begin{array}{ll}
%-\left\{4(M-1)\right\}^{-1}v_B ( \delta_{(wm)\text{-}B}+\delta_{(wm)\text{-}AB}A_w) (\delta_{(w m^\prime)\text{-}B} +  \delta_{(w m^\prime)\text{-}AB}A_w), &w = w^\prime, m\neq m^\prime\,,\\
%0\,, &w\neq w^\prime\,.
%\end{array}\right.
\end{align*}
This expression for $\cov_{\SP}( \epsilon_{(wm)}, \epsilon_{(w^\prime m^\prime)} \mid \mathcal{A})$ allows us to compute the expectation of covariance term in \eqref{eq::ECCE_thmcov} as
\begin{align}
&\E_{\SP}\{\cov_{\SP}( \epsilon_{(wm)}, \epsilon_{(w^\prime m^\prime)} \mid \mathcal{A})\}\label{eq::EC_thmcov}\\
%%%%%%%%%%
&=\,
4^{-1}( \delta_{(wm)\text{-}B},\delta_{(wm)\text{-}AB})
\E_{\SP}\begin{pmatrix} 1& A_w \\ A_{w^\prime} & A_wA_{w^\prime} \end{pmatrix}
\left(\begin{array}{l}\delta_{(w^\prime m^\prime)\text{-}B}\\ \delta_{(w^\prime m^\prime)\text{-}AB}\end{array}\right) \cov_{\SP}(  B_{(wm)}, B_{(w^\prime m^\prime)})\nonumber \\
%%%%%%%%%%%%%%%%%
&=\,
4^{-1}( \delta_{(wm)\text{-}B},\delta_{(wm)\text{-}AB})
\begin{pmatrix} 1& e_A\\ e_A& \E_{\SP}(A_wA_{w^\prime}) \end{pmatrix}
\left(\begin{array}{l}\delta_{(w^\prime m^\prime)\text{-}B}\\ \delta_{(w^\prime m^\prime)\text{-}AB}\end{array}\right) \cov_{\SP}(  B_{(wm)}, B_{(w^\prime m^\prime)})\,,\nonumber
%&=\,
%\left\{ \begin{array}{ll}
%-\left\{4(M-1)\right\}^{-1}v_B ( \delta_{(wm)\text{-}B},\delta_{(wm)\text{-}AB})
%\begin{pmatrix} 1& e_A\\ e_A& 1 \end{pmatrix}
%\left(\begin{array}{l}\delta_{(w^\prime m^\prime)\text{-}B}\\ \delta_{(w^\prime m^\prime)\text{-}AB}\end{array}\right), &w = w^\prime, m\neq m^\prime\,,\\
%0\,, &w\neq w^\prime\,.
%\end{array}\right.
\end{align}
%%%%%%%%%%%%%%%%%%%%%%%%%%%
%%%%%%%%%%%%%%%%%%%%%%%%%%%%
where $e_A = E_{\SP}(A_w)= (r_A-1)/(r_A+1)$ is the common expectation of the identically distributed $\{A_{w}\}$, $\cov_{\SP}( B_{(wm)},B_{(w^\prime m^\prime)})=0
$ if {$w \neq w^\prime$} by Definition \ref{def::SP}, and
\begin{align*}
\cov_{\SP}(  B_{(wm)}, B_{(w^\prime m^\prime)})
&=\, \cov_{\SP}( B_{(wm)},B_{(wm^\prime)})\\
&=\,- \frac{4M_{+1}M_{-1}}{M^2(M-1)} \,=\,
-\frac{4 r_B}{(1+r_B)^2 (M-1)}
\end{align*}
if {$w = w^\prime$, $m \neq m^\prime$} by Lemma \ref{lem::I}.
%%%%%%%%%%%%%%%%%
%%%%%%%%%%%%%%%%%
Given \eqref{eq::CE_thmcov}, \eqref{eq::EC_thmcov}, and the covariances of the treatment indicators,  the decomposition \eqref{eq::ECCE_thmcov} simplifies to
\begin{align}
&\cov_{\SP}(\epsilon_{(wm)}, \epsilon_{(w^\prime m^\prime)})\label{eq::almost1}\\
&\quad =\, 4^{-1}
\left(\delta_{(wm)\text{-}A}, \delta_{(wm)\text{-}AB}\right) \begin{pmatrix}1&e_B\\e_B& e_B^2 \end{pmatrix}
\left(\begin{array}{l}
\delta_{(w^\prime m^\prime)\text{-}A} \\ \delta_{(w^\prime m^\prime)\text{-}AB}
\end{array}\right)
\cov_{\SP}(A_w,A_{w^\prime})\nonumber\\
&\quad\quad +
4^{-1}( \delta_{(wm)\text{-}B},\delta_{(wm)\text{-}AB})
\begin{pmatrix} 1& e_A\\ e_A& E_\SP(A_w^2) \end{pmatrix}
\left(\begin{array}{l}\delta_{(w^\prime m^\prime)\text{-}B}\\ \delta_{(w^\prime m^\prime)\text{-}AB}\end{array}\right) \cov_{\SP}(  B_{(wm)}, B_{(w^\prime m^\prime)})\nonumber\\
%%%%%%%%%%%%%%%%%%%%%%%%
&\quad =\,\frac{r_A}{(r_A+1)^{2}} \left(\delta_{(wm)\text{-}A},\delta_{(wm)\text{-}AB}\right)
\begin{pmatrix}1&e_B\\e_B&e^2_B \end{pmatrix}
\left(\begin{array}{l}\delta_{(w^\prime m^\prime)\text{-}A}\\ \delta_{(w^\prime m^\prime)\text{-}AB}\end{array}\right)\nonumber \\
&\quad\quad - (M-1)^{-1}\frac{r_B}{(r_B+1)^{2}} \left( \delta_{(wm)\text{-}B},\delta_{(wm)\text{-}AB}\right)
\begin{pmatrix}1&e_A\\e_A&1 \end{pmatrix}
\left(\begin{array}{l}\delta_{(w^\prime m^\prime)\text{-}B}\\ \delta_{(w^\prime m^\prime)\text{-}AB} \end{array}\right)\nonumber
\end{align}
if {$w = w^\prime$, $m \neq m^\prime$}, and
%%%%%%%%%%%%%%%%%%%%%%%%
%%%%%%%%%%%%%%%%%%%%%%%%
\begin{align}
&\cov_{\SP}(\epsilon_{(wm)}, \epsilon_{(w^\prime m^\prime)})\label{eq::almost2}\\
&\quad =\,
- (W-1)^{-1}\frac{r_A}{(r_A+1)^{2}}\left(\delta_{(wm)\text{-}A},\delta_{(wm)\text{-}AB}\right)
\begin{pmatrix}1&e_B\\e_B&e^2_B \end{pmatrix}
\left(\begin{array}{l}\delta_{(w^\prime m^\prime)\text{-}A}\\  \delta_{(w^\prime m^\prime)\text{-}AB}\end{array}\right). \nonumber
\end{align}
if {$w \neq w^\prime$}.
%%%%%%%%%%%%%%%%%%%%%%%%%%
%%%%%%%%%%%%%%%%%%%%%%%%%%%
Letting $W$ and $M$ approach infinity in \eqref{eq::almost1} and \eqref{eq::almost2} proves the result.
%\begin{align*}
%& \\
%&\to\, \left\{
%\begin{array}{ll}
%r_A(r_A+1)^{-2}\left(\delta_{(wm)\text{-}A},\delta_{(wm)\text{-}AB}\right)
%\begin{pmatrix}1&e_B\\e_B&e^2_B \end{pmatrix}
%\left(\begin{array}{l}\delta_{(w^\prime m^\prime)\text{-}A}\\ \delta_{(w^\prime m^\prime)\text{-}AB}\end{array}\right),
%&\text{if $w=w^\prime$, $m\neq m^\prime$,}\\
%0\,, &\text{if $w\neq w^\prime$.}
%\end{array}
%\right.
%\end{align*}
\end{proof}

\end{document}